\documentclass[
  twocolumn,
  amsmath,
  amssymb,
  nofootinbib,
  prx,
  aps,
  10pt,
  floatfix
]{revtex4-2}

\usepackage[utf8]{inputenc}

\usepackage{amsmath}
\usepackage{amssymb}
\usepackage{amsfonts}
\usepackage{amsthm}
\usepackage{bm}

\usepackage[braket]{qcircuit}
\usepackage{graphicx}
\usepackage[export]{adjustbox}
\usepackage{float}
\usepackage{dcolumn}
\usepackage{multirow}
\usepackage{placeins}

\usepackage[
  format=plain,
  justification=raggedright,
  singlelinecheck=false
]{caption}
\usepackage{subcaption}

\usepackage{environ}

\usepackage{xcolor}
\usepackage{listings}

\usepackage{hyperref}

\newtheorem{theorem}{Theorem}
\newtheorem{corollary}{Corollary}[theorem]
\newtheorem{lemma}[theorem]{Lemma}
\newtheorem{remark}[theorem]{Remark}
\newtheorem{proposition}[theorem]{Proposition}
\newtheorem{definition}[theorem]{Definition}

\newcommand{\braket}[2]{%
  \left\langle #1\,\middle|\,#2\right\rangle%
}
\newcommand{\Tr}{\operatorname{Tr}}
\newcommand{\Pval}{P_{\mathrm{val}}}
\newcommand{\Psucc}{P_{\mathrm{succ}}}

\definecolor{codegreen}{rgb}{0,0.6,0}
\definecolor{codegray}{rgb}{0.5,0.5,0.5}
\definecolor{codepurple}{rgb}{0.58,0,0.82}
\definecolor{backcolour}{rgb}{0.95,0.95,0.92}

\lstdefinestyle{mystyle}{
  backgroundcolor=\color{backcolour},
  commentstyle=\color{codegreen},
  keywordstyle=\color{magenta},
  numberstyle=\tiny\color{codegray},
  stringstyle=\color{codepurple},
  basicstyle=\ttfamily\footnotesize,
  breakatwhitespace=false,
  breaklines=true,
  captionpos=b,
  keepspaces=true,
  numbers=left,
  numbersep=5pt,
  showspaces=false,
  showstringspaces=false,
  showtabs=false,
  tabsize=2
}

\newcommand{\MeshQuadraticC}{7.465353}
\newcommand{\MeshMaxProfileDeviation}{1.29}
\newcommand{\MeshIsoMinBias}{6.81}
\newcommand{\MeshIsoMaxBias}{9.55}
\newcommand{\MeshMaxPairedRelative}{1.25\times10^{-5}}
\newcommand{\MeshMaxSE}{1.67}
\newcommand{\LossGapMinPercent}{0.0518}
\newcommand{\LossGapMaxPercent}{31.5687}
\newcommand{\LossCertifiedMinPercent}{0.0518}
\newcommand{\MixedIdentityResidual}{8.97\times10^{-9}}
\newcommand{\RobustnessMaxGap}{1.07\times10^{-7}}
\newcommand{\SCSMaxDifference}{3.35\times10^{-10}}
\newcommand{\MeshMaxReconstruction}{8.76\times10^{-15}}

\providecommand{\HOneCodeLinkText}{The code and numerical data accompany
this manuscript as a reproducibility package.}

\begin{document}

\title{Quantum Oracle Value Readout: Coherence, Complementarity, and Fragility}

\author{Milad Ghadimi}
\email{milad.ghadimi@tu-dresden.de}
\affiliation{
Deutsche Telekom Chair of Communication Networks,
Technische Universit\"at Dresden,
01062 Dresden, Germany
}

\author{Hesam Soltanpanahi}
\email{hesam.soltan@gmail.com}
\affiliation{
Institute of Theoretical Physics and Mark Kac Center for
Complex Systems Research,
Jagiellonian University,
{\L}ojasiewicza 11,
30-348 Cracow, Poland
}

\author{Vahid Salari}
\email{vahid.salari1@ucalgary.ca}
\affiliation{
Institute for Quantum Science and Technology,
Department of Physics and Astronomy,
University of Calgary,
Calgary, Alberta, Canada
}

\author{Mariam Zomorodi}
\email{zomorodi@outlook.com}
\affiliation{
Department of Computer Science, Cracow University of Technology, Krakow, Poland
}

\date{\today}

\begin{abstract}
A response-register quantum oracle can reveal a function value or encode it
as a phase through kickback. We quantify how much value information can be
recovered from one query at a fixed fidelity to a target phase character,
and how noise and implementation errors affect readout performance.
For pure response states over finite Abelian groups, the known covariant
discrimination solution fixes the optimal readout probability exactly in
terms of the $\ell_1$-norm coherence of the response state in the character
basis. Using a known coherence bound, we determine the complete pure-state
readout interval at fixed target-character fidelity. Discarding the response
leaves a mixture of phase-oracle channels whose distance from the target,
measured by half the diamond norm, is at most the target-character
infidelity, with equality for surjective functions. We then identify when
coherence alone determines performance and when apparatus correlations or
loss profiles are also needed. For coherent analyzer errors, we derive the
full quadratic loss and an exact scalar reduction of the profile model on
the upper phase--value boundary. For mode-dependent loss, no-detection events
can carry value information, and a universal bound is attained under a
phase--loss compatibility condition. Photonic-mesh simulations test the
analyzer-error predictions, with the controlled oracle itself assumed.
Selected results extend to finite non-Abelian groups, where scalar phase
kickback is restricted to one-dimensional representations.
\end{abstract}

\maketitle
\onecolumngrid
\vspace{-22pt}
\begin{center}
\itshape
In memory of our dear friend, Farhad Fazileh.
\end{center}

\vspace{11pt}

\twocolumngrid

% Optional:
% \tableofcontents

% -----------------------------------
% Main body
% -----------------------------------

% -----------------------------------
\section{Introduction}
\label{sec:intro}
% -----------------------------------
Quantum algorithms are commonly formulated in an oracle, or black-box, model in which the computational problem is accessed through a controlled unitary.  This viewpoint already appears in the earliest quantum speedups and became central to Deutsch--Jozsa, Bernstein--Vazirani, Simon's algorithm, and quantum phase estimation \cite{deutsch1992rapid,bernstein1993quantum,simon1997power,Kitaev:1995qy,Cleve:1997dh, ghadimi2026generalized}. In the standard response-register representation, an oracle evaluates a function by translating an auxiliary register according to its value.  If that register is prepared in a translation eigenstate---a character state---the translation is returned as a phase on the query register.  This phase-kickback mechanism is one of the basic bridges between value and phase oracles and underlies many interference-based quantum algorithms.

The same physical oracle can therefore be used in two operationally different ways.  One may measure the translated response register to learn the value itself, or preserve it in a character state so that it supplies a well-defined phase while revealing no value information.  Intermediate response states interpolate between these limits.  Relations between value and phase oracles at the level of computational simulation have been studied in Refs.~\cite{Cleve:1997dh,BonehLipton1995,kashefi2002comparison, Hoyer:2005luj}.  We study the following single-query question in the response-register setting: how much value information can one response state carry while retaining fidelity to a prescribed phase character?

% -----------------------------------
\subsection{Related Work}
\label{sec:related_work}
The relevant literature spans covariant state discrimination, resource
theories of coherence and asymmetry, complementarity relations,
finite-group Fourier analysis, and linear-optical interferometry.
The ideal orbit-discrimination formulas used below are known.

\paragraph{Covariant discrimination of symmetric states.}
Value readout is minimum-error discrimination of a geometrically
uniform orbit ensemble
\cite{Helstrom:1969fri,holevo1973statistical,Yuen:1975ptn,
holevo2001statistical,Barnett:2008uix}. For a uniformly weighted
group orbit, an optimal measurement may be chosen covariant
\cite{holevo1973statistical,holevo2011probabilistic}; the square-root
measurement \cite{Hausladen01121994} is optimal for the Abelian
ensembles considered here
\cite{Eldar:2000dav,ban1997optimum,Krovi:2015bah} and for pure-state
geometrically uniform ensembles with non-Abelian generating groups
\cite{eldar2004optimal,Zhou:2025iac}.
Equation~\eqref{eq:main-resource} is therefore established, while
Eq.~\eqref{eq:main-nonabelian} specializes the covariant
maximum-likelihood solution of Chiribella \emph{et al.}
\cite{Chiribella:2004wyf}. Related expressions appear in
Bagan \emph{et al.}\ \cite{Bagan:2018uuk} and the multiplicity-bearing
ensemble of Bacon, Childs and van Dam \cite{Bacon:2005cri}.

\paragraph{Coherence, asymmetry, and operational tasks.}
For pure Abelian response states, Eq.~\eqref{eq:main-resource}
identifies value-readout capability with the $\ell_1$-norm coherence
in the character basis, a standard coherence monotone
\cite{Baumgratz:2013ecx}. Operational interpretations of coherence
and asymmetry through discrimination tasks were established in
Refs.~\cite{Napoli:2016vpc,Piani:2016lse,Bagan:2018uuk}.
Our response-register setting specializes these connections:
$\ell_1$ coherence determines the pure-state optimum, while mixed-state
readout is governed by robustness of coherence,
Eq.~\eqref{eq:main-mixed}.

\paragraph{Complementarity relations.}
Wave--particle complementarity compares interference with information
about a particle's path
\cite{Jaeger:1995sv,Englert:1996zz,durr1998origin}.
Bagan \emph{et al.}\ \cite{bagan2016relations} relate coherence in the path
system to minimum-error discrimination of detector states. In our
setting, the response register plays the detector role. The induced
query channel is governed by the same detector-overlap kernel, but
the quantities compared differ: path coherence depends on weighted
absolute overlaps, whereas the target-character population $\Phi$
is a Fourier coefficient of that kernel. Our $C(\eta)$ measures
coherence of the response before the oracle, in the character basis.
Appendix~\ref{app:h1querychannel} makes these relations explicit.
Related formulations use unambiguous path discrimination
\cite{bera2015duality} or operational measures of interferometric visibility
\cite{biswas2017interferometric}.

Bagan \emph{et al.}\ \cite{Bagan:2018uuk} also compare discrimination
of finite-group transformations with discrimination of ancillary
states that label irreducible subspaces. Their multiplicity-free
pure-state formula includes our Abelian readout identity. We use these
established discrimination results to study a prescribed phase-kickback
target, the induced query channel, coherent analyzer errors, and
mode-dependent loss, including information carried by the vacuum outcome.

\paragraph{Group-theoretic oracle models.}
For non-Abelian groups, Fourier coefficients in the regular
representation are matrix valued
\cite{serre_1977,Terras_1999,fulton_1991,diaconis1988group}.
Bucicovschi \emph{et al.}\ study group-orbit discrimination in an
oracle setting \cite{Bucicovschi:2015qhz}; related query-complexity
questions appear in Ref.~\cite{Copeland:2021bvd}.
Shakeel optimizes the response-register slate for hidden-subgroup
identification \cite{Shakeel:2012tgs}. More broadly, the hidden-subgroup
literature develops Fourier sampling and measurements for coset states
\cite{Ettinger_1999,hallgren_2003,Moore:2003icv,Childs_2010,
lomont2004hidden,Moore:2005rmn,hayashi2006quantum}; these ensembles
are indexed by candidate subgroups rather than translations of a
fixed response state.

\paragraph{Oracle coherence and interferometric realizations.}
Coherence has also been analyzed in Deutsch--Jozsa decision problems
\cite{hillery2016coherence} and through its depletion in Grover search
\cite{shi2017coherence}. Those studies concern coherence in an algorithm,
whereas our functional concerns information recoverable from one
response. The faulty-oracle result of Regev and Schiff
\cite{regev2008impossibility} concerns search query complexity under stochastic
oracle failures; our coherent-error model concerns the analyzer following
an assumed oracle. Universal Reck \cite{Reck:1994ujl} and Clements
\cite{Clements:2016rkw} interferometers implement this response analyzer.
Appendix~\ref{sec:photonic_numerical} specifies its numerical error model.

\subsection{Contributions and Scope}
\label{sec:contributions_scope}
The contributions of this work are:
\begin{enumerate}
\item a response-register formulation connecting established coherence
and discrimination results to phase kickback and explicit value recovery,
together with the complete pure-state phase--value interval and a
query-channel interpretation: discarding the response gives a mixture
of phase-oracle channels, with target-channel distance bounded by
$1-\Phi$ and equal to it for surjective functions
(Theorem~\ref{thm:tradeoff} and Appendix~\ref{app:h1querychannel});

\item a mixed-state formulation using the known robustness-of-coherence
identity, with attainability conditions for the $\ell_1$ upper bound
and an exact coherence-scaling corollary for phase-factorizable inputs
(Appendix~\ref{sec:incoherent});

\item a quadratic loss formula retaining the full shift-indexed
correlations of coherent analyzer errors, conditions for its reduction
to a variance profile, and an exact off-diagonal scalar reduction of
that profile on the upper phase--value boundary
(Appendices~\ref{sec:coherent} and \ref{app:second});

\item an exact survival--vacuum decomposition under mode-dependent
loss, a universal dual-certificate bound with a sufficient phase--loss
compatibility condition for attainment, and the result that compatible
nonuniform loss is never harmful at fixed mean transmission when the
vacuum outcome is used optimally (Appendix~\ref{sec:mdl});

\item a matrix-Fourier description of the known finite-group readout
formula and its Plancherel-divergence interpretation, the restriction
of scalar phase kickback to one-dimensional representations, and
finite-group extensions of selected noise and loss results
(Appendix~\ref{sec:nonabelian}).
\end{enumerate}

The ideal discrimination formulas and robustness identities are
established results. The upper phase--value boundary uses the known
coherence bound of Ref.~\cite{Piani:2016lse}; an elementary lower
bound and its attainability complete the interval. Stationarity follows
from optimality, and the linear coherence-scaling relation is a corollary.
We also record repeated independent probes of the same value with a
collective readout; this fixes the probe preparation rather than
optimizing a general adaptive query strategy
(Appendix~\ref{sec:composition}). The channel information--disturbance
interval concerns the specified controlled-translation interaction and
a surjective function.

Together, these results identify when a scalar coherence resource
determines readout performance and when apparatus correlations or
survival information are also needed. The framework can be used to
select response states and benchmark response analyzers. The photonic
construction assumes the controlled oracle on the joint query and
response registers. Figure~\ref{fig:Oracles} summarizes its two uses.

\begin{figure*}
\centering
% Preserve the existing circuit asset when installing in the author's project.
% The fallback makes this source bundle independently compilable.
\IfFileExists{fig_1_tikz.pdf}{%
  \includegraphics[width=0.70\linewidth]{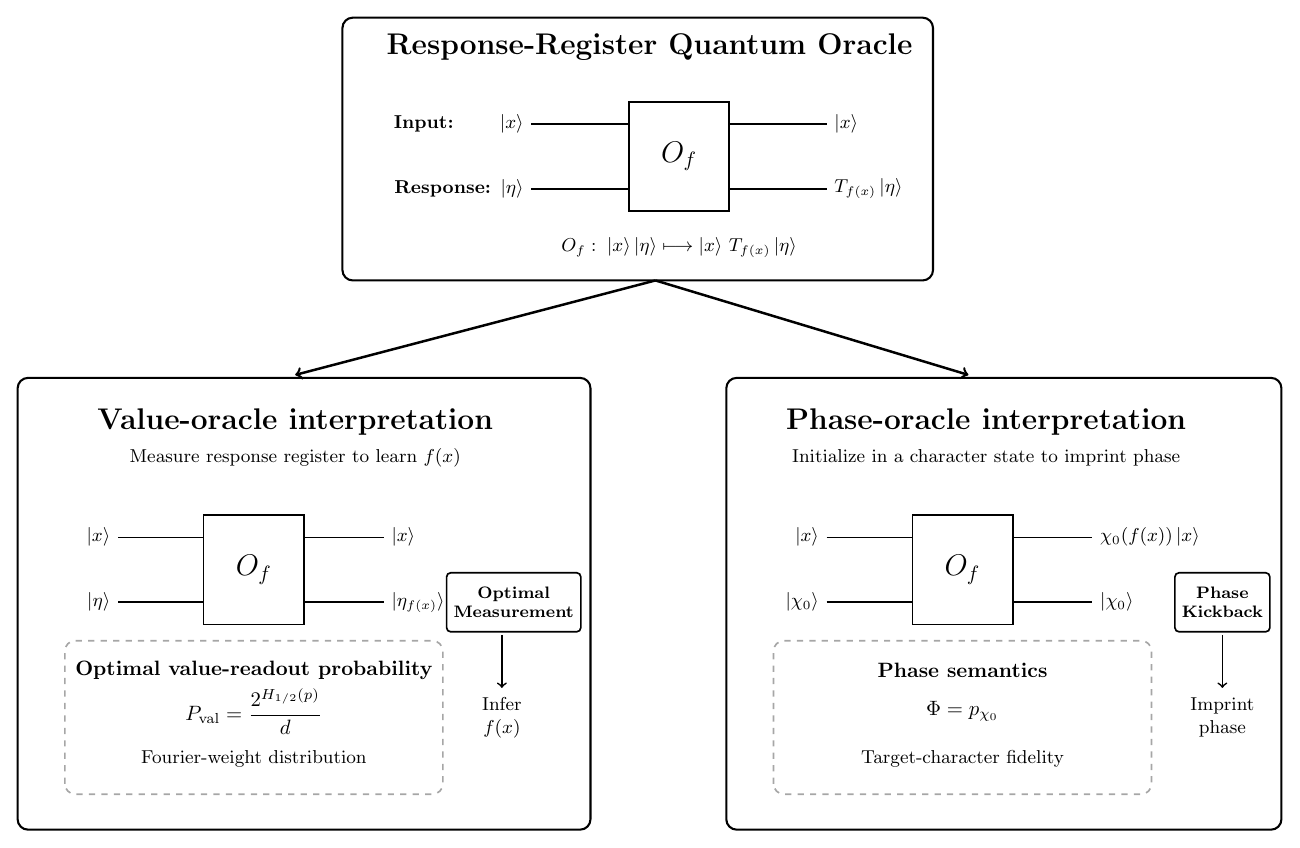}%
}{%
  \includegraphics[width=0.70\linewidth]{fig_1_fallback.pdf}%
}
\caption{Two uses of the response-register oracle.  Translated response
states are discriminated for value readout (left), whereas a character
state produces the phase $\chi_0(f(x))$ by kickback (right).}
\label{fig:Oracles}
\end{figure*}

The remainder of the paper is organized as follows.
Section~\ref{sec:compact-results} presents the Abelian operational results and phase--value tradeoff.
Section~\ref{sec:compact_robustness} analyzes covariant noise, coherent analyzer errors, and mode-dependent loss.
Section~\ref{sec:compact_nonabelian} extends the framework beyond Abelian response groups, and
Sec.~\ref{sec:photonic_impl_check} develops the photonic implementation and numerical checks.
Section~\ref{sec:conclusion} concludes.
Detailed proofs, equality conditions, and numerical methods are provided in the appendices following the references.

% -----------------------------------
\section{Operational Results}
\label{sec:compact-results}
% -----------------------------------

We begin with finite Abelian groups because their commuting translations
are simultaneously diagonalized by one-dimensional characters.  This
places value discrimination and phase kickback in the same scalar Fourier
basis and includes the common cyclic and bit-string response registers;
the matrix-valued extension to finite non-Abelian groups is summarized in
Sec.~\ref{sec:compact_nonabelian}.

Let $A$ be a finite Abelian group of order $d>1$, with character group
$\widehat A$.  A response-register oracle acts as
\begin{equation}
  O_f\ket{x}\ket{\eta}
  =
  \ket{x}T_{f(x)}\ket{\eta},
  \qquad
  T_a\ket{z}=\ket{z+a}.
  \label{eq:main-oracle}
\end{equation}
In the character basis,
$T_a\ket{\chi}=\chi(a)\ket{\chi}$, and we write a pure response state as
\begin{equation}
  \ket{\eta}
  =
  \sum_{\chi\in\widehat A}
  \sqrt{p_\chi}\,e^{i\theta_\chi}\ket{\chi},
  \qquad
  p_\chi\geq0,
  \qquad
  \sum_\chi p_\chi=1.
  \label{eq:main-state}
\end{equation}
For fixed $x$, suppose that the unknown value $a=f(x)$ is uniformly
distributed over $A$.  One-query value readout is then the
minimum-error discrimination problem for the orbit ensemble
\begin{equation}
  \mathcal E_\eta
  =
  \left\{
    \frac1d,\,
    T_a\ket{\eta}
  \right\}_{a\in A},
  \label{eq:main-ensemble}
\end{equation}
and we denote its optimal success probability by $\Pval$.
The uniform prior does not require the function itself to be random.
An optimal covariant POVM has the same conditional success probability
for every value $a$. Since the worst-case success of any POVM is
bounded by its uniform-prior average, $\Pval$ also equals the minimax
optimum over the value labels.

For a response-register state $\rho$, define its $\ell_1$-norm coherence
in the character basis by
\begin{equation}
  C(\rho)
  :=
  \sum_{\chi\neq\chi'}
  \left|
    \bra{\chi}\rho\ket{\chi'}
  \right|.
  \label{eq:main-coherence}
\end{equation}
For the pure state \eqref{eq:main-state}, we write
$C(\eta):=C(\ket{\eta}\!\bra{\eta})$.  The square-root measurement is
optimal for the Abelian geometrically uniform ensemble
\eqref{eq:main-ensemble}
\cite{Eldar:2000dav,ban1997optimum,Krovi:2015bah,Zhou:2025iac},
and gives
\begin{equation}
  \boxed{
  \Pval
  =
  \frac1d
  \left(
    \sum_{\chi\in\widehat A}\sqrt{p_\chi}
  \right)^2
  =
  \frac{1+C(\eta)}{d}
  }.
  \label{eq:main-resource}
\end{equation}
One attaining measurement is the covariant POVM
\begin{equation}
  M_a=V\ket{a}\!\bra{a}V^\dagger,
  \qquad
  V=
  \sum_{\chi\in\widehat A}
  e^{i\theta_\chi}\ket{\chi}\!\bra{\chi},
  \label{eq:main-optimal-povm}
\end{equation}
where phases associated with vanishing $p_\chi$ may be chosen
arbitrarily.  Thus the optimal probability depends only on the Fourier
weights $p_\chi$, whereas this convenient optimal measurement depends
only on the phases $\theta_\chi$.

Equivalently, with base-$2$ R\'enyi entropy,
\begin{equation}
  S_{\mathrm{eff}}(p)
  :=
  \left(
    \sum_\chi\sqrt{p_\chi}
  \right)^2
  =
  2^{H_{1/2}(p)}
  =
  1+C(\eta)
  =
  d\Pval .
  \label{eq:main-effective-support}
\end{equation}
We call $S_{\mathrm{eff}}(p)$ the \emph{effective Fourier support}.
It is an entropy-based effective number rather than the literal support
cardinality.  Equation~\eqref{eq:main-resource} therefore identifies
character-basis coherence as the exact pure-state resource for
single-query value readout.

Fix a target character $\chi_0$ and define the target-character fidelity
of the pure response state by
\begin{equation}
  \Phi
  :=
  \left|
    \braket{\chi_0}{\eta}
  \right|^2
  =
  p_{\chi_0}.
  \label{eq:main-fidelity}
\end{equation}

The parameter $\Phi$ also controls the operation induced on the query.
Define $U_\chi^f:=\sum_x\chi(f(x))\ket{x}\!\bra{x}$ and let
$\mathcal U_\chi^f$ denote its unitary channel. Discarding the response
register gives
\begin{equation}
\begin{aligned}
\mathcal E_f(\sigma)
&:=\Tr_R[O_f(\sigma\otimes\ket\eta\!\bra\eta)O_f^\dagger]\\
&=\sum_\chi p_\chi\mathcal U_\chi^f(\sigma),
\end{aligned}
\label{eq:h1mainquerychannel}
\end{equation}
and therefore
\begin{equation}
\delta_f:=\tfrac12\|\mathcal E_f-\mathcal U_{\chi_0}^f\|_\diamond
\leq1-\Phi.
\label{eq:h1mainquerybound}
\end{equation}
Equality holds when $f$ is surjective. For a pure response, combining
this equality with the interval defined in
Eqs.~\eqref{eq:main-tradeoff}--\eqref{eq:main-tradeoff-boundaries} gives
$P_-(1-\delta_f)\leq\Pval\leq P_+(1-\delta_f)$.
This is an information--disturbance relation for the specified
controlled-translation interaction. Surjectivity is sufficient;
for a restricted image, distinct characters can induce the same channel.
Uniform character weights dephase the query between distinct function
values, while a point mass gives a unitary phase oracle. The proof and
the detector-overlap kernel are given in
Appendix~\ref{app:h1querychannel}.

At fixed $\Phi$, the complete achievable interval of value-readout
probabilities is
\begin{equation}
  \boxed{
    P_-(\Phi)\leq\Pval\leq P_+(\Phi)
  },
  \label{eq:main-tradeoff}
\end{equation}
where
\begin{align}
  P_-(\Phi)
  &:=
  \frac{1}{d}
  \left(
    \sqrt{\Phi}+\sqrt{1-\Phi}
  \right)^2,
  \nonumber\\
  P_+(\Phi)
  &:=
  \frac{1}{d}
  \left[
    \sqrt{\Phi}
    +\sqrt{(d-1)(1-\Phi)}
  \right]^2.
  \label{eq:main-tradeoff-boundaries}
\end{align}
The upper boundary follows from the coherence bound
of Lemma~2 in Ref.~\cite{Piani:2016lse} and the
discrimination identity \eqref{eq:main-resource}.
Here it expresses the compatibility between fidelity
to a prescribed phase-kickback response state and
the optimal probability of recovering the oracle value.
The upper boundary is attained if and only if the non-target weights are
uniform,
\begin{equation}
  p_\chi
  =
  \frac{1-\Phi}{d-1}
  \qquad
  (\chi\neq\chi_0),
  \label{eq:main-upper-equality}
\end{equation}
The lower bound follows from nonnegativity of the cross terms and
is attained by concentrating all non-target weight on one character. The Fourier phases are arbitrary on
both boundaries.  For $d=2$, the two boundaries coincide.

Two points on the upper boundary are the extremal operational modes of
the same oracle.  At $\Phi=1$, the response state is
$\ket{\chi_0}$ up to global phase, and every translated state differs
only by the global factor $\chi_0(a)$.  Consequently, no measurement can
extract value information and $\Pval=1/d$.  At $\Phi=1/d$, all character
weights are uniform and $\Pval=1$, giving exact value readout.  The
interval $1/d\leq\Phi\leq1$ is the Pareto-relevant phase--value tradeoff
branch; for $\Phi<1/d$, increasing $\Phi$ toward $1/d$ can improve both
target-character fidelity and optimal value readout.

\paragraph{Repeated probes of the same value.}
If $k$ independent copies of the same pure response are prepared,
the hypotheses are $(T_a\ket\eta)^{\otimes k}$ with a common
unknown value $a$. An optimal collective measurement has success
\begin{equation}
\Pval^{(k)}=\frac1d
\left(\sum_{\chi\in\widehat A}\sqrt{p^{*k}(\chi)}\right)^2,
\label{eq:h1mainrepeated}
\end{equation}
where $p^{*k}$ is convolution on the character group. This fixes
product probe preparation and permits a collective final measurement;
it does not optimize adaptive or entangled query strategies.
Appendix~\ref{sec:composition} gives the derivation and distinguishes
this task from independent values in a product group.

Figure~\ref{fig:phase-value complementarity} shows the complete achievable
region for $d=16$.  The orbit-discrimination derivation, optimal
measurement, coherence identity, equality conditions, and boundary
derivations are given in
Secs.~\ref{sec:setting}--\ref{sec:tradeoff}.

\begin{figure}[t]
\centering
\includegraphics[width=.96\columnwidth]{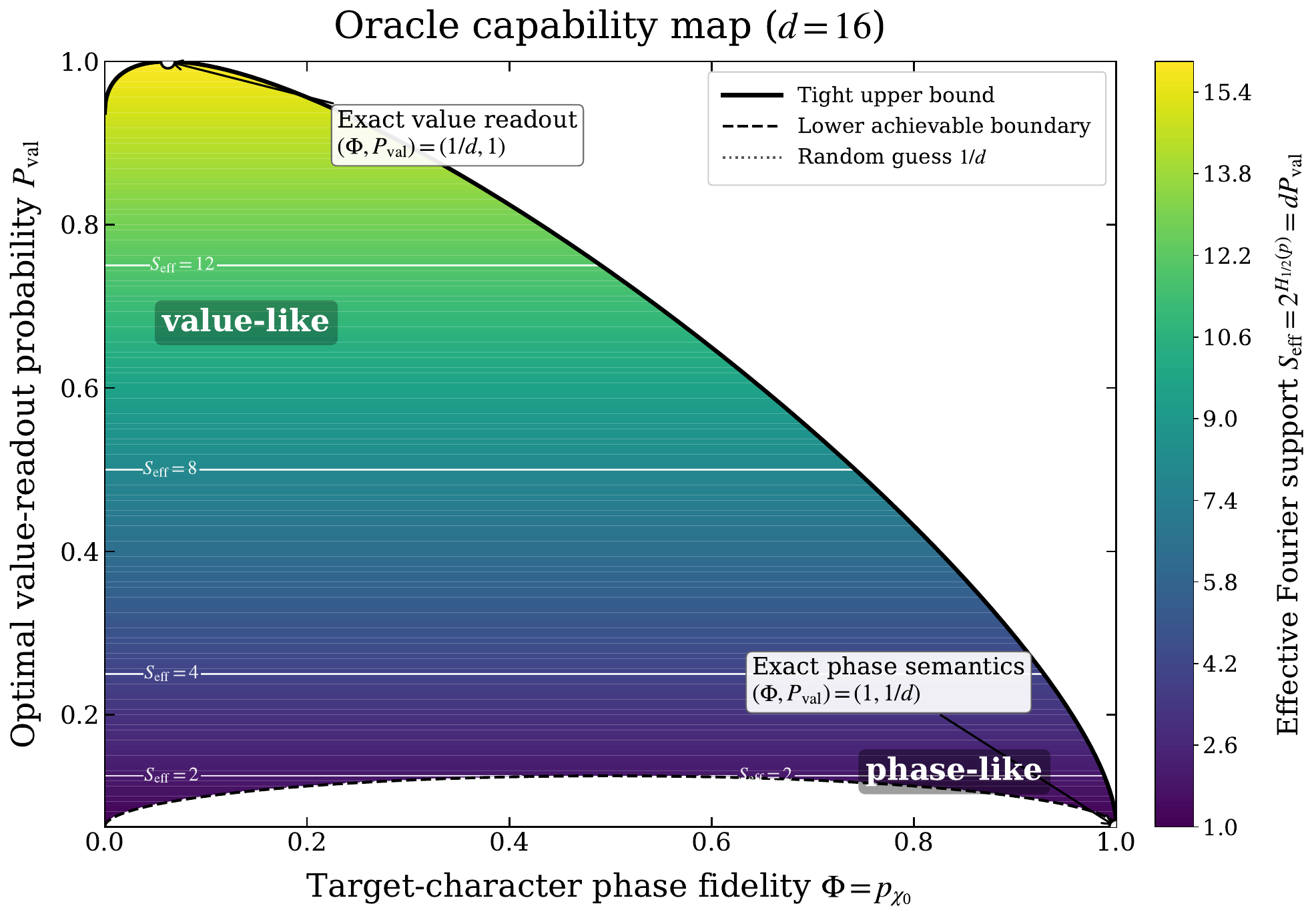}
\caption{Complete phase--value capability region for $d=16$.
The horizontal axis is the target-character fidelity $\Phi=p_{\chi_0}$,
and the vertical axis is the optimal one-query value-readout probability.
The solid and dashed curves give the upper and lower boundaries in
Eq.~\eqref{eq:main-tradeoff}; the dotted line marks $1/d$.
The color scale resolves $S_{\mathrm{eff}}=d\Pval$.
The points $(1/d,1)$ and $(1,1/d)$ give exact value readout and
exact target-character phase semantics. The branch $1/d\leq\Phi\leq1$
is Pareto relevant.}
\label{fig:phase-value complementarity}
\end{figure}

% -----------------------------------
\section{Robustness and Fragility}
\label{sec:compact_robustness}
% -----------------------------------

The ideal theory identifies the best attainable value-readout
capability.  We now ask how much of that capability remains under
imperfect preparation, evolution, and measurement.  The answer depends
on the error mechanism: coherence alone determines the response to
coherence-scaling noise, whereas coherent analyzer errors and
mode-dependent loss require additional device-dependent information.

\emph{Mixed states and covariant noise.}
For a mixed response state $\rho$, let
\begin{equation}
\begin{aligned}
  C_R(\rho)
  :=
  \min\bigl\{
    &s\geq0:\rho\leq(1+s)\delta,\\
    &\delta\ \text{incoherent in the character basis}
  \bigr\}.
\end{aligned}
\end{equation}

denote the robustness of coherence.  The uniformly weighted orbit
ensemble generated by $\rho$ satisfies the exact identity and bound
\begin{equation}
  \boxed{
  \Pval(\rho)
  =
  \frac{1+C_R(\rho)}{d}
  \leq
  \frac{1+C(\rho)}{d}
  }.
  \label{eq:main-mixed}
\end{equation}
The upper bound is attained whenever there exist real numbers
$\{\alpha_\chi\}$ such that every nonzero matrix element obeys
\begin{equation}
  \bra{\chi}\rho\ket{\chi'}
  =
  \left|\bra{\chi}\rho\ket{\chi'}\right|
  e^{i(\alpha_\chi-\alpha_{\chi'})}.
  \label{eq:main-phase-factorization}
\end{equation}
If every character-basis matrix element is nonzero, this condition is
also necessary for equality.  Thus the $\ell_1$-norm coherence is a
universal upper certificate for mixed-state readout, but it determines
the optimum exactly only when the coherence phases are compatible.
The detailed statement, proof, and numerical saturation and strictness
tests are given in Theorem~\ref{thm:mixed},
Eq.~\eqref{eq:mixedbound}, and Appendix~\ref{app:numsdp}.

Suppose now that the post-oracle channel is translation covariant,
\begin{equation}
  \Lambda\!\left(T_aXT_a^\dagger\right)
  =
  T_a\Lambda(X)T_a^\dagger
  \qquad
  \text{for every }a\in A.
  \label{eq:main-covariance}
\end{equation}
The noisy states then remain a group orbit, so
Eq.~\eqref{eq:main-mixed} applies to
$\rho=\Lambda(\ket{\eta}\!\bra{\eta})$.  Covariance alone, however,
does not imply a scalar noise law: a general covariant channel may mix
operators belonging to the same coherence mode.  The corresponding
mode decomposition is developed in Sec.~\ref{sec:modes}.

\emph{Coherence-scaling noise.}
A stronger and operationally important class consists of
coherence-scaling channels $\Lambda_\tau$, which satisfy
\begin{equation}
  \bra{\chi}\Lambda_\tau(\rho)\ket{\chi'}
  =
  \tau\bra{\chi}\rho\ket{\chi'}
  \qquad
  (\chi\neq\chi'),
  \qquad
  0\leq\tau\leq1.
  \label{eq:main-coherence-scaling}
\end{equation}
For every pure input, and more generally for every mixed input satisfying
Eq.~\eqref{eq:main-phase-factorization}, the readout probability obeys
the exact law
\begin{equation}
\boxed{
\begin{aligned}
  &\Pval\!\left(\Lambda_\tau(\rho)\right)
  =
  \frac{1+\tau C(\rho)}{d},\\
  &\Delta\Pval
  :=
  \Pval\!\left(\Lambda_\tau(\rho)\right)-\Pval(\rho),\\
  &\Delta\Pval
  =
  -\,(1-\tau)\frac{C(\rho)}{d}.
\end{aligned}
}
\label{eq:main-linear}
\end{equation}
Thus the coherence that supplies ideal readout capability also fixes,
in exact proportion, its degradation under this class of incoherent
processes.  This result is proved as
Corollary~\ref{cor:linear}, following
Definition~\ref{def:coherencescaling}; standard examples and its
operational consequences are discussed in
Secs.~\ref{sec:modes} and \ref{sec:incohcons}.

\emph{Coherent analyzer errors.}
The optimal measurement in Eq.~\eqref{eq:main-optimal-povm} is realized
by applying $V^\dagger$ and measuring in the value basis.  Suppose that
the implemented analyzer is instead
\begin{equation}
  U(\epsilon)
  =
  e^{i\epsilon G}V^\dagger,
  \qquad
  G=G^\dagger.
  \label{eq:main-imperfect-analyzer}
\end{equation}
The ideal apparatus is stationary against every fixed Hermitian
generator:
\begin{equation}
  \left.
  \frac{d\Pval}{d\epsilon}
  \right|_{\epsilon=0}
  =
  0.
  \label{eq:main-stationarity}
\end{equation}
This stationarity follows from optimality of the measurement;
Lemma~\ref{lem:stat} also verifies it directly in shift coordinates.

To describe the quadratic term, define the phase-stripped state
\begin{equation}
\begin{aligned}
  \ket{\varphi}
  &:=
  \sum_{\chi}\sqrt{p_\chi}\ket{\chi},\\
  \varphi_w
  &:=
  \braket{w}{\varphi},
  \qquad
  \varphi_0^2
  =
  \frac{1+C(\eta)}{d}.
\end{aligned}
\label{eq:main-varphi}
\end{equation}
For a centered Hermitian error ensemble satisfying
\begin{equation}
  \mathbb E[G_{ij}]=0,
  \qquad
  \mathbb E[G_{ij}G_{kl}]=0
  \quad\text{unless}\quad
  (k,l)=(j,i),
  \label{eq:main-pairing}
\end{equation}
introduce the shift-offset variance profile
\begin{equation}
  V(w)
  :=
  \sum_{a\in A}
  \mathbb E\!\left[
    |G_{a,a+w}|^2
  \right].
  \label{eq:main-profile}
\end{equation}
The leading average loss is then
\begin{equation}
  \boxed{
  \mathbb E[\Delta\Pval]
  =
  -\frac{\epsilon^2}{d}
  \sum_{\substack{w\in A\\w\neq0}}
  V(w)
  \left(
    \varphi_0^2-|\varphi_w|^2
  \right)
  +O(\epsilon^3)
  }.
  \label{eq:main-coherent}
\end{equation}
Unlike Eq.~\eqref{eq:main-linear}, this loss is not generally determined
by $C$ alone: it depends on the correlation profile of the implemented
apparatus.  Equation~\eqref{eq:main-coherent} is
Proposition~\ref{prop:Vw}; the expression for general correlated error
ensembles is derived in Appendix~\ref{app:second}.

At the exact phase corner, the translated response density operators
are identical. Any value-independent post-oracle channel preserves
this indistinguishability, so the uniform-prior value-readout
probability remains $1/d$; see Eq.~\eqref{eq:cornerexact}.
This concerns the absence of value information in the response,
not immunity of a phase-based algorithm to physical noise.

\emph{Mode-dependent loss.}
Spatially nonuniform loss generally breaks translation covariance.
For a pure response state, write
\begin{equation}
\begin{aligned}
  L
  &:=
  \sum_z\sqrt{\tau_z}\ket{z}\!\bra{z},
  &\qquad
  \bar{\tau}
  &:=
  \frac{1}{d}\sum_z\tau_z,\\
  \ket{v_a}
  &:=
  LT_a\ket{\eta},
  &
  w_a
  &:=
  \braket{v_a}{v_a}.
\end{aligned}
\label{eq:main-loss-definitions}
\end{equation}
Here $w_a$ is the survival probability conditioned on oracle value
$a$.  Let
\begin{equation}
  P_{\mathrm{ph}}
  :=
  \max_{\substack{M_a^{\mathrm{ph}}\succeq0\\
                  \sum_aM_a^{\mathrm{ph}}=\mathbb I_d}}
  \sum_a
  \bra{v_a}M_a^{\mathrm{ph}}\ket{v_a}
  \label{eq:main-photonic-objective}
\end{equation}
denote the unnormalized optimal discrimination objective in the
surviving photonic sector.  Orthogonality of the surviving and vacuum
sectors gives the exact decomposition
\begin{equation}
  \Pval
  =
  \frac1d
  \left[
    P_{\mathrm{ph}}+1-\min_a w_a
  \right].
  \label{eq:main-loss-decomposition}
\end{equation}
Thus a no-detection event is not necessarily devoid of information:
when $w_a$ depends on $a$, it changes the posterior distribution of the
oracle value.

The surviving-sector objective and total readout probability satisfy
the universal bounds
\begin{equation}
\boxed{
\begin{aligned}
  P_{\mathrm{ph}}
  &\leq
  \bar{\tau}\bigl(1+C(\eta)\bigr),\\
  \Pval
  &\leq
  \frac{
    \bar{\tau}\bigl(1+C(\eta)\bigr)
    +1-\min_a w_a
  }{d}.
\end{aligned}
}
\label{eq:main-loss}
\end{equation}
Let $V$ be the phase unitary defined in
Eq.~\eqref{eq:main-optimal-povm}, with its phases extended arbitrarily
outside the support of $p$.  Both bounds are attained whenever those
extensions can be chosen so that
\begin{equation}
  [L,V]=0.
  \label{eq:main-loss-compatibility}
\end{equation}
This is a sufficient phase--loss compatibility condition, not a claimed
necessary condition for saturation by an arbitrary POVM.  Under this
condition,
\begin{equation}
  \Pval\bigl(\{\tau_z\}\bigr)
  \geq
  \Pval\bigl(\text{uniform }\bar\tau\bigr)
  =
  \frac{1+\bar\tau C(\eta)}{d},
  \label{eq:main-nonuniform-gain}
\end{equation}
with equality if and only if $w_a$ is independent of $a$.  Hence
nonuniform loss at fixed mean transmission can improve value readout:
the symmetry breaking encoded in the survival probabilities can itself
carry information.

The exact decomposition, dual bound, phase--loss compatibility
condition, and comparison with uniform loss are derived in
Eqs.~\eqref{eq:wa} and \eqref{eq:mdldecomp},
Theorem~\ref{thm:mdl}, Corollary~\ref{cor:LVcompatibility}, and
Corollary~\ref{cor:mdlgain}, respectively.

Taken together, these results delimit the role of coherence precisely.
It is the exact scalar fragility parameter for coherence-scaling noise.
For coherent analyzer errors, it sets the state-dependent scale but
must be combined with the full apparatus correlations. The profile
$V(w)$ suffices under the stated pairing condition. For
mode-dependent loss, it enters a universal bound, while the actual
performance also depends on the survival pattern $\{w_a\}$ and its
compatibility with the response phases.

% -----------------------------------
\section{Beyond Abelian Response Groups}
\label{sec:compact_nonabelian}
% -----------------------------------

The orbit-discrimination formulation extends to every finite group, although its Fourier description becomes matrix-valued.  Let $G$ be a finite group of order $d$, and let the response register carry the group algebra $\mathbb{C}[G]$, with left translations $T_a\ket{z}=\ket{az}$.  Let $\widehat{G}$ denote a complete set of inequivalent irreducible unitary representations $\rho_\lambda$ of $G$, with $d_\lambda:=\dim\rho_\lambda$.  For each $\lambda\in\widehat{G}$, define
\begin{equation}
  \widehat{\eta}(\lambda)
  :=
  \sqrt{\frac{d_\lambda}{d}}
  \sum_{z\in G}
  \braket{z}{\eta}\,\rho_\lambda(z).
  \label{eq:main-na-fourier}
\end{equation}
This convention gives $\sum_\lambda\|\widehat{\eta}(\lambda)\|_F^2=1$, where $\|\cdot\|_F$ denotes the Frobenius norm, while left translation acts within each Fourier block as
\begin{equation}
  \widehat{T_a\eta}(\lambda)
  =
  \rho_\lambda(a)\widehat{\eta}(\lambda).
  \label{eq:main-na-translation}
\end{equation}

For a pure response state, optimal discrimination of the uniformly weighted orbit $\{d^{-1},T_a\ket{\eta}\}_{a\in G}$ gives
\begin{equation}
  \Pval
  =
  \frac1d
  \left(
    \sum_{\lambda\in\widehat G}
    \sqrt{d_\lambda}\,
    \|\widehat{\eta}(\lambda)\|_1
  \right)^2,
  \label{eq:main-nonabelian}
\end{equation}
where $\|\cdot\|_1$ is the trace norm \cite{Chiribella:2004wyf,eldar2004optimal,Zhou:2025iac}.  Accordingly, define
\begin{equation}
  C_G(\eta)
  :=
  \left(
    \sum_{\lambda\in\widehat G}
    \sqrt{d_\lambda}\,
    \|\widehat{\eta}(\lambda)\|_1
  \right)^2-1,
  \label{eq:main-na-resource}
\end{equation}
so that
\begin{equation}
  \boxed{
    \Pval=\frac{1+C_G(\eta)}{d}
  }.
  \label{eq:main-na-master}
\end{equation}
The trace norm incorporates the singular-value structure linking the representation and multiplicity spaces.  Thus $C_G$ is an operational pure-state readout resource, rather than a basis-dependent coherence norm.  It satisfies $0\leq C_G\leq d-1$, is invariant under left translation, composes multiplicatively through $1+C_G$, and reduces to the character-basis $\ell_1$-coherence $C$ when $G$ is Abelian.  The supporting Fourier analysis and proofs of these properties are given in Appendix~\ref{sec:naresource}, particularly Propositions~\ref{prop:naprops} and \ref{prop:namult}.

The limitation concerns phase semantics.  Exact scalar kickback requires
\begin{equation}
  T_a\ket{\eta}
  =
  \omega(a)\ket{\eta}
  \qquad
  \text{for every }a\in G,
  \label{eq:main-na-kickback}
\end{equation}
which forces $\omega$ to be a one-dimensional unitary representation. Every such representation factors through the abelianization $G/[G,G]$; higher-dimensional irreducible representations instead produce genuinely matrix-valued actions and do not define scalar phase kickback.  Nevertheless, for a chosen one-dimensional target representation $\lambda_0$, with
\begin{equation}
  \Phi
  :=
  \|\widehat{\eta}(\lambda_0)\|_F^2,
  \label{eq:main-na-fidelity}
\end{equation}
the tight upper phase--value boundary retains the Abelian form
\begin{equation}
\boxed{
\begin{aligned}
  \Pval
  \leq
  \frac1d
  \left[
    \sqrt{\Phi}
    +\sqrt{(d-1)(1-\Phi)}
  \right]^2 
\end{aligned}
}.
\label{eq:main-na-tradeoff}
\end{equation}
Equality holds if and only if
\begin{equation}
  \sigma_{\lambda,k}^{2}
  =
  \frac{d_\lambda(1-\Phi)}{d-1}
  \qquad
  \left(
    \lambda\neq\lambda_0,\;
    k=1,\ldots,d_\lambda
  \right),
  \label{eq:main-na-equality}
\end{equation}
where $\sigma_{\lambda,k}$ are the singular values of $\widehat{\eta}(\lambda)$.  Equivalently, the non-target weight is Plancherel distributed across the irreducible blocks, with equal singular values within each block.  The scalar-phase restriction and the derivation of these equality conditions are given in Appendix~\ref{sec:naphase}.

Several robustness statements also survive without commutativity.
For mixed states, define
\begin{equation}
C_G(\rho):=d\Pval(\rho)-1=R_G(\rho),
\label{eq:main-na-mixed-resource}
\end{equation}
where $\Pval(\rho)$ is the optimal discrimination probability for
the uniform orbit $\{d^{-1},T_a\rho T_a^\dagger\}_{a\in G}$.
The equality with robustness of asymmetry is established
\cite{Piani:2016lse,Bagan:2018uuk}. This resource is monotone under
translation-covariant channels by Proposition~\ref{prop:namonotonicity}.
The remaining mixed-state question is a closed evaluation in
matrix-Fourier coordinates, rather than the existence of a robustness
interpretation.

For $\tau\in[0,1]$, depolarizing noise $\rho\mapsto\tau\rho+(1-\tau)\mathbb{I}/d$ and uniform flagged loss $\rho\mapsto\tau\rho+(1-\tau) \ket{\mathrm{vac}}\!\bra{\mathrm{vac}}$, with $\ket{\mathrm{vac}}$ an orthogonal vacuum flag, obey the exact linear law
\begin{equation}
  P_{\mathrm{val}}^{(\mathrm{out})}
  =
  \frac{1+\tau C_G(\rho)}{d},
  \label{eq:main-na-linear}
\end{equation}
as shown in Proposition~\ref{prop:nadepolarizing} and Eq.~\eqref{eq:nauniformloss}.  Moreover, an optimal projective analyzer remains stationary to first order (Proposition~\ref{prop:nastationarity}).  For mode-dependent loss, let
\begin{equation}
\begin{aligned}
  L&:=\sum_{z\in G}\sqrt{\tau_z}\ket{z}\!\bra{z},
  &\quad
  \bar\tau&:=\frac1d\sum_z\tau_z,\\
  \ket{v_a}&:=LT_a\ket{\eta},
  &
  w_a&:=\braket{v_a}{v_a}.
\end{aligned}
\label{eq:main-na-loss-definitions}
\end{equation}
Writing $P_{\mathrm{ph}}$ for the unnormalized optimal discrimination objective in the surviving photonic sector, one has
\begin{equation}
\begin{aligned}
  P_{\mathrm{ph}}
  &\leq
  \bar\tau\bigl(1+C_G(\eta)\bigr),\\
  \Pval
  &\leq
  \frac{
    \bar\tau\bigl(1+C_G(\eta)\bigr)
    +1-\min_a w_a
  }{d}.
\end{aligned}
\label{eq:main-na-loss}
\end{equation}
This group-independent bound is proved in Appendix~\ref{sec:namdl}.  It is attained if there exists a common phase $\vartheta$ such that $e^{-i\vartheta}\widehat{\eta}(\lambda)\succeq0$ for every 
$\lambda\in\widehat{G}$.

What does not transfer directly is equally important.  For general mixed states, no closed Fourier-block expression for $C_G(\rho)$ is established here; general covariant channels need not scale the resource by a single contraction parameter; and the Abelian shift-profile formula for quadratic coherent errors has no immediate scalar non-Abelian counterpart.  These limitations and the resulting open problems are stated precisely in Appendix~\ref{sec:naopen}.

% -----------------------------------
\section{Photonic Implementation and Checks}
\label{sec:photonic_impl_check}
The operational measurement of Eq.~\eqref{eq:main-optimal-povm}
maps directly onto a photonic response-register analyzer. Encode the
response in $d$ spatial modes and realize each fixed translation $T_a$
as a permutation of those modes. For $A\cong\mathbb Z_d$, these
permutations are cyclic. The construction implements the response
analyzer and individual permutations; it assumes the controlled
oracle $O_f$ on the joint query--response register. The physical query
encoding and coherent control are outside the present construction.

Optimal value readout applies $V^\dagger$ before detection in the value
basis, as established in Proposition~\ref{prop:mesh}.
For the cyclic case, with $F_d\ket\chi_{\mathrm{lab}}=\ket\chi$,
\begin{equation}
V^\dagger=F_d\operatorname{diag}_\chi(e^{-i\theta_\chi})F_d^\dagger.
\label{eq:main-photonic-unitary}
\end{equation}
The chronological stages are $F_d^\dagger$, the character-phase mask,
and $F_d$. Equivalently, a universal Reck or Clements mesh synthesizes
the complete unitary using $d(d-1)/2$ Mach--Zehnder interferometers
and phase settings \cite{Reck:1994ujl,Clements:2016rkw}.
Holding the Fourier phases fixed leaves the analyzer optimal while
the preparation weights vary.

\begin{figure*}[t]
\centering
\includegraphics[width=.90\linewidth]{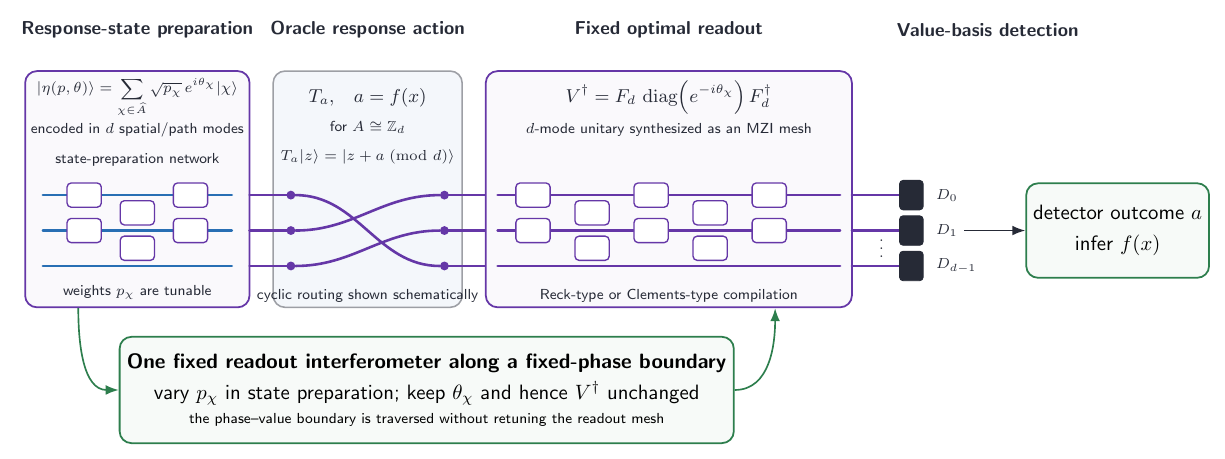}
\caption{Photonic response analyzer for $A\cong\mathbb Z_d$.
After the oracle translates the response, a universal MZI mesh
implements $V^\dagger$ and value-basis detection estimates the value.
A fixed Fourier-phase pattern permits variation of the preparation
weights without retuning the analyzer. The oracle block represents
an assumed controlled transformation; its physical query encoding and
coherent control are outside the analyzer construction.}
\label{fig:photonic}
\end{figure*}

Independent parameter errors generally induce correlated effective
generator entries. The full quadratic expression in
Eq.~\eqref{eq:generalsecond} retains these correlations, whereas the
reduced profile uses
\begin{equation}
V(w)=\sum_{a\in A}\mathbb E|G_{a,a+w}|^2.
\label{eq:main-photonic-profile}
\end{equation}
We compare five descriptions: direct finite-error simulation, the
full correlated quadratic coefficient, the reduced shift-profile
coefficient, the off-diagonal scalar approximation, and the all-shift
isotropic approximation. On the upper phase--value boundary, the
profile and off-diagonal scalar coincide exactly by
Eq.~\eqref{eq:h1boundaryscalar}; this does not assert equality with
the full correlated coefficient without further assumptions.

The profile depends on the mesh settings and fluctuating parameters.
Table~\ref{tab:profiles} gives the identity configuration and three
error mechanisms for one nontrivial $d=8$ analyzer. The nontrivial
analyzer is the same one used in Tables~\ref{tab:reck} and
\ref{tab:d8scan} and Fig.~\ref{fig:quadraticscaling}.
The profiles are evaluated from analytic tangent generators rather
than estimated from finite Monte Carlo samples. Non-flatness alone
does not preclude the boundary scalar reduction.

The triangular meshes use the explicit two-mode blocks and noise
coordinates specified in Appendix~\ref{app:numreck}.
Independent Gaussian perturbations act on effective mixing angles,
phase settings, and output phases. Full quadratic coefficients are
computed directly from analytic parameter derivatives. Direct
simulations use $3000$ independent antithetic Gaussian pairs,
or $6000$ unitary evaluations, at $\sigma=10^{-3}$.
All states at a given dimension share the same programmed analyzer
and the same Gaussian vectors.

The direct means lie within $\MeshMaxSE$ Monte Carlo standard errors
of the ensemble quadratic predictions. The reduced profile differs
from the full coefficient by at most $\MeshMaxProfileDeviation\%$
for these examples. On boundary states it equals the off-diagonal
scalar exactly, while both are computed separately for the explicit
skewed states. The all-shift isotropic normalization overestimates
the nonzero full coefficients by
$\MeshIsoMinBias\%$--$\MeshIsoMaxBias\%$.
These comparisons concern the specified configurations and noise
coordinates; they are not universal device-error bounds.

For the $d=8$, $\Phi=0.50$ state on the same mesh,
\begin{equation}
-\mathbb E[\Delta\Pval]
=c\sigma^2+O(\sigma^4),\qquad c=\MeshQuadraticC.
\label{eq:main-quadratic-scan}
\end{equation}
The coefficient is evaluated independently from the tangent generators.
Sign symmetry of the joint Gaussian distribution removes odd powers
in the average. Figure~\ref{fig:quadraticscaling} tests the quadratic
regime and the onset of higher-order corrections. The scale scan uses
a separate set of Gaussian vectors, shared across error strengths.
Appendix~\ref{sec:photonic_numerical} gives the configurations,
sampling conventions, independently solved discrimination and
robustness SDPs, and measured primal--dual feasibility diagnostics.

\begin{figure}[t]
\centering
\includegraphics[width=.94\columnwidth]{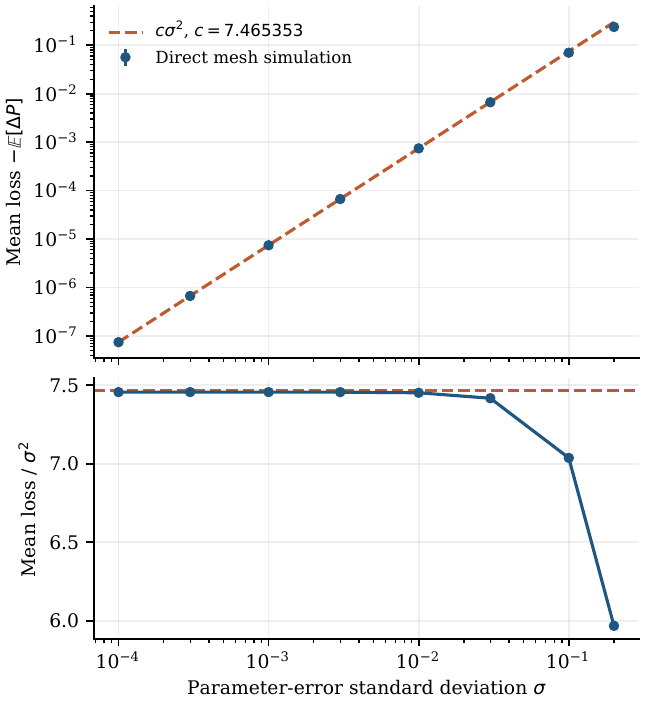}
\caption{Coherent-error scale scan for $d=8$, $\Phi=0.50$, using
the same analyzer as Tables~\ref{tab:reck} and \ref{tab:d8scan}.
The ensemble quadratic coefficient is $c=\MeshQuadraticC$.
Each scale uses $25000$ independent antithetic Gaussian pairs
($50000$ unitary evaluations), with the same vectors rescaled across
$\sigma$. Error bars are one standard error of the pair averages.
The ratio approaches the quadratic prediction at small $\sigma$
and decreases when higher-order corrections become appreciable.}
\label{fig:quadraticscaling}
\end{figure}

The construction illustrates three implementation consequences.
One fixed Fourier-phase configuration serves the entire weight-varying
boundary family. Coherent analyzer errors have no linear correction,
but their quadratic coefficient depends on the device and error model.
On the boundary family, the reduced profile has an exact off-diagonal
scalar form whose accuracy relative to the full model can be assessed
using the retained apparatus correlations.

\section{Conclusion}
\label{sec:conclusion}
Oracle value readout matters whenever information encoded by a black-box
transformation must be recovered explicitly for composition, verification,
or experimental benchmarking. The response-register model places this
task and phase kickback within one controlled interaction. For finite
Abelian groups, known discrimination and coherence results give the exact
pure-state readout probability and the complete interval at fixed
target-character fidelity. Tracing out the response identifies the
induced mixture of phase-oracle channels. For a surjective function,
$1-\Phi$ equals the target-channel distance measured by half the diamond
norm, giving the interval a channel information--disturbance interpretation.
The repeated-probe formula describes collective value readout from
independently prepared copies with a common unknown value.

The robustness analysis identifies when the scalar coherence resource
remains sufficient. Coherence-scaling noise gives a linear corollary,
whereas coherent analyzer errors depend at second order on apparatus
correlations. On the upper-boundary family, the reduced profile has
an exact off-diagonal scalar form, even for non-flat profiles; its
agreement with the full coefficient depends on the device and error
model. Mode-dependent loss admits a separate survival--vacuum
decomposition. Phase--loss compatibility attains the universal
certificate, and the vacuum likelihood can improve discrimination
at fixed mean transmission within that compatible family.
The numerical analysis uses explicitly specified meshes and retained
solver certificates to test these distinctions.

For finite non-Abelian groups, the known matrix-Fourier trace-norm
formula admits a Plancherel-divergence interpretation, and transitivity
preserves the loss-bound mechanism. Scalar phase kickback remains
restricted to one-dimensional representations. Further work can address
closed mixed-state Fourier-block evaluations, complete loss-bound
equality conditions, and device-specific correlation models. The
photonic construction supplies the response analyzer while assuming
the controlled query--response oracle.

\begin{acknowledgments}
The authors thank Saleh Rahimi-Keshari for his careful reading of an earlier version of this manuscript and for his insightful comments and suggestions. His guidance to relevant literature and constructive feedback greatly improved both the presentation and the positioning of this work within the existing body of research.
HS acknowledges the QIntern 2026 summer internship program for the opportunity to contribute to this work, and thanks Pawe\l{} Gora for his mentorship throughout the project.
\end{acknowledgments}

% \FloatBarrier

\bibliography{clean_bib}

\FloatBarrier

% \clearpage

\appendix
% ===================
% Appendices
% ===================
\onecolumngrid

% -----------------------------------
\section{Abelian Readout Resource and Complementarity}\label{sec:abelian_readout_resource_complementarity}
% -----------------------------------
% -----------------------------------
\subsection{Setting and the value-readout functional}
\label{sec:functional}
\label{sec:setting}
\label{sec:readoutfunctional}
% -----------------------------------

Let $A$ be a finite Abelian group with $|A| = d > 1$ and $\widehat{A}$ its character group.  A response-register oracle acts as
\begin{equation}
O_f : \ket{x}\ket{\eta} \longmapsto \ket{x}\, T_{f(x)}\ket{\eta},
\qquad
T_a\ket{z} = \ket{z+a},
\label{eq:oracle}
\end{equation}
where $\ket{\eta}$ is the \emph{response state}.  We use throughout the character basis
\begin{equation}
\ket{\chi} := \frac{1}{\sqrt{d}}\sum_{z \in A}
\overline{\chi(z)}\,\ket{z},
\end{equation}
so that
\begin{equation}
T_a\ket{\chi} = \chi(a)\ket{\chi},
\label{eq:charbasis}
\end{equation}
and write a pure response state as
\begin{equation}
\ket{\eta} = \sum_{\chi \in \widehat{A}}
\sqrt{p_\chi}\, e^{i\theta_\chi} \ket{\chi},
\qquad p_\chi \ge 0, \qquad \sum_\chi p_\chi = 1 .
\label{eq:state}
\end{equation}

For the canonical value-readout task, fix $x$ and suppose that the unknown value $a=f(x)$ is uniformly distributed over $A$.  Recovering $f(x)$ from one query is then the problem of identifying $a$ from one copy of $T_a\ket{\eta}$, i.e.\ of discriminating the orbit ensemble
\begin{equation}
\mathcal{E}_\eta
= \left\{ \tfrac1d,\ T_a\ket{\eta} \right\}_{a \in A} .
\label{eq:ensemble}
\end{equation}
We denote the optimal success probability for this discrimination problem by $\Pval$ and call it the \emph{value-readout probability}.

The complementary use of the same oracle is the phase-oracle mode. Preparing the response register in a fixed character state $\ket{\chi_0}$ makes the oracle contribute only the phase $\chi_0(f(x))$.  For a general pure response state \eqref{eq:state}, we define its \emph{target-character fidelity} by
\begin{equation}
\Phi := \bigl|\braket{\chi_0}{\eta}\bigr|^2 = p_{\chi_0},
\end{equation}
which quantifies its fidelity to the ideal response state for that phase-oracle mode.

The ensemble \eqref{eq:ensemble} is geometrically uniform, and its optimal discrimination probability is known.  We record it here and claim no priority for it.

\begin{proposition}[Known]
\label{prop:known}
For the ensemble \eqref{eq:ensemble} with $\ket{\eta}$ as in
Eq.~\eqref{eq:state},
\begin{equation}
\Pval
= \frac{1}{d}\Bigg( \sum_{\chi \in \widehat{A}} \sqrt{p_\chi}\Bigg)^{2},
\label{eq:Pval}
\end{equation}
attained by the covariant POVM $M_a = T_a \widetilde{M}_0 T_a^{\dagger}$ with rank-one seed $\widetilde{M}_0 = \tfrac1d\ket{\mu}\!\bra{\mu}$, $\ket{\mu} = \sum_\chi e^{i\theta_\chi}\ket{\chi}$. For any $\chi$ with $p_\chi=0$, the corresponding phase $\theta_\chi$ in this construction may be chosen arbitrarily. 
\end{proposition}

Equation~\eqref{eq:Pval} follows from the optimality of the square-root measurement for Abelian geometrically uniform ensembles \cite{Eldar:2000dav,ban1997optimum,Krovi:2015bah,Zhou:2025iac} together with the Gram spectrum computed in Appendix~\ref{app:gram}; the same expression appears in the complementarity relations of Bagan \emph{et al.} \cite{Bagan:2018uuk} for finite-group actions, and, for a related ensemble, as the pretty-good-measurement success probability of Bacon, Childs and van Dam \cite{Bacon:2005cri}.  It is a specialization of the covariant maximum-likelihood solution of Chiribella \emph{et al.} \cite{Chiribella:2004wyf}; we return to that connection in Appendix~\ref{sec:nonabelian}, where it supplies the non-Abelian case.

Two structural features of Eq.~\eqref{eq:Pval} drive everything that follows.  First, $\Pval$ depends only on the Fourier weights $p_\chi$ and not on the phases $\theta_\chi$ (Appendix~\ref{app:gram}).  Second, the displayed optimal seed depends only on the phases and not on the Fourier weights.  The optimal POVM need not be unique, particularly when some Fourier weights vanish.  The phase dependence of this optimal construction has direct experimental consequences developed in Appendix~\ref{sec:coherent}.

% -----------------------------------
\subsection{Coherence as the Readout Resource}
\label{sec:resource}
\label{sec:defC}
\label{sec:props}
% -----------------------------------

For a state $\rho$ on the response register, let
\begin{equation}
  C(\rho) := \sum_{\chi \neq \chi'}
  \left| \bra{\chi}\rho\ket{\chi'} \right|
  \label{eq:Cdef}
\end{equation}
denote the $\ell_1$-norm coherence of $\rho$ with respect to the character basis \cite{Baumgratz:2013ecx}.  For a pure state \eqref{eq:state}, we write $C(\eta) := C(\ket{\eta}\!\bra{\eta})$.  Since $\bra{\chi}\eta\rangle\!\langle\eta\ket{\chi'} = \sqrt{p_\chi p_{\chi'}}\,e^{i(\theta_\chi-\theta_{\chi'})}$, we have
\begin{equation}
  C(\eta) = \Bigg(\sum_\chi \sqrt{p_\chi}\Bigg)^{2} - 1 ,
  \label{eq:Cpure}
\end{equation}
and Proposition~\ref{prop:known} takes the form
\begin{equation}
  \boxed{\ \Pval = \frac{1 + C(\eta)}{d}\ }
  \label{eq:master}
\end{equation}

Equation~\eqref{eq:master} is the organizing identity of this paper. For pure response states, all readout capability above the random-guessing baseline $1/d$ is supplied by $C(\eta)$, in exact proportion.  More generally, a response state with no character-basis coherence achieves exactly the random-guessing probability, since its translated density operators are identical.  The corresponding mixed-state bound and its equality conditions are established in Theorem~\ref{thm:mixed}.

It is also useful to record the relation to R\'enyi entropy.  With base-$2$ logarithms, the R\'enyi entropy of order $\tfrac12$ is \cite{renyi1961measures,Tomamichel:2015gtd}
\begin{equation}
H_{1/2}(p)
= 2\log_2\Bigg(\sum_\chi\sqrt{p_\chi}\Bigg),
\end{equation}
and hence
\begin{equation}
1 + C(\eta) = 2^{H_{1/2}(p)} =: S_{\mathrm{eff}}(p) .
\label{eq:Seffrel}
\end{equation}
We call $S_{\mathrm{eff}}(p)$ the \emph{effective Fourier support} of the response state.  This is an entropy-based effective number, not the literal support cardinality: if $r = |\{\chi:p_\chi>0\}|$, then $S_{\mathrm{eff}}(p)\le r$, with equality if and only if $p$ is uniform on its support.  We use $C$ throughout because it gives the affine readout identity \eqref{eq:master} and governs several of the error laws below.

Thus the optimal value-readout probability is exactly the normalized R\'{e}nyi-$\frac{1}{2}$ effective support size of the Fourier-weight distribution. The joint operational consequences of this identity and the phase--value complementarity relation are illustrated in Fig.~\ref{fig:phase-value complementarity} for $d=16$. The color scale resolves the effective Fourier support $S_{\mathrm{eff}}=2^{H_{1/2}(p)}=dP_{\mathrm{val}}$, while the boundary curves delimit the achievable value-readout performance at fixed target-character fidelity.

\begin{proposition}
\label{prop:Cprops}
For any response state $\rho$ on a $d$-dimensional register:
\begin{enumerate}
\item $0 \le C(\rho) \le d-1$;

\item $C(\rho) = 0$ if and only if $\rho$ is incoherent in the character basis; for a pure state this means $\ket{\eta} = e^{i\vartheta}\ket{\chi_0}$ for some character $\chi_0$ and global phase $\vartheta$;

\item $C(\rho) = d-1$ if and only if $\rho$ is a maximally coherent pure state in the character basis, i.e.
\begin{equation}
\ket{\eta}
= \frac{1}{\sqrt{d}}\sum_{\chi\in\widehat{A}}
e^{i\theta_\chi}\ket{\chi}
\end{equation}
for arbitrary phases $\theta_\chi$.
\end{enumerate}
For pure response states, Eq.~\eqref{eq:master} therefore gives $1/d \le \Pval \le 1$.  Moreover, $C(\eta)$ is strictly Schur-concave in the Fourier-weight distribution $p$: if $p \succ q$ in the majorization order, then $C(p) \le C(q)$.
\end{proposition}

\begin{proof}
Positivity of $\rho$ gives
\begin{equation}
  |\rho_{\chi\chi'}|
  \leq
  \sqrt{
    \rho_{\chi\chi}\rho_{\chi'\chi'}
  }.
\end{equation}
Hence
\begin{equation}
  1+C(\rho)
  =
  \sum_{\chi,\chi'}
  |\rho_{\chi\chi'}|
  \leq
  \left(
    \sum_\chi\sqrt{\rho_{\chi\chi}}
  \right)^2
  \leq d,
\end{equation}
where the last inequality is Cauchy--Schwarz.  This proves
item~(1), while item~(2) follows directly from
Eq.~\eqref{eq:Cdef}.  Equality in the upper bound of item~(1)
requires a uniform diagonal and saturation of
\begin{equation}
  |\rho_{\chi\chi'}|
  \leq
  \sqrt{
    \rho_{\chi\chi}\rho_{\chi'\chi'}
  }
\end{equation}
for every pair, which is equivalent to the rank-one form stated in
item~(3).  Finally, strict Schur-concavity of $C(p)$ follows from the
Hardy--Littlewood--P\'olya theorem
\cite{MarshallOlkinArnold2011}.
\end{proof}

Item (2) deserves emphasis because it is exact and non-perturbative. For a pure state with $C(\eta)=0$, one has $\ket{\eta}=e^{i\vartheta}\ket{\chi_0}$ for some $\chi_0$, and Eq.~\eqref{eq:charbasis} gives $T_a\ket{\eta}=\chi_0(a)\ket{\eta}$.  Thus every orbit state represents the same physical state.  Consequently, for any POVM $\{M_a\}$,
\begin{equation}
\frac{1}{d}\sum_a \bra{\eta} M_a \ket{\eta}
= \frac{1}{d}\bra{\eta}\mathbb{I}\ket{\eta}
= \frac{1}{d}.
\label{eq:cornerexact}
\end{equation}
The exact phase corner is not merely limited by the choice of optimal measurement: it carries no value information at all.  Equation~\eqref{eq:cornerexact} provides a consistency check for the readout-error models below.

% --------------------------------------
\subsection{Composition}
\label{sec:composition}
% --------------------------------------
The readout functional composes multiplicatively for product response states under the canonical uniform value-readout task.

\begin{proposition}[Multiplicativity] 
\label{prop:mult} 
Let $A$ and $B$ be finite Abelian groups of orders $d_A$ and $d_B$, and let the response register for $A\times B$ be prepared in the product state $\ket{\eta_A}\otimes\ket{\eta_B}$.  For the uniform ensemble generated by the product action,
\begin{equation}
\left\{
\frac{1}{d_Ad_B},\
T_a\ket{\eta_A}\otimes T_b\ket{\eta_B}
\right\}_{(a,b)\in A\times B},
\end{equation}
one has
\begin{equation}
1+C(\eta_A\otimes\eta_B)
=
\bigl(1+C(\eta_A)\bigr)
\bigl(1+C(\eta_B)\bigr),
\label{eq:Cmult}
\end{equation}
and
\begin{equation}
\Pval^{AB}
=
\Pval^{A}\Pval^{B},
\label{eq:mult}
\end{equation}
where $\Pval^{AB}$ is the optimal probability of jointly identifying
the pair $(a,b)$.
\end{proposition}

\begin{proof}
The dual group of a direct product is
$\widehat{A\times B}\cong\widehat A\times\widehat B$, with character
\begin{equation}
  (\chi,\psi)(a,b)=\chi(a)\psi(b)
\end{equation}
\cite{serre_1977,Terras_1999}.  Hence the character weights of the
product state factorize as
$p_{\chi,\psi}=p_\chi q_\psi$, and
\begin{equation}
  \sum_{\chi,\psi}\sqrt{p_\chi q_\psi}
  =
  \left(
    \sum_\chi\sqrt{p_\chi}
  \right)
  \left(
    \sum_\psi\sqrt{q_\psi}
  \right).
\end{equation}
Equation~\eqref{eq:Cmult} follows from Eq.~\eqref{eq:Cpure}, and
Eq.~\eqref{eq:mult} follows from Eq.~\eqref{eq:master} together with
$|A\times B|=d_Ad_B$.
\end{proof}

Equivalently, Eq.~\eqref{eq:Cmult} expresses the additivity of $H_{1/2}$ on product distributions \cite{renyi1961measures,Tomamichel:2015gtd}.  The product of optimal measurements for the $A$ and $B$ orbit ensembles achieves the success probability $\Pval^{A}\Pval^{B}$, while Proposition~\ref{prop:mult} shows that this equals the optimum over all POVMs on the joint register.  Thus entangled measurements provide no additional value-readout advantage for these product response states under the uniform product task.  
\paragraph{Repeated probes of one common value.}
Consider $k\geq1$ independently prepared copies of the pure response,
with hypotheses $(T_a\ket\eta)^{\otimes k}$ and uniform prior on the
single label $a\in A$. This is the diagonal action $T_a^{\otimes k}$,
rather than the independent-label product action of
Proposition~\ref{prop:mult}. Define
\begin{equation}
p^{*k}(\chi):=
\sum_{\chi_1\cdots\chi_k=\chi}
\prod_{j=1}^k p_{\chi_j}.
\label{eq:h1convolution}
\end{equation}
The optimal collective readout probability is
\begin{equation}
\Pval^{(k)}=\frac1d
\left(\sum_{\chi\in\widehat A}\sqrt{p^{*k}(\chi)}\right)^2.
\label{eq:h1repeated}
\end{equation}
To see this, project $\ket\eta^{\otimes k}$ onto the total-character
sectors. Their orthogonal projections $\ket{v_\chi}$ have squared
norms $p^{*k}(\chi)$ and obey
$T_a^{\otimes k}\ket{v_\chi}=\chi(a)\ket{v_\chi}$.
Although the representation has multiplicities, the orbit uses only
one vector per occupied character sector. The pure-state covariant
discrimination argument on this span gives Eq.~\eqref{eq:h1repeated};
the measurement can be completed arbitrarily on the orthogonal
complement. A character-state response has a point mass at
$\chi_0^k$, so its success remains $1/d$ for every $k$.
The statement fixes product preparation and is distinct from applying
$T_a^k$ to one register or optimizing general multi-query algorithms.

The upper coherence bound in the following theorem is
the bound established in Lemma~2 of
Ref.~\cite{Piani:2016lse}, applied in the character basis.
Specifically, the selected diagonal entry in that lemma
is here $\Phi=|\langle\chi_0|\eta\rangle|^2$.
Together with Eq.~\eqref{eq:master}, this gives the upper
bound on oracle value readout.
We present this bound alongside the lower boundary,
the pure-state equality conditions, and the attainability
of the full interval.
% --------------------------------------
\subsection{Phase--value complementarity}
\label{sec:tradeoff}
% --------------------------------------
Fix a target character $\chi_0$ and recall 
\begin{equation}
\Phi=p_{\chi_0}=|\braket{\chi_0}{\eta}|^2.  
\end{equation}

\begin{theorem}[Phase--value tradeoff]
\label{thm:tradeoff}
For any pure response state with target-character fidelity
$\Phi\in[0,1]$, the character-basis coherence satisfies
\begin{equation}
C(\eta)
\;\geq\;
C_{\min}(\Phi)
:=
\left[
\sqrt{\Phi}+\sqrt{1-\Phi}
\right]^2-1,
\label{eq:Cmin}
\end{equation}
and
\begin{equation}
C(\eta)
\;\leq\;
C_{\max}(\Phi)
:=
\left[
\sqrt{\Phi}
+\sqrt{(d-1)(1-\Phi)}
\right]^2-1.
\label{eq:Cmax}
\end{equation}
Equivalently, the optimal value-readout probability obeys
\begin{equation}
\Pval
\;\geq\;
\frac{1}{d}
\left[
\sqrt{\Phi}+\sqrt{1-\Phi}
\right]^2,
\label{eq:tradeofflower}
\end{equation}
and
\begin{equation}
\Pval
\;\leq\;
\frac{1}{d}
\left[
\sqrt{\Phi}
+\sqrt{(d-1)(1-\Phi)}
\right]^2.
\label{eq:tradeoff}
\end{equation}

The upper bounds in Eqs.~\eqref{eq:Cmax} and
\eqref{eq:tradeoff} are attained if and only if the non-target
weights are uniform,
\begin{equation}
p_\chi
=
\frac{1-\Phi}{d-1}
\qquad
\text{for every }\chi\neq\chi_0.
\label{eq:upperweights}
\end{equation}
For each $\Phi$, they are attained by the family
\begin{equation}
\ket{\eta_\Phi^{\max}}
=
\sqrt{\Phi}\,\ket{\chi_0}
+
\sqrt{\frac{1-\Phi}{d-1}}
\sum_{\chi\neq\chi_0}
e^{i\theta_\chi}\ket{\chi},
\label{eq:satfamily}
\end{equation}
with arbitrary phases $\theta_\chi$.

The lower bounds in Eqs.~\eqref{eq:Cmin} and
\eqref{eq:tradeofflower} are attained if and only if at most one
non-target weight is nonzero.  Equivalently, for $\Phi<1$, they are
attained by
\begin{equation}
\ket{\eta_{\Phi,\chi_1}^{\min}}
=
\sqrt{\Phi}\,\ket{\chi_0}
+
\sqrt{1-\Phi}\,
e^{i\vartheta}\ket{\chi_1},
\qquad
\chi_1\neq\chi_0,
\label{eq:lowerfamily}
\end{equation}
with arbitrary phase $\vartheta$.  Every value between the lower and
upper bounds is achievable.  Hence
Eqs.~\eqref{eq:tradeofflower}--\eqref{eq:tradeoff} describe the
complete phase--value region at fixed $\Phi$.
\end{theorem}

\begin{proof}
For completeness, we give a direct proof of the known
upper bound in our notation, followed by the lower bound
and the interval-attainability argument.

At fixed target-character fidelity
$p_{\chi_0}=\Phi$, write
\begin{equation}
\sum_\chi\sqrt{p_\chi}
=
\sqrt{\Phi}
+
\sum_{\chi\neq\chi_0}\sqrt{p_\chi},
\qquad
\sum_{\chi\neq\chi_0}p_\chi=1-\Phi.
\label{eq:fixedphisum}
\end{equation}

For the upper bound, Cauchy--Schwarz applied to the $d-1$
non-target terms gives
\begin{equation}
\sum_{\chi\neq\chi_0}\sqrt{p_\chi}
\leq
\sqrt{(d-1)(1-\Phi)}.
\label{eq:upperCS}
\end{equation}
Equality holds if and only if
$(\sqrt{p_\chi})_{\chi\neq\chi_0}$ is proportional to the
all-ones vector, which is equivalent to
Eq.~\eqref{eq:upperweights}.

For the lower bound, nonnegativity of the cross terms gives
\begin{equation}
\begin{split}
\left(
\sum_{\chi\neq\chi_0}\sqrt{p_\chi}
\right)^2
&=
\sum_{\chi\neq\chi_0}p_\chi
+
2
\sum_{\substack{\chi<\chi'\\
                 \chi,\chi'\neq\chi_0}}
\sqrt{p_\chi p_{\chi'}} \\
&\geq
1-\Phi.
\end{split}
\label{eq:lowercross}
\end{equation}
Therefore,
\begin{equation}
\sum_{\chi\neq\chi_0}\sqrt{p_\chi}
\geq
\sqrt{1-\Phi}.
\label{eq:lowerrootsum}
\end{equation}
Equality holds if and only if every cross term in
Eq.~\eqref{eq:lowercross} vanishes, i.e.\ if and only if at most one
non-target weight is nonzero.

Substituting Eqs.~\eqref{eq:upperCS} and
\eqref{eq:lowerrootsum} into Eq.~\eqref{eq:fixedphisum}, and then
using
Eqs.~\eqref{eq:Cpure} and \eqref{eq:master}, gives
Eqs.~\eqref{eq:Cmin}--\eqref{eq:tradeoff} and their equality
conditions.  The phases do not enter because both $C(\eta)$ and
$\Pval$ depend only on the Fourier weights.

It remains to show that no values inside the stated interval are
missing.  For $d\geq3$, the non-target weights form the connected
simplex
\begin{equation}
\left\{
(p_\chi)_{\chi\neq\chi_0}:
p_\chi\geq0,\quad
\sum_{\chi\neq\chi_0}p_\chi=1-\Phi
\right\}.
\end{equation}
The map
\begin{equation}
(p_\chi)_{\chi\neq\chi_0}
\longmapsto
\frac1d
\left(
\sqrt{\Phi}
+
\sum_{\chi\neq\chi_0}\sqrt{p_\chi}
\right)^2
\end{equation}
is continuous.  Its image is therefore a connected subset of
$\mathbb{R}$ containing both its attained minimum and maximum, and
hence is the entire closed interval between them.  For $d=2$, the
non-target simplex consists of a single point and the lower and upper
bounds coincide.
\end{proof}

Two points on the boundary represent the extremal operational modes of the same oracle.  At $\Phi=1$, one has $C(\eta)=0$ and $\Pval=1/d$: the response state has exact target-character phase semantics and carries no value information.  At $\Phi=1/d$, all character weights are uniform, so $C(\eta)=d-1$ and $\Pval=1$: value readout is perfect and no character is distinguished by its weight.

The upper bound in Eq.~\eqref{eq:tradeoff} is attained for every
$0\leq\Phi\leq1$, while $1/d\leq\Phi\leq1$ is the Pareto-relevant
phase--value branch. Indeed, for $\Phi<1/d$, increasing $\Phi$ toward $1/d$ increases both the target-character fidelity and the maximal value-readout probability. Because the phases in Eq.~\eqref{eq:satfamily} are arbitrary, the boundary can be traversed by redistributing the Fourier weights while holding any chosen phase pattern fixed.

% -----------------------------------
\paragraph{Grover-Type Response State.}
% -----------------------------------

For $A=\mathbb{Z}_2^2$ ($d=4$) with Grover target character $\chi_0=\chi_{11}$ and predicate $(-1)^{f(x)}=\chi_{11}(F(x))$, the saturating state at phase fidelity $\Phi$ is 
\begin{equation}
\ket{\eta_\Phi}=
\sqrt{\Phi}\ket{\chi_{11}}+\sqrt{\frac{1-\Phi}{3}}\sum_{u\neq(1,1)}\ket{\chi_u},
\end{equation}
giving $\Pval=[\sqrt{\Phi}+\sqrt{3(1-\Phi)}]^2/4$. 

For the saturating state, when $\Phi=0.95$ one obtains $\Pval=0.464$.  For response dimension $d=16$, the same $5\%$ phase infidelity gives the maximum attainable value-readout probability $\Pval^{\max}(0.95)=0.212$, compared with the random-guessing baseline $1/16=6.25\%$.

\begin{figure}[h]
\includegraphics[width=0.6\linewidth]{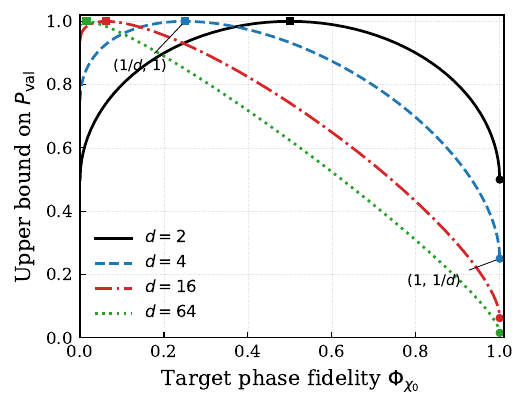}
\caption{Tight phase--value tradeoff, Eq.~\eqref{eq:tradeoff}, for $d\in\{2,4,16,64\}$. The horizontal axis is the target-character phase fidelity $\Phi_{\chi_0}$; the vertical axis is the optimal single-query value-readout probability $\Pval$. Each curve is saturated by the response states in Eq.~\eqref{eq:satfamily}. Exact phase semantics occurs at $\Phi_{\chi_0}=1$, where $\Pval=1/d$ (bottom right); exact value readout occurs at $\Phi_{\chi_0}=1/d$, where $\Pval=1$ (marked on each curve). Near the phase corner, $\Pval-1/d$ grows as $O(\sqrt{\epsilon})$ in the phase infidelity $\epsilon=1-\Phi_{\chi_0}$.}
\label{fig:tradeoff}
\end{figure}

% -----------------------------------
\paragraph{Small-perturbation Scaling.}
% -----------------------------------

Writing $\Phi_{\chi_0}\equiv\Phi$ and $\epsilon:=1-\Phi$,
expansion of Eq.~\eqref{eq:tradeoff} at small $\epsilon$ gives

\begin{equation}
\Pval-\frac1d 
\leq
\frac{2 \sqrt{d-1}}{d}\sqrt{\epsilon} + \frac{d-2}{d}\epsilon +\mathcal{O}(\epsilon^{3/2}).
\end{equation}
The dominant correction near the phase corner is $\Pval-\frac1d=\mathcal{O}(\sqrt{\epsilon})$, with sharp upper-bound coefficient $\tfrac{2 \sqrt{d-1}}{d}$.  %This square-root behavior is characteristic of perturbative distinguishability bounds and is saturated by the family \eqref{eq:satfamily}. 
For example, for $d = 16$ and $\epsilon = 0.01$ the expansion gives $\Pval \leq 0.119$, compared with random guessing $1/d = 0.0625$; for $\epsilon = 0.10$ the exact bound gives $\Pval \leq0.295$. At fixed small phase infidelity, increasing the response dimension lowers the maximum attainable value-readout probability and therefore sharpens the absolute separation between the phase and value modes.

\subsection{Induced query channel and its overlap kernel}
\label{app:h1querychannel}
\label{sec:query-channel}
The character convention gives
$O_f=\sum_\chi U_\chi^f\otimes\ket\chi\!\bra\chi$.
For a response state $\rho_R$, put
$p_\chi=\bra\chi\rho_R\ket\chi$. Tracing out the response removes
its off-diagonal character-basis terms, so
\begin{equation}
\mathcal E_f(\sigma)=\sum_\chi p_\chi
U_\chi^f\sigma U_\chi^{f\dagger}.
\label{eq:h1channelmixture}
\end{equation}
For a pure response, these are the weights in Eq.~\eqref{eq:state}.
Let us define $\delta_f:=\tfrac12\|\mathcal E_f-\mathcal U_{\chi_0}^f\|_\diamond$.
Convexity of the diamond norm and the bound of two on the diamond
distance between channels imply
\begin{equation}
\frac12\|\mathcal E_f-\mathcal U_{\chi_0}^f\|_\diamond
\leq\sum_{\chi\ne\chi_0}p_\chi=1-\Phi.
\label{eq:h1channelbound}
\end{equation}
If $f$ is surjective, choose $x_a$ with $f(x_a)=a$ for every $a$.
On $\ket{s}=d^{-1/2}\sum_a\ket{x_a}$, the states
$U_\chi^f\ket{s}$ are orthonormal by character orthogonality.
The output mixture is therefore at trace distance $1-\Phi$ from
the target output, proving equality in Eq.~\eqref{eq:h1channelbound}.
Surjectivity is sufficient, not asserted necessary. For constant $f$,
all the unitary channels coincide and the distance is zero.

\begin{corollary}[Information--disturbance interval for surjective functions]
\label{cor:h1informationdisturbance}
For a pure response and a surjective $f$, the achievable readout
probabilities at fixed $\delta_f$ satisfy
\begin{equation}
\frac{\bigl(\sqrt{1-\delta_f}+\sqrt{\delta_f}\bigr)^2}{d}
\leq\Pval\leq
\frac{\bigl(\sqrt{1-\delta_f}+\sqrt{(d-1)\delta_f}\bigr)^2}{d}.
\label{eq:h1informationdisturbance}
\end{equation}
Both endpoints, and every intermediate value, are attainable.
\end{corollary}
\begin{proof}
Surjectivity gives $\delta_f=1-\Phi$. Substitute this equality in
Theorem~\ref{thm:tradeoff}; its equality conditions and interval
attainability prove the claim.
\end{proof}
Here $\Pval$ measures optimal guessing of the uniform value label
from the response, while $\delta_f$ measures the unconditioned query
channel's deviation from a prescribed phase oracle. This is a relation
for the specified controlled-translation interaction, not a theorem
for arbitrary instruments. For a non-surjective $f$, the upper bound
on $\delta_f$ alone does not justify substituting $\Phi=1-\delta_f$.
The Pareto-relevant branch is $0\leq\delta_f\leq1-1/d$.

The common overlap kernel is
\begin{equation}
K(g):=\Tr(\rho_RT_g)=\sum_\chi p_\chi\chi(g).
\label{eq:h1channelkernel}
\end{equation}
For a pure response this is the kernel in Eq.~\eqref{eq:kernel}.
Writing $\sigma_{xx'}=\bra{x}\sigma\ket{x'}$, one obtains
\begin{equation}
\bra{x}\mathcal E_f(\sigma)\ket{x'}
=K\bigl(f(x)-f(x')\bigr)\sigma_{xx'}.
\label{eq:h1kernelchannelentries}
\end{equation}
The sign agrees with
$\langle T_{f(x')}\eta|T_{f(x)}\eta\rangle
=K(f(x)-f(x'))$. For a pure query
$\ket\psi=\sum_x\alpha_x\ket{x}$, the output coherence in the query basis is
\begin{equation}
C_X\bigl(\mathcal E_f(\ket\psi\!\bra\psi)\bigr)
=\sum_{x\ne x'}|\alpha_x\alpha_{x'}|
\left|K\bigl(f(x)-f(x')\bigr)\right|,
\label{eq:h1querycoherence}
\end{equation}
whereas Fourier inversion gives
\begin{equation}
\Phi=p_{\chi_0}
=\frac1d\sum_{g\in A}K(g)\overline{\chi_0(g)}.
\label{eq:h1kernelfidelity}
\end{equation}
For a uniform query over one representative of every group value,
Eq.~\eqref{eq:h1querycoherence} reduces to
$\sum_{g\ne0}|K(g)|$. The normalized path-coherence measure of
Ref.~\cite{bagan2016relations} divides this by the number of paths.
For general amplitudes or repeated function values, the weights in
Eq.~\eqref{eq:h1querycoherence} must be retained. Thus the two
descriptions use different functionals of the same overlap kernel.

\subsection{Orbit Geometry and Proof of the Readout Formula}
\label{app:gram}
% ==================================================

This subsection collects the finite-Abelian-group Fourier facts underlying Proposition~\ref{prop:known} and the mixed-state analysis of Theorem~\ref{thm:mixed}.

Throughout, $A$ is a finite Abelian group with $|A|=d$, the character basis is that of Eq.~\eqref{eq:charbasis}, and we use the orthogonality and completeness relations \cite{serre_1977,Terras_1999,diaconis1988group}
\begin{equation}
\sum_{a\in A}\chi(a)\,\overline{\chi'(a)}
=d\,\delta_{\chi,\chi'},
\qquad
\sum_{\chi\in\widehat A}\overline{\chi(z)}
=d\,\delta_{z,0}.
\label{eq:orth}
\end{equation}

% --------------------------------------
\subsubsection{Overlaps depend only on the Fourier weights}
\label{app:overlaps}
% --------------------------------------

Since $T_a^\dagger T_b=T_{b-a}$, the overlap of two orbit states is
\begin{equation}
  \braket{T_a\eta}{T_b\eta}
  =
  \bra{\eta}T_{b-a}\ket{\eta}.
\end{equation}
Inserting the expansion \eqref{eq:state} and using
$T_g\ket{\chi'}=\chi'(g)\ket{\chi'}$ together with orthonormality of
the character basis gives
\begin{equation}
\begin{split}
  \bra{\eta}T_g\ket{\eta}
  &=
  \sum_{\chi,\chi'}
  \sqrt{p_\chi p_{\chi'}}\,
  e^{-i\theta_\chi}e^{i\theta_{\chi'}}\,
  \chi'(g)\,\braket{\chi}{\chi'} \\
  &=
  \sum_{\chi\in\widehat A}p_\chi\chi(g)
  =:K(g).
\end{split}
  \label{eq:kernel}
\end{equation}
The phases cancel identically.  Defining the \emph{overlap kernel}
$K$ by Eq.~\eqref{eq:kernel}, every orbit overlap is
\begin{equation}
  \braket{T_a\eta}{T_b\eta}
  =
  K(b-a).
  \label{eq:overlapK}
\end{equation}
Thus the geometry of the pure-state orbit is the Fourier synthesis of
the weight sequence
$p=(p_\chi)_{\chi\in\widehat A}$ and is independent of the phases
$\theta_\chi$.  For pure-state ensembles with fixed prior, the optimal
discrimination probability is determined by the Gram matrix.
Equation~\eqref{eq:overlapK} therefore establishes the phase
independence of $\Pval$ asserted in
Appendix~\ref{sec:readoutfunctional}.

% --------------------------------------
\subsubsection{Gram spectrum}
\label{app:spectrum}
% --------------------------------------

By Eq.~\eqref{eq:overlapK}, the Gram matrix
\begin{equation}
  G_{ab}
  :=
  \braket{T_a\eta}{T_b\eta}
  =
  K(b-a)
  \label{eq:gram}
\end{equation}
depends on $b-a$ alone and is therefore translation invariant.  For
$A=\mathbb Z_d$ it is circulant; for a general finite Abelian group it
is the corresponding group-circulant matrix.  Such matrices are
diagonalized by the characters.

\begin{proposition}
\label{prop:gram}
For each $\chi'\in\widehat A$, the vector
\begin{equation}
  v_{\chi'}
  :=
  \bigl(\overline{\chi'(a)}\bigr)_{a\in A}
\end{equation}
is an eigenvector of $G$ with eigenvalue
$\lambda_{\chi'}=dp_{\chi'}$.  The normalized Gram spectrum therefore
coincides with the Fourier-weight distribution:
\begin{equation}
  p_{\chi'}=\frac{\lambda_{\chi'}}d.
\end{equation}
\end{proposition}

\begin{proof}
Using Eq.~\eqref{eq:gram}, substituting $g=b-a$, and using the
multiplicativity of $\chi'$,
\begin{equation}
\begin{split}
  (Gv_{\chi'})_a
  &=
  \sum_{b\in A}
  K(b-a)\,\overline{\chi'(b)} \\
  &=
  \sum_{g\in A}
  K(g)\,\overline{\chi'(a+g)} \\
  &=
  \overline{\chi'(a)}
  \sum_{g\in A}K(g)\,\overline{\chi'(g)}.
\end{split}
\end{equation}
Hence
$Gv_{\chi'}=\lambda_{\chi'}v_{\chi'}$, where
\begin{equation}
  \lambda_{\chi'}
  =
  \sum_{g\in A}K(g)\,\overline{\chi'(g)}.
\end{equation}
Substituting Eq.~\eqref{eq:kernel} and applying the first relation in
Eq.~\eqref{eq:orth},
\begin{equation}
\begin{split}
  \lambda_{\chi'}
  &=
  \sum_\chi p_\chi
  \sum_{g\in A}\chi(g)\,\overline{\chi'(g)} \\
  &=
  \sum_\chi p_\chi d\,\delta_{\chi,\chi'}
  =
  dp_{\chi'}.
\end{split}
\end{equation}
The $d$ vectors $\{v_{\chi'}\}_{\chi'\in\widehat A}$ are mutually
orthogonal by Eq.~\eqref{eq:orth} and therefore form an eigenbasis.
\end{proof}

Proposition~\ref{prop:gram} is the precise sense in which the
Fourier-weight distribution, the overlap kernel, and the normalized
Gram spectrum carry the same information.  This structure generally
fails under mode-dependent loss: the modified Gram matrix
\eqref{eq:noncirc} depends on $a$ and $b-a$ separately, need not be
translation invariant, and need not be diagonalized by the
characters.  Special states or transmission profiles may nevertheless
restore translation invariance accidentally.

% --------------------------------------
\subsubsection{The covariant constraint}
\label{app:completeness}
% --------------------------------------

The ensemble $\mathcal E_\eta$ is covariant under the unitary
representation $a\mapsto T_a$ and has uniform prior.  An optimal
minimum-error measurement may therefore be taken covariant
\cite{holevo1973statistical,holevo2011probabilistic}:
\begin{equation}
  M_a=T_aM_0T_a^\dagger,
  \qquad
  M_0\succeq0,
  \label{eq:covseed}
\end{equation}
where $M_0$ is a single positive semidefinite \emph{seed}.  Expanding
the seed in the character basis,
\begin{equation}
  M_0
  =
  \sum_{\chi,\chi'}
  m_{\chi\chi'}\ket{\chi}\!\bra{\chi'},
\end{equation}
and using $T_a\ket{\chi}=\chi(a)\ket{\chi}$ together with
Eq.~\eqref{eq:orth}, one obtains
\begin{equation}
\begin{split}
  \sum_{a\in A}T_aM_0T_a^\dagger
  &=
  \sum_{\chi,\chi'}
  m_{\chi\chi'}
  \left(
  \sum_{a\in A}
  \chi(a)\overline{\chi'(a)}
  \right)
  \ket{\chi}\!\bra{\chi'} \\
  &=
  d\sum_{\chi\in\widehat A}
  m_{\chi\chi}\ket{\chi}\!\bra{\chi}.
\end{split}
  \label{eq:twirl}
\end{equation}
The completeness condition $\sum_aM_a=\mathbb I$ is therefore
equivalent to
\begin{equation}
  \boxed{\quad
  m_{\chi\chi}=\frac1d
  \qquad\forall\chi\in\widehat A
  \quad}
  \label{eq:diagcond}
\end{equation}
with the off-diagonal entries constrained only by positivity of
$M_0$.

Equation~\eqref{eq:twirl} uses only the left regular representation of $A$ and character orthogonality; it does not depend on the response state. For a mixed orbit ensemble, covariance permits restriction to a measurement of the form \eqref{eq:covseed}, after which the same completeness calculation applies unchanged. This is the state-independent fact invoked in the proof of
Theorem~\ref{thm:mixed}.

% --------------------------------------
\subsubsection{Proof of Proposition~\ref{prop:known}}
\label{app:proofknown}
% --------------------------------------

Covariance and $T_a^\dagger T_a=\mathbb I$ give
\begin{equation}
  \Psucc
  =
  \frac1d
  \sum_{a\in A}
  \bra{T_a\eta}
  T_aM_0T_a^\dagger
  \ket{T_a\eta}
  =
  \bra{\eta}M_0\ket{\eta},
  \label{eq:Psuccseed}
\end{equation}
independent of $a$.  Expanding in the character basis,
\begin{equation}
  \Psucc
  =
  \sum_{\chi,\chi'\in\widehat A}
  \sqrt{p_\chi p_{\chi'}}\,
  e^{-i\theta_\chi}e^{i\theta_{\chi'}}\,
  m_{\chi\chi'}.
  \label{eq:Psuccexp}
\end{equation}
Positivity of $M_0$, applied to its $2\times2$ principal submatrices,
together with Eq.~\eqref{eq:diagcond}, gives
\begin{equation}
  |m_{\chi\chi'}|^2
  \leq
  m_{\chi\chi}m_{\chi'\chi'}
  =
  \frac1{d^2}.
  \label{eq:mCS}
\end{equation}
Since $\bra{\eta}M_0\ket{\eta}$ is real and non-negative, applying the
triangle inequality to Eq.~\eqref{eq:Psuccexp} and then using
Eq.~\eqref{eq:mCS} yields
\begin{equation}
  \Psucc
  \leq
  \sum_{\chi,\chi'}
  \sqrt{p_\chi p_{\chi'}}
  |m_{\chi\chi'}| 
  \leq
  \frac1d
  \left(
  \sum_{\chi\in\widehat A}\sqrt{p_\chi}
  \right)^2.
  \label{eq:upperbound}
\end{equation}

For attainability, take
\begin{equation}
  \widetilde M_0
  :=
  \frac1d\ket{\mu}\!\bra{\mu},
  \qquad
  \ket{\mu}
  :=
  \sum_{\chi\in\widehat A}
  e^{i\theta_\chi}\ket{\chi}.
  \label{eq:seed}
\end{equation}
If $p_\chi=0$, the corresponding $\theta_\chi$ may be chosen
arbitrarily.  The seed is positive semidefinite by construction, and
its matrix elements are
\begin{equation}
  \widetilde m_{\chi\chi'}
  =
  \frac1d
  e^{i\theta_\chi}e^{-i\theta_{\chi'}},
\end{equation}
so $\widetilde m_{\chi\chi}=1/d$ as required by
Eq.~\eqref{eq:diagcond}.  Equation~\eqref{eq:twirl} then gives
\begin{equation}
  \sum_{a\in A}
  T_a\widetilde M_0T_a^\dagger
  =
  \sum_{\chi\in\widehat A}
  \ket{\chi}\!\bra{\chi}
  =
  \mathbb I,
\end{equation}
and the operators
$\widetilde M_a=T_a\widetilde M_0T_a^\dagger$ form a valid POVM.  Its
success probability is
\begin{equation}
\begin{split}
  \widetilde P_{\mathrm{succ}}
  &=
  \bra{\eta}\widetilde M_0\ket{\eta} \\
  &=
  \frac1d|\braket{\mu}{\eta}|^2 \\
  &=
  \frac1d
  \left(
  \sum_{\chi\in\widehat A}\sqrt{p_\chi}
  \right)^2,
\end{split}
  \label{eq:attain}
\end{equation}
which meets the upper bound \eqref{eq:upperbound}.  Hence
Eq.~\eqref{eq:Pval} follows, and the POVM generated by
\eqref{eq:seed} is optimal.  It need not be unique, particularly when
some Fourier weights vanish. \hfill$\square$

Two remarks are useful.  First, optimality of the square-root measurement for Abelian geometrically uniform ensembles is known \cite{Eldar:2000dav,ban1997optimum,Krovi:2015bah,Zhou:2025iac}. The direct argument above fixes the normalization conventions and exhibits the seed constraint \eqref{eq:diagcond} used below.

Second, by Proposition~\ref{prop:gram}, the right-hand side of
Eq.~\eqref{eq:Pval} may equivalently be written in terms of the Gram
spectrum:
\begin{equation}
  \Pval
  =
  \frac1{d^2}
  \left(
  \sum_{\chi\in\widehat A}\sqrt{\lambda_\chi}
  \right)^2
  =
  \frac1{d^2}
  \left(\Tr\sqrt G\right)^2.
  \label{eq:gramform}
\end{equation}
Zero eigenvalues are included, so Eq.~\eqref{eq:gramform} remains valid
when the orbit states are linearly dependent.  This is the form common
in the geometrically uniform literature and the one used for the
non-Abelian comparison in Appendix~\ref{sec:nonabelian}.

% -----------------------------------
\section{Mixed-State Readout under Covariant Noise}
\label{sec:incoherent}
% -----------------------------------

This appendix examines how character-basis coherence constrains value readout for noisy, generally mixed response states.  It distinguishes the settings in which the $\ell_1$-norm coherence determines readout degradation exactly from those in which it provides only an upper bound.

% --------------------------------------
\subsection{Covariant noise and the mixed-state bound}
\label{sec:covbound}
% --------------------------------------

Let $\Lambda$ be a CPTP map on the response register which is \emph{translation covariant},
\begin{equation}
\Lambda\!\left(T_a X T_a^{\dagger}\right)
= T_a\,\Lambda(X)\,T_a^{\dagger}
\qquad \forall a\in A .
\label{eq:covariance}
\end{equation}
Because the oracle applies $T_{f(x)}$ before the noise acts, covariance gives
\begin{equation}
\Lambda\!\left(
T_a\ket{\eta}\!\bra{\eta}T_a^\dagger
\right)
=
T_a\rho T_a^\dagger,
\qquad
\rho:=\Lambda(\ket{\eta}\!\bra{\eta}).
\end{equation}
Thus the noisy value-readout problem remains the uniformly weighted mixed-state orbit ensemble
\begin{equation}
\mathcal{E}_\rho
=
\left\{
\frac1d,\,
T_a\rho T_a^\dagger
\right\}_{a\in A}.
\end{equation}

For comparison with the $\ell_1$-norm coherence $C(\rho)$ defined in Eq.~\eqref{eq:Cdef}, let $C_R(\rho)$ denote the robustness of coherence in the character basis,
\begin{equation}
C_R(\rho)
:=
\min\big\{
s\ge0:
\rho\le(1+s)\delta,\quad
\delta\ \text{incoherent in the character basis}
\big\}.
\end{equation}
For a uniformly weighted group orbit, the optimal discrimination probability is determined exactly by the robustness of asymmetry \cite{Piani:2016lse,Bagan:2018uuk}.  In the present regular Abelian representation, the group twirl is complete dephasing in the character basis, so the robustness of asymmetry reduces to $C_R(\rho)$ \cite{Napoli:2016vpc,Piani:2016lse}.

\begin{theorem}[Mixed-state readout identity and coherence bound]
\label{thm:mixed}
For any density operator $\rho$ on the response register, the ensemble $\mathcal{E}_\rho$ satisfies
\begin{equation}
\Pval
=
\frac{1+C_R(\rho)}{d}
\;\le\;
\frac{1+C(\rho)}{d}.
\label{eq:mixedbound}
\end{equation}
The upper bound is attained whenever there exist real numbers $\{\alpha_\chi\}$ such that
\begin{equation}
\bra{\chi}\rho\ket{\chi'}
=
\left|\bra{\chi}\rho\ket{\chi'}\right|
e^{i(\alpha_\chi-\alpha_{\chi'})}
\end{equation}
for every nonzero matrix element.  If every matrix element $\bra{\chi}\rho\ket{\chi'}$ is nonzero, this phase-factorization condition is also necessary for equality.
\end{theorem}

\begin{proof}
The identity
$\Pval=[1+C_R(\rho)]/d$
is the uniform-prior covariant-discrimination formula for the
robustness of asymmetry, specialized to the present Abelian regular
representation \cite{Piani:2016lse,Bagan:2018uuk}.  The inequality
$C_R(\rho)\leq C(\rho)$ is the standard upper bound of the robustness
of coherence by the $\ell_1$-norm coherence
\cite{Napoli:2016vpc,Piani:2016lse}.

For completeness, the equality conditions can be seen directly from
the covariant measurement.  An optimal POVM may be taken in the form
\begin{equation}
  M_a=T_aM_0T_a^\dagger
\end{equation}
\cite{holevo1973statistical,holevo2011probabilistic}.  The
completeness calculation in Appendix~\ref{app:completeness} fixes
\begin{equation}
  m_{\chi\chi}
  :=
  \bra{\chi}M_0\ket{\chi}
  =
  \frac1d.
\end{equation}
Covariance of the ensemble and the measurement gives
$\Psucc=\Tr(\rho M_0)$.  Since $M_0\succeq0$,
\begin{equation}
  |m_{\chi\chi'}|^2
  \leq
  m_{\chi\chi}m_{\chi'\chi'}
  =
  \frac1{d^2},
\end{equation}
and therefore
\begin{equation}
  \Psucc
  =
  \left|
    \sum_{\chi,\chi'}
    \rho_{\chi\chi'}m_{\chi'\chi}
  \right|
  \leq
  \frac1d
  \sum_{\chi,\chi'}
  |\rho_{\chi\chi'}|
  =
  \frac{1+C(\rho)}d.
\end{equation}

If the phases factorize as stated, choose
\begin{equation}
  M_0
  =
  \frac1d\ket{\mu}\!\bra{\mu},
  \qquad
  \ket{\mu}
  =
  \sum_\chi e^{i\alpha_\chi}\ket{\chi}.
\end{equation}
This seed satisfies the completeness constraint and makes every
nonzero term in $\Tr(\rho M_0)$ real and nonnegative, so the bound is
attained.  This sufficient condition is equivalently the condition
that a diagonal unitary maps $\rho$ to an entrywise nonnegative matrix
\cite{Piani:2016lse}.

Conversely, suppose every matrix element of $\rho$ is nonzero and the
upper bound is attained.  Equality then requires
$|m_{\chi\chi'}|=1/d$ for every pair.  Since $M_0$ is positive
semidefinite, it is a Gram matrix; saturation of the Gram
Cauchy--Schwarz inequality for every pair makes all its Gram vectors
collinear.  Hence
\begin{equation}
  M_0
  =
  \frac1d\ket{\mu}\!\bra{\mu},
  \qquad
  |\mu_\chi|=1.
\end{equation}
Equality in the triangle inequality then forces the phases of all
matrix elements of $\rho$ to factorize as stated.
\end{proof}

\paragraph{An explicit primal--dual certificate.}
For a phase-factorizable state, write $\rho=D B D^\dagger$ in
the character basis, where $D_{\chi\chi}=e^{i\alpha_\chi}$ and
$B_{\chi\chi'}=|\rho_{\chi\chi'}|$. Define
\begin{equation}
s_\chi:=\sum_{\chi'}B_{\chi\chi'},\qquad
Y=\frac1d\operatorname{diag}(s_\chi),\qquad
M_a=T_a\frac{\ket\mu\!\bra\mu}{d}T_a^\dagger,
\quad \mu_\chi=e^{i\alpha_\chi}.
\label{eq:h1explicitcertificate}
\end{equation}
The POVM is feasible by Eq.~\eqref{eq:diagcond}. Moreover,
\begin{equation}
v^\dagger[\operatorname{diag}(s)-B]v
=\frac12\sum_{\chi,\chi'}B_{\chi\chi'}
|v_\chi-v_{\chi'}|^2\geq0.
\end{equation}
Since $Y$ commutes with $D$ and every $T_a$, this proves
$Y\succeq T_a\rho T_a^\dagger/d$ for every $a$.
Its trace equals the success of the displayed POVM,
$\Tr Y=[1+C(\rho)]/d$, giving an algebraic certificate independent
of solver tolerances.

\begin{remark}[Relation to prior work]
\label{rem:bcvd}
Bacon, Childs and van Dam \cite{Bacon:2005cri} obtain a related multiplicity-bearing Abelian covariant ensemble in their analysis of hidden subgroups of semidirect products.  After diagonal rephasing, their reference state decomposes into entrywise nonnegative blocks, so $C_R(\rho)=C(\rho)$ \cite{Piani:2016lse}.  Their success probability can consequently be written as $\Psucc=[1+C(\rho)]/|A|$, saturating the general robustness identity. Their representation contains multiplicity sectors and is therefore not an instance of the multiplicity-free discrimination setting of Theorem~\ref{thm:mixed}.  Optimality of their pretty-good measurement is established directly through the Holevo--Yuen--Kennedy--Lax conditions \cite{holevo1973statistical,Yuen:1975ptn}.
\end{remark}

For a pure response state, the phase-factorization condition holds with $\alpha_\chi=\theta_\chi$, and $C_R(\ket{\eta}\!\bra{\eta})=C(\eta)$.  Equation~\eqref{eq:mixedbound} then reduces to Eq.~\eqref{eq:master}.  For a general mixed response state, $C(\rho)$ upper-bounds readout capability and determines it exactly when the equality condition of Theorem~\ref{thm:mixed} holds.  In particular, if $\rho$ is incoherent in the character basis, then $T_a\rho T_a^\dagger=\rho$ for every $a$, and consequently $\Pval=1/d$.

The numerical results in Appendix~\ref{app:numsdp} illustrate both attainment and strictness of the mixed-state coherence bound.

% --------------------------------------
\subsection{The structure of covariant noise}
\label{sec:modes}
% --------------------------------------

Covariance constrains $\Lambda$ sharply.  Under the adjoint action of $T_a$, the operator basis transforms as
\begin{equation}
T_a\ket{\chi}\!\bra{\chi'}T_a^\dagger
=
\chi(a)\overline{\chi'(a)}
\ket{\chi}\!\bra{\chi'}
=
(\chi\bar\chi')(a)\ket{\chi}\!\bra{\chi'} .
\end{equation}
The operator space therefore decomposes into \emph{modes} labelled by the character quotient $\nu=\chi\bar\chi'\in\widehat A$.  A covariant channel preserves each mode:
\begin{equation}
\Lambda\bigl(\ket{\chi}\!\bra{\chi'}\bigr)
\in
\operatorname{span}\left\{
\ket{\kappa}\!\bra{\kappa'}:
\kappa\bar\kappa'=\chi\bar\chi'
\right\}.
\label{eq:modes}
\end{equation}
Covariance may mix operators within a mode, but it cannot mix distinct modes.  For the regular Abelian representation, the trivial mode $\nu=1$ is precisely the diagonal operator space, while all character-basis coherence lies in the nontrivial modes.

\begin{definition}
\label{def:coherencescaling}
A translation-covariant channel $\Lambda_\tau$ is \emph{coherence-scaling} with contraction $\tau\in[0,1]$ if its restriction to every nontrivial mode is $\tau$ times the identity superoperator, i.e.
\begin{equation}
\bra{\chi}\Lambda_\tau(\rho)\ket{\chi'}
=
\tau\bra{\chi}\rho\ket{\chi'}
\qquad
\text{for all }\chi\neq\chi'
\end{equation}
and every response state $\rho$.
\end{definition}
Coherence scaling is an additional restriction beyond covariance: a general translation-covariant channel may mix operators within a fixed nontrivial mode, whereas Definition~\ref{def:coherencescaling} requires the same scalar action on every nontrivial mode.

\begin{corollary}[Coherence-scaling readout]
\label{cor:linear}
Let $\Lambda_\tau$ be a coherence-scaling channel acting on a pure response state $\ket{\eta}$.  Then
\begin{equation}
\Pval\!\left(
\Lambda_\tau(\ket{\eta}\!\bra{\eta})
\right)
=
\frac{1+\tau C(\eta)}{d},
\qquad
\Delta\Pval
=
-\,(1-\tau)\frac{C(\eta)}{d},
\label{eq:linearlaw}
\end{equation}
where
\begin{equation}
\Delta\Pval
:=
\Pval\!\left(
\Lambda_\tau(\ket{\eta}\!\bra{\eta})
\right)
-
\Pval(\eta).
\end{equation}
More generally, for any mixed input $\rho$ satisfying the phase-factorization condition of Theorem~\ref{thm:mixed}, the same relations hold with $C(\eta)$ replaced by $C(\rho)$ and $\Delta\Pval :=\Pval(\Lambda_\tau(\rho))-\Pval(\rho)$.
\end{corollary}

\begin{proof}
By definition,
\begin{equation}
  C\!\left(
    \Lambda_\tau(\ket{\eta}\!\bra{\eta})
  \right)
  =
  \tau C(\eta).
\end{equation}
Moreover, multiplication of every off-diagonal matrix element by the
nonnegative scalar $\tau$ preserves the factorized phases of the pure
input state.  The equality condition of
Theorem~\ref{thm:mixed} therefore gives
\begin{equation}
  \Pval\!\left(
    \Lambda_\tau(\ket{\eta}\!\bra{\eta})
  \right)
  =
  \frac{1+\tau C(\eta)}{d}.
\end{equation}
Subtracting the noiseless value
$\Pval(\eta)=[1+C(\eta)]/d$
yields Eq.~\eqref{eq:linearlaw}.  The same argument applies to a phase-factorizable mixed input: coherence scaling preserves its phase factorization and multiplies its $\ell_1$-norm coherence by $\tau$.
\end{proof}

Two standard processes are coherence-scaling channels. \emph{Character-basis dephasing} of strength $\gamma\geq0$ acts as
\begin{equation}
\bra{\chi}\Lambda_\gamma(\rho)\ket{\chi'}
=
\begin{cases}
\bra{\chi}\rho\ket{\chi}, & \chi=\chi',\\
e^{-\gamma}\bra{\chi}\rho\ket{\chi'}, & \chi\neq\chi',
\end{cases}
\end{equation}
and hence has $\tau=e^{-\gamma}$.  \emph{Depolarizing noise},
\begin{equation}
\Lambda_\tau(\rho)
=
\tau\rho+(1-\tau)\frac{\mathbb I}{d},
\end{equation}
also scales every off-diagonal matrix element by $\tau$ and preserves the input phases.  Corollary~\ref{cor:linear} therefore applies to both processes.

Mode-independent loss with transmission $\tau$ produces the same readout law, but it enlarges the state space by introducing a vacuum sector and is not a coherence-scaling channel on the original response register.  It is treated separately as the constant-profile case of Theorem~\ref{thm:mdl}.

% --------------------------------------
\subsection{Consequences of coherence scaling}
\label{sec:incohcons}
% --------------------------------------

For a response state $\rho$ satisfying the phase-factorization condition of Theorem~\ref{thm:mixed}, subject to coherence-scaling noise, Eq.~\eqref{eq:linearlaw} has three immediate consequences. This includes every pure response state $\rho=\ket{\eta}\!\bra{\eta}$.

\emph{(i) The loss is bilinear in the contraction deficit and the resource.}  The exact loss is
\begin{equation}
\Delta\Pval
=
-(1-\tau)\frac{C(\rho)}{d}.
\end{equation}
It is therefore linear in both the contraction deficit $1-\tau$ and the input coherence $C(\rho)$.  If a physical noise strength $\gamma$ is parametrized by $\tau=e^{-\gamma}$, as for character-basis dephasing, then $1-\tau=\gamma+O(\gamma^2)$ and the loss is first order in $\gamma$ in the weak-noise limit.

\emph{(ii) Incoherent responses remain indistinguishable.}
For any response diagonal in the character basis,
\begin{equation}
T_a\rho T_a^\dagger=\rho\qquad\forall a\in A.
\label{eq:phasecornerinvariance}
\end{equation}
A value-independent post-oracle channel preserves this equality,
so the readout probability remains $1/d$ by Eq.~\eqref{eq:cornerexact}.

\emph{(iii) Within this model, greater capability entails greater fragility.}  For the class of inputs considered here---pure states and mixed states satisfying the phase-factorization condition of Theorem~\ref{thm:mixed}---the noiseless value-readout probability is
\begin{equation}
\Pval^{(0)}
:=
\Pval(\rho)
=
\frac{1+C(\rho)}{d}.
\end{equation}
Combining this identity with Eq.~\eqref{eq:linearlaw} gives
\begin{equation}
\frac{\Delta\Pval}{\Pval^{(0)}}
=
-\,(1-\tau)\frac{C(\rho)}{1+C(\rho)}.
\label{eq:relative}
\end{equation}
For fixed $\tau$, the magnitudes of both the absolute and relative losses increase monotonically with $C(\rho)$.  The relative-loss magnitude ranges from zero at the phase corner to
\begin{equation}
  (1-\tau)\frac{d-1}{d}
\end{equation}
at the value-readout corner $C(\rho)=d-1$.  Hence, within the coherence-scaling model and the stated class of inputs, states with greater ideal value-readout capability undergo greater degradation, while the exact phase corner has no value-readout capability to lose.

% -----------------------------------
\section{Coherent Apparatus Imperfections}
\label{sec:coherent_apparatus_imperfections}
\label{sec:coherent}
% -----------------------------------

We now turn from noise on the response state to coherent imperfections in the apparatus that realizes an optimal measurement.  We first express a chosen optimal POVM as a fixed interferometer and then study its sensitivity to implementation errors.

% --------------------------------------
\subsection{The optimal measurement is a fixed interferometer}
\label{sec:mesh}
% --------------------------------------

Choose phases $\{\theta_\chi\}_{\chi\in\widehat A}$ in the decomposition \eqref{eq:state}, with arbitrary choices permitted when $p_\chi=0$, and define the diagonal unitary
\begin{equation}
V
:=
\sum_{\chi\in\widehat A}
e^{i\theta_\chi}\ket{\chi}\!\bra{\chi},
\qquad
\ket{\varphi}
:=
\sum_{\chi\in\widehat A}
\sqrt{p_\chi}\,\ket{\chi}
=
V^\dagger\ket{\eta}.
\label{eq:Vdef}
\end{equation}

\begin{proposition}[Interferometric form]
\label{prop:mesh}
An optimal POVM from Proposition~\ref{prop:known} has the form
\begin{equation}
M_a
=
V\ket{a}\!\bra{a}V^\dagger,
\label{eq:meshPOVM}
\end{equation}
so this optimal value-readout measurement can be implemented by applying $V^\dagger$ and then measuring in the computational basis. When some $p_\chi$ vanish, the arbitrary phase choices on those components select one member of a generally nonunique family of optimal POVMs.
\end{proposition}

\begin{proof}
By the second relation in Eq.~\eqref{eq:orth},
\begin{equation}
  \sum_{\chi\in\widehat A}\ket{\chi}
  =
  \sqrt d\,\ket0.
\end{equation}
The seed displayed in Proposition~\ref{prop:known} therefore satisfies
\begin{equation}
  \ket{\mu}
  =
  \sum_{\chi\in\widehat A}
  e^{i\theta_\chi}\ket{\chi}
  =
  \sqrt d\,V\ket0,
  \qquad
  \widetilde M_0
  =
  V\ket0\!\bra0V^\dagger.
\end{equation}
Both $T_a$ and $V$ are diagonal in the character basis and hence
commute.  Using $T_a\ket0=\ket a$, we obtain
\begin{equation}
\begin{split}
  M_a
  &=
  T_a\widetilde M_0T_a^\dagger\\
  &=
  VT_a\ket0\!\bra0T_a^\dagger V^\dagger\\
  &=
  V\ket a\!\bra aV^\dagger.
\end{split}
\end{equation}
\end{proof}

For a fixed phase pattern, the chosen unitary $V^\dagger$ is independent of the Fourier weights $p_\chi$.  Consequently, if the weights are varied along the saturating family \eqref{eq:satfamily} while its phase pattern is held fixed, the same apparatus is optimal at every point of the upper boundary \eqref{eq:tradeoff}.  At points where some weights vanish, this fixed phase pattern remains a valid convention, although the physical state does not determine those phases uniquely.

To express the apparatus in the computational basis, let $\{\ket{\chi}_{\mathrm{lab}}\}_{\chi\in\widehat A}$ be an abstract character-label basis and define the finite-group Fourier transform $F$ by
\begin{equation}
F\ket{\chi}_{\mathrm{lab}}
=
\ket{\chi},
\end{equation}
with $\ket{\chi}$ as in Eq.~\eqref{eq:charbasis}.  Then
\begin{equation}
V^\dagger
=
F\operatorname{diag}_{\chi}
\!\left(e^{-i\theta_\chi}\right)F^\dagger.
\label{eq:Vmesh}
\end{equation}
Thus, acting on the response register, the implementation consists of $F^\dagger$, followed by a phase mask $\operatorname{diag}_{\chi}(e^{-i\theta_\chi})$, followed by $F$. For $\theta_\chi\equiv0$, one has $V=\mathbb I$ and $M_a=\ket a\!\bra a$, so the optimal measurement reduces to direct detection in the computational basis.

% --------------------------------------
\subsection{Stationarity}
\label{sec:stationarity}
% --------------------------------------

Let the imperfect apparatus implement
\begin{equation}
U(\epsilon)
=
e^{i\epsilon G}V^\dagger,
\qquad
G=G^\dagger,
\qquad
\epsilon\in\mathbb R,
\end{equation}
where $G$ acts on the response register.  This models a coherent unitary error applied after the ideal readout unitary $V^\dagger$. Using $[V,T_a]=0$, the success probability obtained with this apparatus is
\begin{equation}
P(\epsilon)
=
\frac1d
\sum_{a\in A}
\left|
\bra a e^{i\epsilon G}T_a\ket{\varphi}
\right|^2,
\qquad
P(0)=\Pval(\eta).
\label{eq:Peps}
\end{equation}

Write
\begin{equation}
\varphi_w
:=
\braket{w}{\varphi}
=
\frac1{\sqrt d}
\sum_{\chi\in\widehat A}
\sqrt{p_\chi}\,
\overline{\chi(w)}.
\label{eq:varphiw}
\end{equation}
In particular,
\begin{equation}
\varphi_0
=
\frac1{\sqrt d}\sum_\chi\sqrt{p_\chi}
=
\sqrt{\frac{1+C(\eta)}d},
\qquad
\varphi_{-w}
=
\overline{\varphi_w}.
\end{equation}

\begin{lemma}[First-order stationarity]
\label{lem:stat}
For every Hermitian generator $G$,
\begin{equation}
\left.
\frac{dP}{d\epsilon}
\right|_{\epsilon=0}
=
0.
\end{equation}
\end{lemma}

\begin{proof}
Define
\begin{equation}
  g_a
  :=
  \bra aGT_a\ket{\varphi}.
\end{equation}
At $\epsilon=0$,
\begin{equation}
  \bra aT_a\ket{\varphi}
  =
  \braket{0}{\varphi}
  =
  \varphi_0,
\end{equation}
which is real and nonnegative.  Expanding $e^{i\epsilon G}=\mathbb I+i\epsilon G+O(\epsilon^2)$ therefore gives
\begin{equation}
  \left.
  \frac{dP}{d\epsilon}
  \right|_{\epsilon=0}
  =
  -\frac{2\varphi_0}{d}
  \operatorname{Im}
  \sum_{a\in A}g_a.
  \label{eq:firstderivative}
\end{equation}

Since $(T_a\ket{\varphi})_z=\varphi_{z-a}$, substituting $z=a+w$ yields
\begin{equation}
  g_a
  =
  \sum_{w\in A}
  G_{a,a+w}\varphi_w.
\end{equation}
Hence
\begin{equation}
  \sum_{a\in A}g_a
  =
  \sum_{w\in A}
  \varphi_w\,\widetilde g(w),
  \qquad
  \widetilde g(w)
  :=
  \sum_{a\in A}G_{a,a+w}.
\end{equation}
Hermiticity of $G$ implies
\begin{equation}
  \widetilde g(-w)
  =
  \overline{\widetilde g(w)},
\end{equation}
while Eq.~\eqref{eq:varphiw} gives $\varphi_{-w}=\overline{\varphi_w}$.  Therefore,
\begin{equation}
  \overline{\sum_a g_a}
  =
  \sum_w
  \varphi_{-w}\widetilde g(-w)
  =
  \sum_a g_a.
\end{equation}
Thus $\sum_a g_a$ is real, and Eq.~\eqref{eq:firstderivative} vanishes.
\end{proof}

Lemma~\ref{lem:stat} holds for each fixed Hermitian generator $G$, without averaging over an error ensemble.  It directly verifies in shift coordinates the stationarity already implied by the optimality of $V^\dagger$, and shows that coherent unitary miscalibration produces no first-order change in the readout probability:
\begin{equation}
P(\epsilon)
=
P(0)+O(\epsilon^2).
\end{equation}
For special generators the quadratic term may also vanish.  Since $P(\epsilon)\leq P(0)$ by optimality of the measurement at $\epsilon=0$, the first nonzero term in the Taylor expansion, if any, must have even order and a negative coefficient.

% --------------------------------------
\subsection{The second-order loss}
\label{sec:secondorder}
% --------------------------------------

Expanding through second order and using Lemma~\ref{lem:stat} to eliminate the linear term gives
\begin{equation}
\begin{split}
&\Delta\Pval
:=
P(\epsilon)-P(0)
=
\frac{\epsilon^2}{d}
\sum_{a\in A}
\left(
|g_a|^2-\varphi_0\operatorname{Re}h_a
\right)
+O(\epsilon^3),\\
&h_a
:=
\bra aG^2T_a\ket{\varphi}, \qquad g_a:=\bra aGT_a\ket{\varphi}.
\end{split}
\label{eq:second}
\end{equation}
For a random error ensemble, the expected quadratic coefficient is a bilinear form in the shift-indexed correlations of \(G\); the general expression is derived in Appendix~\ref{app:second}.

A simpler expression results for a centered Hermitian error ensemble whose second moments satisfy
\begin{equation}
\mathbb E[G_{ij}]=0,
\qquad
\mathbb E[G_{ij}G_{kl}]=0
\quad\text{unless}\quad
(k,l)=(j,i).
\label{eq:pairing}
\end{equation}
Define its \emph{shift-offset variance profile} by
\begin{equation}
V(w)
:=
\sum_{a\in A}
\mathbb E\!\left[
|G_{a,a+w}|^2
\right],
\qquad
V_{\mathrm{tot}}
:=
\sum_{w\in A}V(w).
\label{eq:Vw}
\end{equation}

\begin{proposition}[Shift-profile loss]
\label{prop:Vw}
For an error ensemble satisfying Eq.~\eqref{eq:pairing} and having finite third moments,
\begin{equation}
\mathbb E[\Delta\Pval]
=
-\frac{\epsilon^2}{d}
\sum_{\substack{w\in A\\w\neq0}}
V(w)
\left(
\varphi_0^2-|\varphi_w|^2
\right)
+O(\epsilon^3).
\label{eq:Vwformula}
\end{equation}
If the error ensemble is invariant under $G\mapsto-G$ and has
finite fourth moments, the averaged remainder improves to
$O(\epsilon^4)$.
\end{proposition}

Since
\begin{equation}
|\varphi_w|
\le
\frac1{\sqrt d}\sum_\chi\sqrt{p_\chi}
=
\varphi_0,
\end{equation}
every term in the quadratic loss coefficient is nonnegative. Equation~\eqref{eq:Vwformula} therefore confirms, at second order, that coherent unitary miscalibration cannot improve upon the ideal optimal apparatus.  The quantities
\begin{equation}
|\varphi_w|^2
=
|\braket{w}{\varphi}|^2
\end{equation}
form the value-basis probability distribution of the phase-stripped state \(\ket{\varphi}\), while \(V(w)\) is the aggregate error power in matrix elements whose value-basis indices differ by \(w\). The quadratic loss is the weighted mismatch between \(|\varphi_w|^2\) and its maximal value \(\varphi_0^2\).

\paragraph{Corners.}
At the phase corner,
\begin{equation}
\varphi_w
=
\frac1{\sqrt d}\overline{\chi_0(w)},
\qquad
|\varphi_w|^2
=
\frac1d
=
\varphi_0^2,
\end{equation}
so every term in the quadratic coefficient vanishes, independently of the profile \(V(w)\).  In fact, Eq.~\eqref{eq:cornerexact} implies that the loss vanishes to all orders for any value-independent apparatus error.

At the value-readout corner, \(\varphi_w=\delta_{w,0}\), and hence
%\footnote{Unlike Eqs.~\eqref{eq:Vwformula} and \eqref{eq:isotropic}, this improvement needs no assumption on the error ensemble: at this corner $T_a\ket{\varphi}=\ket{a}$, so $P(\epsilon)$ reduces to diagonal matrix elements of $e^{i\epsilon G}$, and Hermiticity of $G$ gives $\langle a|e^{-i\epsilon G}|a\rangle=\overline{\langle a|e^{i\epsilon G}|a\rangle}$, making $P(\epsilon)$ even in $\epsilon$ for every individual $G$.}
\begin{equation}
-\mathbb E[\Delta\Pval]
=
\frac{\epsilon^2}{d}
\left(
V_{\mathrm{tot}}-V(0)
\right)
+O(\epsilon^4),
\end{equation}
which is the aggregate off-diagonal error power at quadratic order.
At this corner, the stronger $O(\epsilon^4)$ remainder holds for each fixed Hermitian $G$, because $P(\epsilon)=d^{-1}\sum_a|\bra a e^{i\epsilon G}\ket a|^2$ is even in $\epsilon$.

\paragraph{The isotropic case.}
If the variance profile is flat,
\begin{equation}
V(w)
=
\frac{V_{\mathrm{tot}}}{d}
=:
d\sigma_m^2
\qquad
\text{for every }w\in A,
\end{equation}
then, using
\(\sum_{w\neq0}|\varphi_w|^2=1-\varphi_0^2\),
\begin{equation}
\begin{split}
\mathbb E[\Delta\Pval]
&=
-\epsilon^2\sigma_m^2
\left(
d\varphi_0^2-1
\right)
+O(\epsilon^3)\\
&=
-\epsilon^2\sigma_m^2 C(\eta)
+O(\epsilon^3).
\end{split}
\label{eq:isotropic}
\end{equation}
Flatness of the entire profile is sufficient for
Eq.~\eqref{eq:isotropic}, but is not necessary for a scalar reduction.
The $w=0$ term in Eq.~\eqref{eq:Vwformula} vanishes identically.
A purely diagonal value-basis unitary immediately before detection
leaves every outcome probability unchanged at all error strengths.
The relevant scalar calibration therefore uses
\begin{equation}
\overline V_{\mathrm{off}}
:=\frac1{d-1}\sum_{w\ne0}V(w),\qquad
c_{\mathrm{off}}:=\frac{\overline V_{\mathrm{off}}}{d}C(\eta).
\label{eq:h1offmean}
\end{equation}
For the upper-boundary family, put $q=(1-\Phi)/(d-1)$.
Character orthogonality gives
\begin{equation}
\varphi_w=\sqrt{dq}\,\delta_{w,0}
+\frac{\sqrt\Phi-\sqrt q}{\sqrt d}\overline{\chi_0(w)}.
\label{eq:h1boundaryamplitudes}
\end{equation}
Hence $\varphi_0^2-|\varphi_w|^2=C(\eta)/(d-1)$ for every
$w\ne0$, and the reduced profile coefficient satisfies
\begin{equation}
c_{\mathrm{profile}}
:=\frac1d\sum_{w\ne0}V(w)
\bigl(\varphi_0^2-|\varphi_w|^2\bigr)
=\frac{\overline V_{\mathrm{off}}}{d}C(\eta)
=c_{\mathrm{off}}.
\label{eq:h1boundaryscalar}
\end{equation}
This identity holds on the upper boundary for any variance profile.
Under Eq.~\eqref{eq:pairing}, it is the full quadratic loss
coefficient. Without pairing, it is an exact reduction of the profile
model; correlations retained by Eq.~\eqref{eq:generalsecond} can
still change the full coefficient. For general response weights,
$c_{\mathrm{off}}$ is an approximation unless additional conditions
on the profile hold.

Writing $\overline V=d^{-1}\sum_wV(w)$, the all-shift isotropic
coefficient is $c_{\mathrm{iso}}=\overline VC(\eta)/d$. Thus,
whenever the denominator is nonzero,
\begin{equation}
\frac{c_{\mathrm{iso}}}{c_{\mathrm{off}}}
=\frac{\overline V}{\overline V_{\mathrm{off}}}
=\frac{d-1}{d-V(0)/\overline V}.
\label{eq:h1calibrationratio}
\end{equation}

% --------------------------------------
\subsection{Independent random residual phases induce dephasing}
\label{sec:phaseerr}
% --------------------------------------

Suppose the apparatus implements $V^{\dagger}$ while the prepared state carries phases $\theta_\chi+\delta_\chi$. For a fixed residual phase pattern $\delta=(\delta_\chi)_\chi$, the success probability of this fixed
apparatus is
\begin{equation}
P_{\mathrm{fix}}(\delta)
= \frac{1}{d}
\left|\sum_\chi\sqrt{p_\chi}\,e^{i\delta_\chi}\right|^{2}.
\label{eq:phasedelta}
\end{equation}
Now let the $\delta_\chi$ be independent and identically distributed, with characteristic function $\phi:=\mathbb{E}[e^{i\delta}]$.  For $\chi\neq\chi'$,
\begin{equation}
  \mathbb{E}\!\left[
  e^{i(\delta_\chi-\delta_{\chi'})}
  \right]
  = |\phi|^{2},
  \label{eq:phasefactor}
\end{equation}
and therefore
\begin{equation}
  \mathbb{E}\left[P_{\mathrm{fix}}(\delta)\right]
  = \frac{1+|\phi|^{2}C(\eta)}{d}.
  \label{eq:phaseresult}
\end{equation}

This equivalence holds at the channel level.  If $\Delta(\rho):=\sum_\chi\ket{\chi}\!\bra{\chi}\rho
\ket{\chi}\!\bra{\chi}$ denotes complete dephasing in the character
basis, averaging the random diagonal unitary gives
\begin{equation}
  \overline{\Lambda}(\rho)
  = |\phi|^{2}\rho
  + \bigl(1-|\phi|^{2}\bigr)\Delta(\rho).
  \label{eq:phasechannel}
\end{equation}
The averaged state retains the original factorized phase pattern,
so the fixed analyzer $V^\dagger$ is optimal for the averaged
ensemble. Equation~\eqref{eq:phaseresult} therefore gives its optimal
readout probability, although the analyzer need not be optimal for
an individual realization. The average is exactly a coherence-scaling
channel with contraction $\tau=|\phi|^2$, equivalently character-basis
dephasing of strength $\gamma=-\ln|\phi|^2$. Individual realizations
remain coherent; $\phi=0$ is the complete-dephasing limit.

For Gaussian phase errors of variance $\sigma^2$,
$|\phi|^2=e^{-\sigma^2}$, independently of their common mean.
The positive fractional decrease from
the phase-matched value is therefore
\begin{equation}
\begin{split}
\frac{P_{\mathrm{fix}}(0)-\mathbb{E}[P_{\mathrm{fix}}(\delta)]}{P_{\mathrm{fix}}(0)}
&= \left(1-e^{-\sigma^{2}}\right)\frac{C(\eta)}{1+C(\eta)}\\
&= \sigma^{2}\frac{C(\eta)}{1+C(\eta)}+O(\sigma^{4}).
\end{split}
  \label{eq:gaussianphaseloss}
\end{equation}

% --------------------------------------
\subsection{Comparison of error scalings}
\label{sec:twoorders}

The preceding results concern separate error models and do not define
an additive joint-noise law.  Within their respective regimes,
coherence-scaling noise produces a loss proportional to
$(1-\tau)C$, whereas isotropic coherent mismatch produces a leading
loss proportional to $\epsilon^2C$; both vanish at the phase corner.
Thus, if $1-\tau=a\lambda+O(\lambda^2)$ and
$\epsilon=b\lambda+O(\lambda^2)$ for a common small parameter
$\lambda$, the former enters at order $\lambda$ and the latter at
order $\lambda^2$, provided its quadratic coefficient is nonzero.
This comparison depends on the physical parametrization: independent
Gaussian phase fluctuations, for example, have
$1-\tau=O(\sigma^2)$.  A simultaneous error model may also contain
mixed corrections and is not considered here.

% ==================================================
\subsection{Derivation of the General Second-Order Formula}
\label{app:second}
% ==================================================
This subsection derives the general quadratic expression used in Appendix~\ref{sec:secondorder} and proves Proposition~\ref{prop:Vw}. Throughout, the apparatus implements
\begin{equation}
  U=e^{i\epsilon G}V^\dagger
\end{equation}
in place of $V^\dagger$, with $G=G^\dagger$, and we write $\varphi_z:=\braket{z}{\varphi}$ for the value-basis amplitudes of the phase-stripped state \eqref{eq:Vdef}.  All group indices and all sums over shifts run over $A$. For random error ensembles, we assume $\mathbb{E}\|G\|^3<\infty$, which is sufficient to interchange expectation with the expansion through second order and to retain an integrable $O(\epsilon^3)$ remainder.

% --------------------------------------
\subsubsection{Reduction to shift coordinates}
\label{app:shift}
% --------------------------------------

Starting from the quadratic expansion \eqref{eq:second}, in which the linear term vanishes by Lemma~\ref{lem:stat}, we reduce its two contributions to shift coordinates. 
Since
\begin{equation}
  (T_a\ket{\varphi})_z=\varphi_{z-a},
\end{equation}
substituting $z=a+w$ gives
\begin{equation}
g_a
=
\sum_zG_{a,z}\varphi_{z-a} 
=
\sum_wG_{a,a+w}\varphi_w,
\label{eq:ga}
\end{equation}
and likewise
\begin{equation}
  h_a
  =
  \sum_w(G^2)_{a,a+w}\varphi_w.
\end{equation}
It follows that
\begin{equation}
  \sum_a|g_a|^2
  =
  \sum_{w,w'}
  \left(
  \sum_a
  G_{a,a+w}\,
  \overline{G_{a,a+w'}}
  \right)
  \varphi_w\overline{\varphi_{w'}},
  \label{eq:sumg2}
\end{equation}
while inserting a resolution of the identity into the second term,
with the intermediate index written as $a+u$, gives
\begin{equation}
  \sum_ah_a
  =
  \sum_w\varphi_w
  \sum_a\sum_u
  G_{a,a+u}G_{a+u,a+w}.
  \label{eq:sumh}
\end{equation}

Both expressions depend on $G$ through shift-indexed correlations.  For
a random error ensemble, define
\begin{align}
  Q(w,w')
  &:=
  \sum_a
  \mathbb E\!\left[
  G_{a,a+w}\,
  \overline{G_{a,a+w'}}
  \right],
  \label{eq:Qdef}\\
  R(w)
  &:=
  \sum_a\sum_u
  \mathbb E\!\left[
  G_{a,a+u}G_{a+u,a+w}
  \right].
  \label{eq:Rdef}
\end{align}
Taking the expectation of Eq.~\eqref{eq:second} and substituting
Eqs.~\eqref{eq:sumg2}--\eqref{eq:sumh} yields
\begin{equation}
\mathbb E[\Delta\Pval]
=
\frac{\epsilon^2}{d}
\Bigg[
\sum_{w,w'}
Q(w,w')\,
\varphi_w\overline{\varphi_{w'}} 
-
\varphi_0\operatorname{Re}
\sum_wR(w)\varphi_w
\Bigg]
+O(\epsilon^3).
\label{eq:generalsecond}
\end{equation}
Equation~\eqref{eq:generalsecond} requires no independence or pairing assumption on the matrix elements of $G$ and gives the exact quadratic coefficient for any Hermitian error ensemble satisfying the stated moment condition.

% --------------------------------------
\subsubsection{Second-moment pairing ensembles}
\label{app:indep}
% --------------------------------------

Under the centered second-moment pairing condition \eqref{eq:pairing}, the general expression simplifies considerably. This condition holds, for example, when the diagonal entries are independent centered real variables and the upper-triangular off-diagonal entries are independent centered circular complex variables, with the lower-triangular entries fixed by Hermiticity. Full statistical independence is not required.

For $Q$, Hermiticity gives
\begin{equation}
\mathbb E\!\left[
G_{a,a+w}\overline{G_{a,a+w'}}
\right] 
=
\mathbb E\!\left[
G_{a,a+w}G_{a+w',a}
\right].
\end{equation}
By Eq.~\eqref{eq:pairing}, this can be nonzero only when
$w'=w$.  Hence $Q$ is diagonal:
\begin{equation}
  Q(w,w')
  =
  \delta_{w,w'}V(w),
  \qquad
  V(w)
  =
  \sum_a
  \mathbb E\!\left[
  |G_{a,a+w}|^2
  \right].
  \label{eq:Qdiag}
\end{equation}
This recovers the shift-offset variance profile of Eq.~\eqref{eq:Vw}.

For $R$, the pairing
\begin{equation}
  \mathbb E\!\left[
  G_{a,a+u}G_{a+u,a+w}
  \right]
\end{equation}
can be nonzero under Eq.~\eqref{eq:pairing} only when $w=0$.  Therefore
\begin{equation}
R(w)
=
\delta_{w,0}
\sum_a\sum_u
\mathbb E\!\left[
|G_{a,a+u}|^2
\right] 
=
\delta_{w,0}V_{\mathrm{tot}},
\label{eq:Rdiag}
\end{equation}
where $V_{\mathrm{tot}}$ is defined in Eq.~\eqref{eq:Vw}.

Substituting Eqs.~\eqref{eq:Qdiag} and \eqref{eq:Rdiag} into
Eq.~\eqref{eq:generalsecond}, and using that $\varphi_0$ is real and
positive, gives
\begin{equation}
\begin{split}
  \mathbb E[\Delta\Pval]
  &=
  \frac{\epsilon^2}{d}
  \left[
  \sum_wV(w)|\varphi_w|^2
  -
  \varphi_0^2V_{\mathrm{tot}}
  \right]
  +O(\epsilon^3) \\
  &=
  -\frac{\epsilon^2}{d}
  \sum_w
  V(w)
  \left(
  \varphi_0^2-|\varphi_w|^2
  \right)
  +O(\epsilon^3).
\end{split}
  \label{eq:appVw}
\end{equation}
The $w=0$ term vanishes identically, so
Eq.~\eqref{eq:appVw} gives Eq.~\eqref{eq:Vwformula} and proves
Proposition~\ref{prop:Vw}. \hfill$\square$

% --------------------------------------
\subsubsection{Relation to physical meshes}
\label{app:validity}
% --------------------------------------

The pairing condition \eqref{eq:pairing} is generally not satisfied by
the effective generator of a physical interferometer.  The
$O(d^2)$ physical MZI and phase-shifter parameters induce correlated
matrix elements of $G$, so $Q(w,w')$ can have off-diagonal entries and
$R(w)$ need not take the form \eqref{eq:Rdiag}.  In that setting,
Eq.~\eqref{eq:generalsecond}, rather than
Eq.~\eqref{eq:appVw}, gives the exact quadratic coefficient.

Two operational corners nevertheless reduce exactly to the
shift-profile expression, for structural reasons independent of the
pairing condition.

At the value corner, $p$ is uniform and the second relation in
Eq.~\eqref{eq:orth} gives
\begin{equation}
  \varphi_w=\delta_{w,0}.
\end{equation}
The general quadratic coefficient \eqref{eq:generalsecond} then reduces
to
\begin{equation}
  \frac{\epsilon^2}{d}
  \left[
  Q(0,0)-R(0)
  \right].
  \label{eq:valuegeneral}
\end{equation}
Hermiticity alone gives
\begin{equation}
  Q(0,0)
  =
  \sum_a\mathbb E|G_{a,a}|^2
  =
  V(0),
\end{equation}
and
\begin{equation}
\begin{split}
  R(0)
  &=
  \sum_a\sum_u
  \mathbb E\!\left[
  G_{a,a+u}G_{a+u,a}
  \right] \\
  &=
  \sum_a\sum_u
  \mathbb E|G_{a,a+u}|^2
  =
  V_{\mathrm{tot}}.
\end{split}
\end{equation}
Consequently, Eq.~\eqref{eq:valuegeneral} equals
\begin{equation}
  -\frac{\epsilon^2}{d}
  \left[
  V_{\mathrm{tot}}-V(0)
  \right],
\end{equation}
which is exactly the quadratic coefficient predicted by
Eq.~\eqref{eq:appVw}, even when the effective matrix elements are
correlated.

At the phase corner, $p$ is a point mass at $\chi_0$, so
\begin{equation}
  \varphi_w
  =
  \frac1{\sqrt d}\overline{\chi_0(w)},
  \qquad
  |\varphi_w|^2
  =
  \frac1d
  =
  \varphi_0^2
\end{equation}
for every $w$.  The reduced expression \eqref{eq:appVw} therefore vanishes.  More strongly, Eq.~\eqref{eq:cornerexact} shows that the complete loss vanishes to all orders: all oracle hypotheses coincide up to a global phase, so no value-independent apparatus imperfection can change their optimal success probability from $1/d$.

As shown above, the reduced profile formula is exact at both operational corners.  The numerical comparison in Appendix~\ref{sec:photonic_numerical} finds agreement to within a few percent for the tested intermediate states, while the full correlated expression \eqref{eq:generalsecond} provides the appropriate quantitative model.

% -----------------------------------
\section{Mode-Dependent Loss Beyond Covariance}
\label{sec:mdl}
% -----------------------------------

Appendix~\ref{sec:incoherent} treated post-oracle noise that preserves
the uniform translation-orbit structure through covariance,
Eq.~\eqref{eq:covariance}.  A physically important process that
generally violates this condition is mode-dependent loss, in which the
$d$ value-basis response modes have unequal transmissions $\tau_z$.
We derive an exact decomposition of the resulting discrimination
problem, obtain a universal upper bound, and characterize a broad
phase--loss-compatible family that attains it.  Within this family,
unequal loss can improve optimal readout because the no-detection event
itself carries information about the oracle value.

% --------------------------------------
\subsection{The obstruction and an exact decomposition}
\label{sec:mdlsetup}
% --------------------------------------

Let
\begin{equation}
  L := \sum_z\sqrt{\tau_z}\,\ket{z}\!\bra{z},
  \qquad 0\leq\tau_z\leq1,
\end{equation}
act on the value basis.  Introducing an orthogonal vacuum state, the
loss channel is
\begin{equation}
  \Lambda(\sigma)
  = L\sigma L^{\dagger}
  + \Tr\!\left[\left(\mathbb{I}-L^{\dagger}L\right)\sigma\right]
    \ket{\mathrm{vac}}\!\bra{\mathrm{vac}} .
\end{equation}
We take translations to act trivially on the vacuum sector. The operator $L$ commutes with every translation $T_a$ if and only if
the transmission profile is constant; hence mode-dependent loss is
generally not translation covariant.

The obstruction is visible in the Gram matrix of the unnormalized
surviving states:
\begin{equation}
  G_{ab}
  = \bra{\eta}T_a^{\dagger}\operatorname{diag}(\tau)
    T_b\ket{\eta}
  = \sum_u \overline{\eta_u}\,\tau_{a+u}\,
    \eta_{a+u-b}.
  \label{eq:noncirc}
\end{equation}
Generically, $G_{ab}$ depends on $a$ and $b-a$ separately rather than
on $b-a$ alone.  It is therefore not translation invariant, and the
characters need not diagonalize it, in contrast with
Appendix~\ref{app:gram}.  Although the response state still has Fourier
weights $p_\chi$, those weights alone no longer determine the Gram
spectrum or the optimal readout probability.  Neither
Proposition~\ref{prop:known} nor Theorem~\ref{thm:mixed} therefore
applies directly.

The output corresponding to oracle value $a$ is
\begin{equation}
  \rho_a
  = \ket{v_a}\!\bra{v_a}
  + (1-w_a)\ket{\mathrm{vac}}\!\bra{\mathrm{vac}},
  \qquad
  \ket{v_a}:=LT_a\ket{\eta},
\end{equation}
with survival probability
\begin{equation}
\begin{split}
& w_a
  := \|LT_a\ket{\eta}\|^{2}
  = \sum_z\tau_z|\eta_{z-a}|^{2},\\
& \sum_a w_a
  = \sum_z\tau_z
  = d\,\bar\tau,
\end{split}
  \label{eq:wa}
\end{equation}
where the mean transmission is $\bar\tau:=d^{-1}\sum_z\tau_z$.
Every $\rho_a$ is block diagonal with respect to
\begin{equation}
  \mathcal{H}
  = \mathbb{C}^{d}\oplus
    \operatorname{span}\{\ket{\mathrm{vac}}\}.
\end{equation}
Pinching any POVM into these two blocks preserves its success
probability, so an optimal POVM may be taken block diagonal.  Its
completeness conditions then separate into
$\sum_aM_a^{\mathrm{ph}}=\mathbb{I}_d$ and
$\sum_am_a^{\mathrm{vac}}=1$, with $m_a^{\mathrm{vac}}\geq0$.  Consequently, the optimization
decouples exactly:
\begin{equation}
\Pval
  = \frac{1}{d}
    \left[P_{\mathrm{ph}}+1-\min_a w_a\right],\qquad
  P_{\mathrm{ph}}
  := \max_{\substack{M_a^{\mathrm{ph}}\succeq0\\
                     \sum_aM_a^{\mathrm{ph}}=\mathbb{I}_d}}
      \sum_a\bra{v_a}M_a^{\mathrm{ph}}\ket{v_a}.
  \label{eq:mdldecomp}
\end{equation}
Here $P_{\mathrm{ph}}$ denotes the optimum in the surviving photonic sector.
Indeed, conditional on observing the vacuum outcome, the optimal
strategy is to guess an index attaining the smallest survival
probability, since
\begin{equation}
  \max_{\substack{m_a\geq0\\ \sum_am_a=1}}
  \sum_a(1-w_a)m_a
  = 1-\min_a w_a .
\end{equation}
The vacuum contribution is therefore not merely a loss: when $w_a$
depends on $a$, the no-detection event carries information about the
oracle value.  Equation~\eqref{eq:wa} also implies
\begin{equation}
  \min_a w_a\leq\bar\tau .
\end{equation}

% --------------------------------------
\subsection{A universal bound}
\label{sec:mdlthm}
% --------------------------------------

\begin{theorem}[Mode-dependent loss]
\label{thm:mdl}
For any transmission profile $\{\tau_z\}$ and any pure response state,
\begin{equation}
  P_{\mathrm{ph}}
  \leq
  \bar\tau(1+C(\eta)),
  \qquad
  \Pval
  \leq
  \frac{\bar\tau(1+C(\eta))+1-\min_a w_a}{d}.
  \label{eq:mdlbound}
\end{equation}
\end{theorem}

\begin{proof}
The discrimination problem defining $P_{\mathrm{ph}}$ has the
semidefinite dual
\begin{equation}
  P_{\mathrm{ph}}
  =
  \min\left\{
    \Tr Y:
    Y\succeq\ket{v_a}\!\bra{v_a}
    \ \text{for every }a
  \right\}
\end{equation}
\cite{holevo1973statistical,Yuen:1975ptn,eldar2003designing}.
Strong duality holds for this finite-dimensional discrimination
problem; only weak duality is needed below.  Consider
\begin{equation}
  Y:=LY_0L,
  \qquad
  Y_0
  :=
  \left(\sum_\chi\sqrt{p_\chi}\right)
  \sum_{\chi\in\widehat A}
  \sqrt{p_\chi}\,
  \ket{\chi}\!\bra{\chi}.
  \label{eq:cert}
\end{equation}
For every $\ket{u}$ and every $a$, Cauchy--Schwarz gives
\begin{equation}
  \left|\bra{u}T_a\ket{\eta}\right|^2
  \leq
  \left(
    \sum_\chi |u_\chi|^2\sqrt{p_\chi}
  \right)
  \left(
    \sum_\chi\sqrt{p_\chi}
  \right)
  =
  \bra{u}Y_0\ket{u},
\end{equation}
where $u_\chi:=\braket{\chi}{u}$ and $|\chi(a)|=1$ have been used.
Thus
\begin{equation}
  Y_0
  \succeq
  T_a\ket{\eta}\!\bra{\eta}T_a^\dagger
\end{equation}
for every $a$ and every phase assignment.  Conjugating by $L$ gives
$Y\succeq\ket{v_a}\!\bra{v_a}$, so $Y$ is dual feasible.

Character states are flat in the value basis:
$|\braket{z}{\chi}|^2=1/d$.  Hence
\begin{equation}
  \bra{\chi}\operatorname{diag}(\tau)\ket{\chi}
  =
  \bar\tau
\end{equation}
and
\begin{equation}
\begin{split}
  \Tr Y
  &=
  \Tr\!\left(
    Y_0\operatorname{diag}(\tau)
  \right)\\
  &=
  \bar\tau
  \left(
    \sum_\chi\sqrt{p_\chi}
  \right)^2
  =
  \bar\tau(1+C(\eta)).
\end{split}
\end{equation}
Weak duality gives the first inequality in
Eq.~\eqref{eq:mdlbound}; the second follows from the exact
decomposition \eqref{eq:mdldecomp}.
\end{proof}

% --------------------------------------
\subsection{Phase--loss compatibility and attainability}
\label{sec:plca}
% --------------------------------------

The certificate is phase independent, but its attainability depends on
the compatibility between the response-state phases and the loss
profile.  To state this precisely, choose arbitrary extensions of
$\theta_\chi$ to indices with $p_\chi=0$ and define
\begin{equation}
  V
  :=
  \sum_{\chi\in\widehat A}
  e^{i\theta_\chi}\ket{\chi}\!\bra{\chi},
  \qquad
  \ket{\eta}=V\ket{\varphi},
  \label{eq:mdlV}
\end{equation}
where $\ket{\varphi}$ is the phase-stripped state of
Eq.~\eqref{eq:Vdef}.

\begin{corollary}[Phase--loss compatibility]
\label{cor:LVcompatibility}
If the phases outside the support of $p$ can be chosen so that
\begin{equation}
  [L,V]=0,
  \label{eq:LVcommute}
\end{equation}
then both bounds in Eq.~\eqref{eq:mdlbound} are attained.  An achieving
POVM on the surviving photonic subspace is
\begin{equation}
  M_a^{\mathrm{ph}}
  =
  V\ket{a}\!\bra{a}V^\dagger .
  \label{eq:LVPOVM}
\end{equation}
\end{corollary}

\begin{proof}
Since $V$ and every $T_a$ are diagonal in the character basis,
\begin{equation}
  [T_a,V]=0
  \qquad
  \text{for every }a\in A.
  \label{eq:TVcommute}
\end{equation}
Together with the hypothesis $[L,V]=0$, this gives
\begin{equation}
  V^\dagger LT_aV
  =
  V^\dagger LVT_a
  =
  LT_a.
  \label{eq:phaseLossCancellation}
\end{equation}
Therefore the POVM \eqref{eq:LVPOVM} achieves
\begin{equation}
\begin{split}
  \sum_a
  \bra{v_a}M_a^{\mathrm{ph}}\ket{v_a}
  &=
  \sum_a
  \left|
    \bra{a}V^\dagger LT_aV\ket{\varphi}
  \right|^2\\
  &=
  \sum_a
  \left|
    \bra{a}LT_a\ket{\varphi}
  \right|^2\\
  &=
  \sum_a
  \tau_a
  \left|
    \bra{a}T_a\ket{\varphi}
  \right|^2\\
  &=
  \varphi_0^2\sum_a\tau_a\\
  &=
  \bar\tau(1+C(\eta)),
\end{split}
\end{equation}
where
$\bra{a}T_a\ket{\varphi}
=\braket{0}{\varphi}
=\varphi_0$
and $d\varphi_0^2=1+C(\eta)$ were used.  Thus the photonic bound is
saturated.  Combining this POVM with the optimal decision on the
vacuum outcome saturates the second bound in
Eq.~\eqref{eq:mdlbound}.
\end{proof}

Constant phases on the support of $p$ are included in
Corollary~\ref{cor:LVcompatibility}: the unspecified phases can then be
chosen so that $V=e^{i\theta}\mathbb{I}$.  Uniform loss is also
included, since $L=\sqrt{\tau}\mathbb{I}$ commutes with every $V$.
The compatibility condition admits the following exact
characterization.

\begin{remark}[Exact commutation criterion]
\label{rem:LVcriterion}
Define the loss-amplitude profile and its Fourier transform by
\begin{equation}
  \ell(z):=\sqrt{\tau_z},
  \qquad
  \widehat{\ell}(\nu)
  :=
  \frac{1}{d}
  \sum_{z\in A}
  \ell(z)\nu(z),
  \qquad
  \nu\in\widehat A,
  \label{eq:lossFourier}
\end{equation}
and let
\begin{equation}
  H_L
  :=
  \left\langle
    \operatorname{supp}\widehat{\ell}
  \right\rangle
  \leq\widehat A
  \label{eq:lossSupportSubgroup}
\end{equation}
be the subgroup generated by its Fourier support.  In the character
basis,
\begin{equation}
  \bra{\chi}L\ket{\chi'}
  =
  \widehat{\ell}(\chi\bar\chi'),
  \label{eq:Lcharentries}
\end{equation}
so
\begin{equation}
  \bra{\chi}[L,V]\ket{\chi'}
  =
  \widehat{\ell}(\chi\bar\chi')
  \left(
    e^{i\theta_{\chi'}}
    -
    e^{i\theta_\chi}
  \right).
  \label{eq:LVmatrix}
\end{equation}
Consequently,
\begin{equation}
  [L,V]=0
  \quad\Longleftrightarrow\quad
  e^{i\theta_\chi}
  \ \text{is constant on every coset of }H_L
  \text{ in }\widehat A.
  \label{eq:LVcosets}
\end{equation}
Consequently, the phases prescribed on $\operatorname{supp}p$ admit a commuting extension if and only if $e^{i\theta_\chi}$ is constant on every nonempty intersection of $\operatorname{supp}p$ with a coset of $H_L$.
Indeed, every nonzero component $\widehat{\ell}(\nu)$ enforces equality of the phase factors at character labels related by multiplication with $\nu$, and repeated such shifts generate exactly the cosets of $H_L$.  Conversely, if the phase factor is constant on those cosets, every matrix element in Eq.~\eqref{eq:LVmatrix} vanishes.

If $H_L=\widehat A$, Eq.~\eqref{eq:LVcosets} forces
$V=e^{i\theta}\mathbb{I}$, so the commuting construction permits only
constant response-state phases.  If $H_L$ is proper, the phase factors
may take different values on its distinct cosets.  Equivalently,
$\ell$, and hence $\tau$, is invariant under the annihilator
\begin{equation}
  H_L^\perp
  :=
  \left\{
    a\in A:
    \nu(a)=1
    \ \text{for every }\nu\in H_L
  \right\},
  \label{eq:lossAnnihilator}
\end{equation}
so nontrivial commuting phase corrections occur precisely when the
loss profile has the corresponding subgroup periodicity.
\end{remark}

For example, let $A=\mathbb{Z}_6$ and take the alternating profile
\begin{equation}
  \tau_z
  =
  \begin{cases}
    0.02, & z\ \text{even},\\
    1,    & z\ \text{odd}.
  \end{cases}
\end{equation}
It satisfies $\tau_{z+2}=\tau_z$, and the Fourier support of
$\ell=\sqrt{\tau}$ generates
\begin{equation}
  H_L=\{0,3\}
  \leq\widehat{\mathbb{Z}}_6.
\end{equation}
Identifying $\widehat{\mathbb Z}_6$ with $\mathbb Z_6$, any phase assignment satisfying
\begin{equation}
e^{i\theta_{k+3}}
=
e^{i\theta_k},
\qquad k\in\mathbb Z_6,
\label{eq:Z6compatiblephases}
\end{equation}
commutes with $L$.
Such a nonconstant phase assignment therefore attains the bound in Eq.~\eqref{eq:mdlbound}. 

Condition~\eqref{eq:LVcosets} exactly characterizes the commuting
attainability mechanism, but Eq.~\eqref{eq:LVcommute} is not claimed
necessary for saturation by an arbitrary POVM.  When it fails, the
certificate remains valid but need not be dual optimal; numerical
results in Appendix~\ref{app:nummdl} show that it can then be
substantially strict.

For a constant profile $\tau_z\equiv\tau$, one has
$L=\sqrt{\tau}\mathbb{I}$ and $w_a=\tau$ for every $a$.  Thus
Corollary~\ref{cor:LVcompatibility} applies to every pure response
state, and the optimal lossless POVM of
Proposition~\ref{prop:mesh} remains optimal.  Consequently,
\begin{equation}
  \Pval
  =
  \frac{\tau(1+C(\eta))+1-\tau}{d}
  =
  \frac{1+\tau C(\eta)}{d},
\end{equation}
recovering the mode-independent loss case deferred in
Appendix~\ref{sec:modes}.

% --------------------------------------
\subsection{Nonuniform loss is never harmful for compatible states}
\label{sec:mdlgain}
% --------------------------------------

\begin{corollary}
\label{cor:mdlgain}
For any pure response state and transmission profile satisfying the
compatibility condition $[L,V]=0$,
\begin{equation}
  \Pval\big(\{\tau_z\}\big)
  \geq
  \Pval\big(\text{uniform }\bar\tau\big)
  =
  \frac{1+\bar\tau C(\eta)}{d},
  \label{eq:mdlgain}
\end{equation}
with equality if and only if $w_a$ is independent of $a$.
\end{corollary}

\begin{proof}
By Corollary~\ref{cor:LVcompatibility} and the exact decomposition
\eqref{eq:mdldecomp},
\begin{equation}
  \Pval\big(\{\tau_z\}\big)
  =
  \frac{
    \bar\tau(1+C(\eta))+1-\min_a w_a
  }{d}.
\end{equation}
Equation~\eqref{eq:wa} gives
$\min_a w_a\leq\bar\tau$, and hence
\begin{equation}
  \Pval\big(\{\tau_z\}\big)
  \geq
  \frac{
    \bar\tau(1+C(\eta))+1-\bar\tau
  }{d}
  =
  \frac{1+\bar\tau C(\eta)}{d}.
\end{equation}
Equality holds precisely when
$\min_a w_a=\bar\tau$.  Since the average of the $w_a$ is
$\bar\tau$, this occurs if and only if every $w_a$ equals
$\bar\tau$.
\end{proof}

Corollary~\ref{cor:mdlgain} shows that compatible nonuniform loss breaks channel covariance while preserving an optimal phase correction.  The resulting dependence of the survival probability on $a$ is exploitable information.  At fixed mean transmission, imbalance across the response modes is therefore never harmful within the compatible family and is beneficial whenever the survival probabilities are nonconstant.  Constant-phase response states are the most immediate members of this family, but Remark~\ref{rem:LVcriterion} identifies additional examples whenever the loss profile has subgroup periodicity.

Finally, the phase corner survives arbitrary mode-dependent loss even without the compatibility condition.  If $p$ is a point mass, then
\begin{equation}
  LT_a\ket{\chi_0}
  =
  \chi_0(a)L\ket{\chi_0},
\end{equation}
so all orbit states coincide up to a global phase.  Moreover,
$|\eta_z|^2=1/d$ gives $w_a=\bar\tau$ for every $a$, so neither the
surviving branch nor the vacuum branch carries information about the
oracle value.  Hence
\begin{equation}
  \Pval=\frac1d
\end{equation}
for every transmission profile.

% -----------------------------------
\section{Non-Abelian Extension and the Limits of Scalar Phase Semantics}
\label{sec:nonabelian}
% -----------------------------------

This appendix provides the matrix-Fourier details, proofs, equality conditions, and the $S_3$ example supporting the finite-group results of Sec.~\ref{sec:compact_nonabelian}.  The final subsection delineates the parts of the Abelian robustness theory for which no comparably explicit non-Abelian extension is established here.

% --------------------------------------
\subsection{Matrix-valued Fourier structure and the ideal readout resource}
\label{sec:naresource}
% --------------------------------------

Let $G$, $\mathbb{C}[G]$, the left translations $T_a$, and the irreducible representations $\rho_\lambda$, $\lambda\in\widehat G$, be as defined in Sec.~\ref{sec:compact_nonabelian}, with $d=|G|$ and $d_\lambda=\dim\rho_\lambda$. The regular representation decomposes as
\begin{equation}
  \mathbb{C}[G]
  \cong
  \bigoplus_{\lambda\in\widehat G}
  \mathbb{C}^{d_\lambda}\otimes\mathbb{C}^{d_\lambda},
  \qquad
  \sum_{\lambda\in\widehat G}d_\lambda^2=d,
  \label{eq:naregular}
\end{equation}
where the second factor is the multiplicity space \cite{serre_1977,Terras_1999}.

Writing $\eta(z):=\braket{z}{\eta}$, we use the normalized matrix-valued Fourier coefficients defined in Eq.~\eqref{eq:main-na-fourier}.
By Eq.~\eqref{eq:main-na-translation}, left translation acts on the first factor in Eq.~\eqref{eq:naregular}, while the second factor is the multiplicity space. Using the above normalization, Plancherel's identity and the regular-character identity are
\begin{equation}
  \sum_{\lambda\in\widehat G}
  \|\widehat\eta(\lambda)\|_F^2=
  \sum_{z}
  |\eta(z)|^2=1,
  \qquad
  \sum_{\lambda}d_\lambda \Tr\rho_\lambda(z)=d\ \delta_{z,e},
  \label{eq:naplancherel}
\end{equation}
and Fourier inversion at the identity element $e\in G$ gives
\begin{equation}
  \eta(e)
  =
  \sum_{\lambda\in\widehat G}
  \sqrt{\frac{d_\lambda}{d}}\,
  \Tr\widehat\eta(\lambda).
  \label{eq:nainversione}
\end{equation}

Consider the uniformly weighted orbit ensemble
\begin{equation}
  \mathcal{E}_\eta^G
  =
  \left\{
  \frac1d,\,
  T_a\ket{\eta}
  \right\}_{a\in G}.
  \label{eq:naensemble}
\end{equation}

For this ensemble, the optimal pure-state readout probability is given by Eq.~\eqref{eq:main-nonabelian}.  This is the covariant maximum-likelihood solution of Chiribella \emph{et al.} \cite{Chiribella:2004wyf}, specialized to the left regular representation. Each irreducible representation $\lambda$ occurs with multiplicity $d_\lambda$, and the singular values connecting its representation and multiplicity spaces produce the trace norm $\|\widehat\eta(\lambda)\|_1$.  Equivalently, the result follows from the square-root measurement for pure-state geometrically uniform ensembles \cite{eldar2004optimal,Zhou:2025iac}.

The pure-state readout resource $C_G(\eta)$ is defined in
Eq.~\eqref{eq:main-na-resource} and obeys the operational identity Eq.~\eqref{eq:main-na-master}.  Unlike the Abelian quantity in Eq.~\eqref{eq:Cdef}, $C_G$ is not a basis-dependent $\ell_1$-coherence norm.

\begin{proposition}[Elementary properties]
\label{prop:naprops}
For every pure response state on $\mathbb{C}[G]$:
\begin{enumerate}
\item $0\leq C_G\leq d-1$;
\item $C_G=0$ if and only if the Fourier support is contained in a single one-dimensional irreducible representation;
\item $C_G=d-1$ if and only if the squared singular-value distribution
      equals the Plancherel reference distribution defined below;
\item for Abelian $G$, $C_G=C$;
\item $C_G$ is invariant under left translation: $C_G(T_a\eta) = C_G(\eta)$ for every $a \in G$.
\end{enumerate}
\end{proposition}

\begin{proof}
Let $\sigma_{\lambda,k}$,
$k=1,\ldots,d_\lambda$, be the singular values of
$\widehat\eta(\lambda)$, including zeros.  Define
\begin{equation}
  P(\lambda,k):=\sigma_{\lambda,k}^2,
  \qquad
  m(\lambda,k):=\frac{d_\lambda}{d}.
  \label{eq:plancherel}
\end{equation}
By Eq.~\eqref{eq:naplancherel} and
$\sum_\lambda d_\lambda^2=d$, both $P$ and $m$ are probability
distributions.  Moreover, the classical root fidelity is given by
\begin{equation}
\begin{aligned}
  F_{\mathrm{cl}}(P,m)
  &:=
  \sum_{\lambda,k}
  \sqrt{P(\lambda,k)m(\lambda,k)}\\
  &=
  \frac{1}{\sqrt d}
  \sum_\lambda
  \sqrt{d_\lambda}\,
  \|\widehat\eta(\lambda)\|_1,
\end{aligned}
\label{eq:nafidelity}
\end{equation}
which leads to
\begin{equation}
  \Pval
  =
  F_{\mathrm{cl}}(P,m)^2
  =
  2^{-D_{1/2}(P\|m)},
  \label{eq:divergence}
\end{equation}
where
\begin{equation}
  D_{1/2}(P\|m)
  :=
  -2\log_2F_{\mathrm{cl}}(P,m).
\end{equation}
Hence $\Pval\leq1$, with equality if and only if $P=m$, giving the
upper bound and item~(3).

For the lower bound, use
$\|X\|_1\geq\|X\|_F$ and $d_\lambda\geq1$:
\begin{equation}
\begin{aligned}
  \sum_\lambda\sqrt{d_\lambda}\,
  \|\widehat\eta(\lambda)\|_1
  &\geq
  \sum_\lambda
  \|\widehat\eta(\lambda)\|_F\\
  &\geq
  \left(
    \sum_\lambda
    \|\widehat\eta(\lambda)\|_F^2
  \right)^{1/2}
  =1.
\end{aligned}
\end{equation}
Equality throughout requires support on a single block,
$d_\lambda=1$, giving items~(1) and~(2).  If $G$ is Abelian, every
$d_\lambda=1$ and
$\|\widehat\eta(\lambda)\|_1=\sqrt{p_\lambda}$, so
Eq.~\eqref{eq:main-na-resource} reduces to Eq.~\eqref{eq:Cpure}. Finally, Eq.~\eqref{eq:main-na-translation} and unitary invariance of the trace norm imply 
$\|\widehat{T_a\eta}(\lambda)\|_1 =\|\widehat{\eta}(\lambda)\|_1$ 
for every $\lambda$; hence Eq.~\eqref{eq:main-na-resource} gives
$C_G(T_a\eta)=C_G(\eta)$.
\end{proof}

The marginal of the reference distribution is
\begin{equation}
  \sum_{k=1}^{d_\lambda}m(\lambda,k)
  =
  \frac{d_\lambda^2}{d},
\end{equation}
the Plancherel measure on $\widehat G$
\cite{serre_1977,Terras_1999,Moore:2005rmn}. Perfect value readout is therefore equivalent to Plancherel-distributed weight across irreducible representations and uniform squared singular values within each block.

\begin{proposition}[Multiplicativity]
\label{prop:namult}
Let $G$ and $H$ be finite groups, and prepare the product response state $\ket{\eta_G}\otimes\ket{\eta_H}$ for $G\times H$.  Then
\begin{equation}
  1+C_{G\times H}(\eta_G\otimes\eta_H)
  =
  \bigl(1+C_G(\eta_G)\bigr)
  \bigl(1+C_H(\eta_H)\bigr),
  \label{eq:naCmult}
\end{equation}
and
\begin{equation}
  \Pval^{G\times H}
  =
  \Pval^G\Pval^H.
  \label{eq:naPmult}
\end{equation}
\end{proposition}

\begin{proof}
The irreducible representations of $G\times H$ are
$\rho_\lambda\otimes\rho_\mu$, with dimension
$d_\lambda d_\mu$, and
\begin{equation}
  \widehat{\eta_G\otimes\eta_H}(\lambda,\mu)
  =
  \widehat\eta_G(\lambda)\otimes\widehat\eta_H(\mu).
\end{equation}
Since the trace norm is multiplicative under tensor products,
\begin{equation}
\begin{aligned}
  &\sum_{\lambda,\mu}
  \sqrt{d_\lambda d_\mu}\,
  \left\|
    \widehat\eta_G(\lambda)
    \otimes
    \widehat\eta_H(\mu)
  \right\|_1\\
  &\quad=
  \left(
    \sum_\lambda
    \sqrt{d_\lambda}\,
    \|\widehat\eta_G(\lambda)\|_1
  \right)
  \left(
    \sum_\mu
    \sqrt{d_\mu}\,
    \|\widehat\eta_H(\mu)\|_1
  \right).
\end{aligned}
\end{equation}
Equations~\eqref{eq:naCmult} and \eqref{eq:naPmult} follow from
Eqs.~\eqref{eq:main-na-resource} and \eqref{eq:main-nonabelian}.
\end{proof}

% --------------------------------------
\subsection{Scalar phase semantics and complementarity}
\label{sec:naphase}
% --------------------------------------

Starting from the scalar-kickback condition \eqref{eq:main-na-kickback}, normalization and unitarity give $|\omega(a)|=1$ for every $a\in G$.  Moreover, $T_{ab}=T_aT_b$ implies
\begin{equation}
  \omega(ab)=\omega(a)\omega(b),
\end{equation}
so $\omega$ is a one-dimensional unitary representation  of $G$. Every such representation is trivial on the commutator subgroup $[G,G]$ and therefore factors through the abelianization
\begin{equation}
  G_{\mathrm{ab}}:=G/[G,G].
  \label{eq:abelianization}
\end{equation}
Thus nontrivial scalar phase semantics exists precisely when
$G_{\mathrm{ab}}$ admits a nontrivial character.  In particular, a
perfect group has only the trivial scalar phase mode.

This restriction is also visible in $C_G$.  If the Fourier support lies
in one irreducible representation $\lambda$ and
$\widehat\eta(\lambda)$ has rank one, then
\begin{equation}
  1+C_G=d_\lambda,
  \qquad
  \Pval=\frac{d_\lambda}{d}.
  \label{eq:narankone}
\end{equation}
Thus a rank-one state supported on a higher-dimensional irreducible block has nonzero readout resource and does not realize scalar kickback.  The scalar phase corner occurs only for support on a one-dimensional irreducible representation.

\begin{proof}[Proof of Eqs.~\eqref{eq:main-na-tradeoff} and \eqref{eq:main-na-equality}]
For each non-target block,
\begin{equation}
  \|\widehat\eta(\lambda)\|_1
  \leq
  \sqrt{d_\lambda}\,
  \|\widehat\eta(\lambda)\|_F,
  \label{eq:natracefrobenius}
\end{equation}
with equality if and only if its singular values are equal.  Therefore,
\begin{equation}
\begin{split}
  \sum_{\lambda\neq\lambda_0}
  \sqrt{d_\lambda}\,
  \|\widehat\eta(\lambda)\|_1
  &\leq
  \sum_{\lambda\neq\lambda_0}
  d_\lambda
  \|\widehat\eta(\lambda)\|_F\\
  &\leq
  \sqrt{
    \sum_{\lambda\neq\lambda_0}
    d_\lambda^2
  }
  \sqrt{
    \sum_{\lambda\neq\lambda_0}
    \|\widehat\eta(\lambda)\|_F^2
  }\\
  &=
  \sqrt{(d-1)(1-\Phi)}.
\end{split}
\end{equation}
Here $\sum_{\lambda\neq\lambda_0}d_\lambda^2=d-1$ because $d_{\lambda_0}=1$, while Plancherel normalization gives $\sum_{\lambda\neq\lambda_0} \|\widehat\eta(\lambda)\|_F^2=1-\Phi$. Because $d_{\lambda_0}=1$, the target block contributes exactly $\sqrt{\Phi}$.  Adding this contribution and using Eq.~\eqref{eq:main-na-resource} proves Eq.~\eqref{eq:main-na-tradeoff}.

Equality in the second inequality requires
\begin{equation}
  \|\widehat\eta(\lambda)\|_F
  =
  d_\lambda
  \sqrt{\frac{1-\Phi}{d-1}}
  \qquad
  (\lambda\neq\lambda_0).
\end{equation}
Equality in Eq.~\eqref{eq:natracefrobenius} requires the $d_\lambda$ singular values of each non-target block to be equal. Consequently,
\begin{equation}
  \sigma_{\lambda,k}^2
  =
  \frac{1}{d_\lambda}
  \|\widehat\eta(\lambda)\|_F^2
  =
  \frac{d_\lambda(1-\Phi)}{d-1},
\end{equation}
which is Eq.~\eqref{eq:main-na-equality}.  Conversely, this singular-value condition saturates both inequalities, proving the stated if-and-only-if condition.
\end{proof}

A saturating family is
\begin{equation}
  \widehat\eta(\lambda_0)
  =
  \sqrt{\Phi}\,e^{i\vartheta},
  \qquad
  \widehat\eta(\lambda)
  =
  \sqrt{\frac{(1-\Phi)d_\lambda}{d-1}}\,
  U_\lambda,
  \quad
  \lambda\neq\lambda_0,
  \label{eq:nasat}
\end{equation}
where $\vartheta$ is arbitrary and every $U_\lambda$ is unitary.  The optimal boundary therefore depends only on the group order $d$, while the saturating states encode the group structure through the Plancherel weights and the singular directions of the irreducible blocks.

% --------------------------------------
\subsection{The smallest non-Abelian example: \texorpdfstring{$S_3$}{S3}}
\label{sec:s3}
% --------------------------------------

The smallest non-Abelian finite group is the permutation group $S_3$
of order six, isomorphic to the dihedral group $D_3$ of symmetries of
the triangle. It has three inequivalent irreducible representations: the trivial and sign representations of dimension one, and the standard representation of dimension two.  Write the corresponding Fourier coefficients as \begin{equation}
  a:=\widehat\eta(\mathrm{triv}),
  \qquad
  b:=\widehat\eta(\mathrm{sgn}),
  \qquad
  M:=\widehat\eta(\mathrm{std})
  \in\mathbb{C}^{2\times2}.
  \label{eq:s3blocks}
\end{equation}
Plancherel normalization becomes
\begin{equation}
  |a|^2+|b|^2+\|M\|_F^2=1,
  \label{eq:s3norm}
\end{equation}
and Eq.~\eqref{eq:main-nonabelian} gives
\begin{equation}
  \Pval
  =
  \frac16
  \left(
  |a|+|b|+\sqrt2\,\|M\|_1
  \right)^2.
  \label{eq:s3Pval}
\end{equation}

Taking the trivial representation as the phase target,
$\Phi=|a|^2$, Eq.~\eqref{eq:main-na-tradeoff} becomes
\begin{equation}
  \Pval
  \leq
  \frac16
  \left[
  \sqrt{\Phi}+\sqrt{5(1-\Phi)}
  \right]^2.
  \label{eq:s3tradeoff}
\end{equation}
Equality holds if and only if
\begin{equation}
  |b|^2=\frac{1-\Phi}{5},
  \qquad
  \|M\|_F^2=\frac{4(1-\Phi)}5,
  \label{eq:s3weights}
\end{equation}
and the two singular values of $M$ must be equal.  Equivalently,
\begin{equation}
  M
  =
  \sqrt{\frac{2(1-\Phi)}5}\,U
  \label{eq:s3sat}
\end{equation}
for an arbitrary $2\times2$ unitary $U$.

The commutator subgroup of $S_3$ is $A_3$, so
\begin{equation}
  S_3/[S_3,S_3]=S_3/A_3\cong\mathbb{Z}_2.
\end{equation}
Accordingly, only the trivial and sign representations provide scalar
phase kickback.  The standard block contributes genuinely
non-Abelian readout capability but cannot define scalar phase
semantics.

% --------------------------------------
\subsection{Robustness results}
\label{sec:narobustness}
% --------------------------------------

For a mixed response state $\rho$, the operational resource
$C_G(\rho)$ is defined in Eq.~\eqref{eq:main-na-mixed-resource}.
For pure states this definition agrees with Eq.~\eqref{eq:main-na-resource}.  For general mixed states it should be understood as an operational asymmetry resource, not as a matrix-entry coherence norm.

\begin{proposition}[Monotonicity under covariant noise]
\label{prop:namonotonicity}
If $\Lambda$ is translation covariant,
\begin{equation}
  \Lambda(T_aXT_a^\dagger)
  =
  T_a\Lambda(X)T_a^\dagger
  \qquad\forall a\in G,
  \label{eq:nacovariance}
\end{equation}
then
\begin{equation}
  C_G(\Lambda(\rho))\leq C_G(\rho).
  \label{eq:namonotonicity}
\end{equation}
\end{proposition}

\begin{proof}
Any POVM applied after $\Lambda$ induces, through the adjoint channel
$\Lambda^\dagger$, a POVM on the input ensemble with the same success
probability.  State discrimination cannot therefore improve under
$\Lambda$.  Covariance ensures that the output hypotheses are the
orbit of the single state $\Lambda(\rho)$, so
\begin{equation}
  P_{\mathrm{val}}\!\left(\Lambda(\rho)\right)
  \leq
  P_{\mathrm{val}}(\rho).
\end{equation}
Equation~\eqref{eq:namonotonicity} follows from
Eq.~\eqref{eq:main-na-mixed-resource}.
\end{proof}

The resource vanishes exactly when the orbit states are identical:
\begin{equation}
  C_G(\rho)=0
  \quad\Longleftrightarrow\quad
  T_a\rho T_a^\dagger=\rho
  \quad\forall a\in G.
  \label{eq:nazero}
\end{equation}
Identical hypotheses give the random-guessing probability $1/d$. Conversely, for an equal-prior ensemble, if two hypotheses $\rho_a$ and $\rho_b$ are distinct, measuring the positive part of $\rho_b-\rho_a$ and otherwise guessing $a$ gives success strictly greater than $1/d$.  Hence $C_G(\rho)=0$ if and only if the orbit is trivial.

A particularly important covariant channel retains the Abelian linear law without modification.
\begin{proposition}[Non-Abelian depolarizing law]
\label{prop:nadepolarizing}
For
\begin{equation}
  \Lambda_\tau(\rho)
  =
  \tau\rho+(1-\tau)\frac{\mathbb{I}}d,
  \qquad 0\leq\tau\leq1,
  \label{eq:nadepolchannel}
\end{equation}
one has
\begin{equation}
  P_{\mathrm{val}}(\Lambda_\tau(\rho))
  =
  \frac{1+\tau C_G(\rho)}d.
  \label{eq:nadepollaw}
\end{equation}
\end{proposition}

\begin{proof}
For any POVM $\{M_a\}$,
\begin{equation}
\begin{split}
  \frac1d\sum_a
  \Tr\!\left[
    M_a\Lambda_\tau(T_a\rho T_a^\dagger)
  \right]
  &=
  \tau\frac1d\sum_a
  \Tr\!\left[
    M_aT_a\rho T_a^\dagger
  \right]\\
  &\quad+
  \frac{1-\tau}{d^2}
  \sum_a\Tr M_a\\
  &=
  \tau\frac1d\sum_a
  \Tr\!\left[
    M_aT_a\rho T_a^\dagger
  \right]
  +\frac{1-\tau}{d}.
\end{split}
\end{equation}
The second term is measurement independent, so the optimizer for the
original orbit remains optimal.  Thus
\begin{equation}
  P_{\mathrm{val}}\!\left(\Lambda_\tau(\rho)\right)
  =
  \tau P_{\mathrm{val}}(\rho)
  +\frac{1-\tau}{d},
\end{equation}
which is Eq.~\eqref{eq:nadepollaw}.
\end{proof}

The same argument applies to mode-independent flagged loss,
\begin{equation}
  \rho_a
  \longmapsto
  \tau\rho_a+(1-\tau)
  \ket{\mathrm{vac}}\!\bra{\mathrm{vac}}.
\end{equation}
Because the surviving and vacuum sectors are orthogonal, the POVM may
be taken block diagonal.  The common vacuum branch contributes the
measurement-independent success probability $(1-\tau)/d$, while the
surviving branch contributes $\tau\Pval(\rho)$.  Hence
\begin{equation}
  P_{\mathrm{val}}^{(\mathrm{out})}
  =
  \tau\Pval(\rho)+\frac{1-\tau}{d}
  =
  \frac{1+\tau C_G(\rho)}{d}.
  \label{eq:nauniformloss}
\end{equation}

The first-order suppression of coherent apparatus errors does not
depend on commutativity.

\begin{proposition}[Non-Abelian stationarity]
\label{prop:nastationarity}
Suppose an optimal projective readout is realized by a unitary $U_0$
followed by measurement in a fixed orthonormal basis
$\{\ket{a}\}_{a\in G}$.  If the apparatus instead implements
\begin{equation}
  U(\epsilon)=e^{i\epsilon K}U_0,
  \qquad K=K^\dagger,
  \label{eq:naapparatus}
\end{equation}
let
\begin{equation}
  P(\epsilon)
  :=
  \frac{1}{d}\sum_{a\in G}
  \left|
    \bra{a}U(\epsilon)T_a\ket{\eta}
  \right|^2 .
\end{equation}
Then
\begin{equation}
  \left.
  \frac{dP}{d\epsilon}
  \right|_{\epsilon=0}
  =0.
  \label{eq:nastationarity}
\end{equation}
\end{proposition}

\begin{proof}
The measurement implemented at $\epsilon=0$ is globally optimal, so
$P(\epsilon)\leq P(0)$ for all sufficiently small positive and
negative $\epsilon$.  Since $P(\epsilon)$ is differentiable,
$\epsilon=0$ is a local maximum and its first derivative vanishes.
\end{proof}

Thus coherent mismatch is at least second order whenever the optimal covariant measurement is implemented interferometrically.  In the Abelian theory, the optimal interferometer depends on the scalar Fourier phases and not on the weights.  In the non-Abelian theory, the corresponding directional data are the polar factors in
\begin{equation}
  \widehat\eta(\lambda)
  =
  U_\lambda
  \left|\widehat\eta(\lambda)\right|.
  \label{eq:napolar}
\end{equation}
The singular values determine the ideal success probability through
Eq.~\eqref{eq:main-nonabelian}, while the polar factors determine how the equivalent representation components must be aligned by the measurement.  This suggests fixed-apparatus families obtained by varying singular values while retaining the same polar factors. A complete treatment of rank-deficient blocks and the resulting measurement nonuniqueness remains unresolved.

% --------------------------------------
\subsection{Mode-dependent loss for an arbitrary finite group}
\label{sec:namdl}
% --------------------------------------

Let $L$, $\bar\tau$, $\ket{v_a}$, and $w_a$ be defined as in Eq.~\eqref{eq:main-na-loss-definitions}, with $0\leq\tau_z\leq1$, and include an orthogonal vacuum flag as in Appendix~\ref{sec:mdl}.  Since left multiplication permutes the group basis,
\begin{equation}
  \sum_{a\in G}w_a
  =
  \sum_{z\in G}\tau_z
  =
  d\bar\tau.
  \label{eq:nawaverage}
\end{equation}
The photonic and vacuum sectors are orthogonal, so the argument leading
to Eq.~\eqref{eq:mdldecomp} does not require commutativity and gives
\begin{equation}
  P_{\mathrm{val}}
  =
  \frac1d
  \left[
  P_{\mathrm{ph}}+1-\min_a w_a
  \right],
  \label{eq:namdldecomp}
\end{equation}
where
\begin{equation}
  P_{\mathrm{ph}}
  :=
  \max_{\substack{M_a^{\mathrm{ph}}\succeq0\\
                   \sum_aM_a^{\mathrm{ph}}=\mathbb I_d}}
  \sum_a
  \bra{v_a}M_a^{\mathrm{ph}}\ket{v_a}.
  \label{eq:naPph}
\end{equation}

\begin{proof}[Proof of Eq.~\eqref{eq:main-na-loss}]
By strong duality for finite-dimensional minimum-error discrimination, the lossless photonic problem admits an optimal dual certificate $Y_0$ satisfying
\begin{equation}
  Y_0
  \succeq
  T_a\ket\eta\!\bra\eta T_a^\dagger
  \quad\forall a,
  \qquad
  \Tr Y_0=1+C_G(\eta).
  \label{eq:nadual0}
\end{equation}
Averaging $Y_0$ over left translations preserves feasibility,
optimality, and trace.  We may therefore take
\begin{equation}
  T_aY_0T_a^\dagger=Y_0
  \qquad\forall a\in G.
  \label{eq:nadualinvariant}
\end{equation}
Because the left regular action is transitive on the group basis,
Eq.~\eqref{eq:nadualinvariant} forces all diagonal entries of $Y_0$
to be equal:
\begin{equation}
  \bra zY_0\ket z
  =
  \frac{\Tr Y_0}{d}
  =
  \frac{1+C_G(\eta)}d.
  \label{eq:nadualdiag}
\end{equation}

Now set
\begin{equation}
  Y:=LY_0L.
  \label{eq:nalossdual}
\end{equation}
Conjugating Eq.~\eqref{eq:nadual0} by $L$ shows that
\begin{equation}
  Y\succeq\ket{v_a}\!\bra{v_a}
  \qquad\forall a,
\end{equation}
so $Y$ is feasible for the lossy photonic dual.  Its trace is
\begin{equation}
\begin{split}
  \Tr Y
  &=
  \Tr(Y_0L^2)\\
  &=
  \sum_z\tau_z\bra zY_0\ket z\\
  &=
  \bar\tau\bigl(1+C_G(\eta)\bigr).
\end{split}
\end{equation}
Weak duality proves the first inequality in
Eq.~\eqref{eq:main-na-loss}, and the second follows from
Eq.~\eqref{eq:namdldecomp}.
\end{proof}

A useful sufficient condition for attainability is the non-Abelian
analogue of constant Fourier phases.

\begin{definition}[Fourier-block phase alignment]
\label{def:naphasealigned}
A pure response state is \emph{Fourier-block phase aligned} if there
exists $\vartheta\in\mathbb R$ such that
\begin{equation}
  e^{-i\vartheta}\widehat\eta(\lambda)\succeq0
  \qquad\forall\lambda\in\widehat G.
  \label{eq:naphasealigned}
\end{equation}
\end{definition}

\begin{corollary}[Attainability and beneficial imbalance]
\label{cor:namdlgain}
For every Fourier-block phase-aligned response state, both
inequalities in Eq.~\eqref{eq:main-na-loss} are equalities.
At fixed mean transmission,
\begin{equation}
  P_{\mathrm{val}}(\{\tau_z\})
  -
  P_{\mathrm{val}}(\text{uniform }\bar\tau)
  =
  \frac{\bar\tau-\min_a w_a}{d}
  \geq0,
  \label{eq:namdlgain}
\end{equation}
with equality if and only if $w_a$ is independent of $a$.
\end{corollary}

\begin{proof}
By Fourier inversion \eqref{eq:nainversione} and phase alignment,
\begin{equation}
\begin{split}
  \sqrt d\,|\eta(e)|
  &=
  \sum_\lambda
  \sqrt{d_\lambda}\,
  \left|\Tr\widehat\eta(\lambda)\right|\\
  &=
  \sum_\lambda
  \sqrt{d_\lambda}\,
  \|\widehat\eta(\lambda)\|_1\\
  &=
  \sqrt{1+C_G(\eta)}.
\end{split}
\end{equation}
The group-basis POVM
$M_a^{\mathrm{ph}}=\ket a\!\bra a$
therefore achieves
\begin{equation}
\begin{split}
  \sum_a
  \left|\bra aLT_a\ket\eta\right|^2
  &=
  \sum_a\tau_a
  \left|\bra aT_a\ket\eta\right|^2\\
  &=
  \sum_a\tau_a|\eta(e)|^2\\
  &=
  \bar\tau\bigl(1+C_G(\eta)\bigr),
\end{split}
\end{equation}
which saturates the first inequality in
Eq.~\eqref{eq:main-na-loss}.  Combining this measurement with the
optimal vacuum decision saturates the second inequality.

For uniform transmission, $w_a=\bar\tau$ and
Eq.~\eqref{eq:nauniformloss} gives
\begin{equation}
  P_{\mathrm{val}}(\text{uniform }\bar\tau)
  =
  \frac{1+\bar\tau\ C_G(\eta)}{d}.
\end{equation}
Subtracting this expression from the attained second line of
Eq.~\eqref{eq:main-na-loss} gives Eq.~\eqref{eq:namdlgain}.  Since the
average of the $w_a$ is $\bar\tau$, equality
$\min_a w_a=\bar\tau$ holds if and only if every $w_a$ equals
$\bar\tau$.
\end{proof}

Thus the vacuum-information mechanism of Appendix~\ref{sec:mdlgain} is not an Abelian effect.  It follows from the transitivity of the regular group action and persists for every finite group.

% --------------------------------------
\subsection{Open problems beyond the present extension}
\label{sec:naopen}
% --------------------------------------

The results above extend the ideal readout resource, selected noise laws, first-order coherent stationarity, and mode-dependent loss. Several parts of the Abelian theory do not yet have comparably explicit non-Abelian forms.

For a uniformly weighted finite-group orbit, the operational quantity
$C_G(\rho)=d\Pval(\rho)-1$ equals the established robustness of
asymmetry \cite{Piani:2016lse,Bagan:2018uuk}:
\begin{equation}
C_G(\rho)=R_G(\rho),\qquad
R_G(\rho):=\min\left\{s\geq0:
\rho\leq(1+s)\delta,\quad
T_a\delta T_a^\dagger=\delta\ \forall a\in G\right\},
\label{eq:h1nonabelianrobustness}
\end{equation}
where $\delta$ is a density operator. The robustness interpretation
is therefore already known. What is not established here is a closed
evaluation for general mixed states in terms of their Fourier blocks.
The relevant covariant seed constraints involve partial traces over
irreducible representation spaces, rather than scalar diagonal
constraints.

Second, non-Abelian covariance does not decompose operator space into one-dimensional character modes.  The adjoint action contains irreducible sectors with multiplicities, and a covariant channel may mix operators within each equivalent sector.  Classifying channels that scale $C_G$ by a single contraction parameter would provide a non-Abelian counterpart of Corollary~\ref{cor:linear}.  The exact depolarizing and uniform-loss laws above are special cases.

Third, Proposition~\ref{prop:nastationarity} eliminates the first-order coherent correction but does not determine the quadratic coefficient. A full extension of Appendix~\ref{app:second} would require representation- and multiplicity-resolved correlations replacing the Abelian functions $Q(w,w')$, $R(w)$, and $V(w)$.  The Abelian shift-profile simulations do not transfer directly: left and right offsets cannot generally be identified as in the Abelian case, and the physical adjacency of an interferometer need not respect conjugacy classes or any intrinsic geometry of $G$.  Determining when the coefficient reduces to a scalar multiple of $C_G$ therefore requires a new correlation model and corresponding simulations.

Finally, the Fourier-block phase-alignment condition \eqref{eq:naphasealigned} is sufficient for attaining the mode-dependent-loss bound, but is not claimed necessary.  A complete characterization of the saturating response states remains open. This question concerns matrix-valued phase alignment and is distinct from scalar phase kickback, which remains confined to the abelianization.

% -----------------------------------
\section{Photonic Implementation and Numerical Validation}
\label{sec:photonic_numerical}
The optical construction implements the response analyzer and assumes
the controlled oracle on the query and response registers. We perform
all calculations below for explicitly stored states and mesh settings.
One programmed analyzer is used at each dimension, and the $d=8$
configuration is shared by Tables~\ref{tab:profiles}--\ref{tab:d8scan}
and Fig.~\ref{fig:quadraticscaling}. The accompanying code generates
the tables and figure directly from its numerical outputs.

The response-register encoding also applies to a finite group $G$:
the $d=|G|$ spatial modes carry $\mathbb C[G]$, and each fixed
left translation is the permutation $z\mapsto az$. These need not
be cyclic permutations, but can be realized by permutation networks.
An optimal projective response readout can be synthesized by a
universal $d$-mode interferometer. A controlled oracle additionally
requires the query encoding and coherent control of $X\otimes Y$;
those are assumed here. The numerical calculations below are restricted
to cyclic Abelian response groups.

\subsection{Methods and solver settings}
\label{app:nummethods}
We solve minimum-error discrimination using the unrestricted primal
and an independently formulated dual,
\begin{equation}
P_{\mathrm{val}}=\max_{M_a\succeq0,\ \sum_aM_a=\mathbb I}
\frac1d\sum_a\Tr(\rho_aM_a),\qquad
P_{\mathrm{val}}=\min_{Y\succeq\rho_a/d\ \forall a}\Tr Y.
\label{eq:sdp}
\end{equation}
The surviving-sector objective $P_{\mathrm{ph}}$ uses the same pair
without the factors $1/d$. We use CVXPY~\cite{diamond2016cvxpy} with
Clarabel~\cite{goulart2026clarabel}, requested absolute and relative gap and
feasibility tolerances $10^{-10}$, and at most $300$ iterations.
All complex Hermitian matrices, solver statuses, and objectives are
stored. We evaluate feasibility in the original matrices through
\begin{equation}
\begin{aligned}
r_{\mathrm{comp}}&=\left\|\sum_aM_a-\mathbb I\right\|_F,&
e_{\mathrm{pr}}&=\max\{0,-\min_a\lambda_{\min}(M_a)\},\\
e_{\mathrm{du}}&=\max\{0,-\min_a\lambda_{\min}(Y-\rho_a/d)\},&
g&=\Tr Y-\frac1d\sum_a\Tr(\rho_aM_a).
\end{aligned}
\label{eq:numerical-certificate-diagnostics}
\end{equation}
For the photonic objective, $\rho_a/d$ is replaced by
$|v_a\rangle\!\langle v_a|$. A requested tolerance is distinct
from these measured residuals. Some runs at the stringent settings
terminate with status \texttt{optimal\_inaccurate}; Table~\ref{tab:diagnostics}
therefore gives the measured diagnostics rather than relying on status alone.

We also construct numerically feasible certificates from the returned
matrices. Let $M_a^+$ be the positive part of $M_a$ and
$S=\sum_aM_a^+$, which is nonsingular in every instance. Set
\begin{equation}
\widetilde M_a=S^{-1/2}M_a^+S^{-1/2},\qquad
\widetilde Y=Y+(e_{\mathrm{du}}+10^{-11})\mathbb I.
\label{eq:numerical-feasibility-repair}
\end{equation}
Evaluating their objectives gives $P_L$ and $P_U$, respectively.
The repaired matrices, their eigenvalue checks, and the enclosure widths
$P_U-P_L$ are included with the raw solver outputs. This procedure gives
floating-point certificates; it is not an interval-arithmetic proof.

\subsection{Attainment and strictness of the mixed-state bound}
\label{app:numsdp}
At $d=4$ the pure and dephased examples share the parent
\begin{equation}
p=(0.7,0.1,0.1,0.1),\qquad\theta=(0,0.2,-0.4,0.7),\qquad
\rho_\gamma=e^{-\gamma}|\eta\rangle\!\langle\eta|
 +(1-e^{-\gamma})\Delta(|\eta\rangle\!\langle\eta|).
\label{eq:table-common-parent}
\end{equation}
The exact coherence is $C(\eta)=2.187450787\ldots$, and
$C_R(\rho_\gamma)=C(\rho_\gamma)=e^{-\gamma}C(\eta)$.
The phase-factorization argument of Theorem~\ref{thm:mixed}
certifies equality analytically for these three rows. The generic states are
$\rho=ZZ^\dagger/\Tr(ZZ^\dagger)$ in the character basis, where the
real and imaginary entries of the $4\times4$ matrix $Z$ are independent
standard normal variables. We take three successive draws with
NumPy's \texttt{default\_rng} seed $20261502$ and retain the matrices.

For every row, we independently compute robustness using
\begin{equation}
1+C_R(\rho)=\min_{t\in\mathbb R^d}
 \left\{\sum_\chi t_\chi:\operatorname{diag}(t)\succeq\rho\right\},
\qquad
1+C_R(\rho)=\max_{X\succeq0,\ X_{\chi\chi}=1}\Tr(\rho X).
\label{eq:numerical-robustness-sdp}
\end{equation}
The largest difference between the independently obtained
$P_{\mathrm{val}}$ and $[1+C_R]/d$ is
$\MixedIdentityResidual$. The largest robustness primal--dual gap
is $\RobustnessMaxGap$. As a separate solver check, SCS~\cite{o2016conic}
with tolerance $10^{-9}$ and at most $10^5$ iterations solves the
three generic discrimination primals; their objective values differ
from Clarabel's by at most $\SCSMaxDifference$.

\begin{table}[!htbp]
\centering
\begin{ruledtabular}
\begin{tabular}{lccccc}
state & $C_R$ & $P_{\rm val}$ & $[1+C_R]/4$ & $[1+C]/4$ & $\ell_1$ gap\\
\colrule
parent, $\gamma=0$ & $2.187451$ & $0.796863$ & $0.796863$ & $0.796863$ & $0.000000$\\
parent, $\gamma=0.3$ & $1.620503$ & $0.655126$ & $0.655126$ & $0.655126$ & $0.000000$\\
parent, $\gamma=1$ & $0.804718$ & $0.451180$ & $0.451180$ & $0.451180$ & $0.000000$\\
generic mixed 1 & $1.245914$ & $0.561478$ & $0.561478$ & $0.575275$ & $0.013797$\\
generic mixed 2 & $1.145542$ & $0.536385$ & $0.536385$ & $0.548199$ & $0.011813$\\
generic mixed 3 & $1.198783$ & $0.549696$ & $0.549696$ & $0.565153$ & $0.015458$\\
\end{tabular}
\end{ruledtabular}
\caption{Mixed-state readout at $d=4$. The first three rows use the common parent in Eq.~\eqref{eq:table-common-parent}; the remaining rows use the three stored Ginibre states. The discrimination and robustness columns come from independent SDPs. The last column is $[1+C]/4-P_{\rm val}$. Display rounding is six decimals; feasibility diagnostics and matrices are retained separately.}
\label{tab:sdp}
\end{table}

\subsection{Coherent mesh errors and profile approximations}
\label{app:numreck}
We use an adjacent-mode triangular decomposition of Reck
type~\cite{Reck:1994ujl}. Each two-mode block, acting on column
amplitudes, is
\begin{equation}
B(t,\phi)=
\begin{pmatrix}
e^{i\phi}\cos t&-\sin t\\
e^{i\phi}\sin t&\cos t
\end{pmatrix}.
\label{eq:numerical-mesh-block}
\end{equation}
Here $t$ is an effective mixing angle and $\phi$ a phase setting.
For $m=d(d-1)/2$ blocks, the analyzer is
$U=D B_m\cdots B_1$, where
$D=\operatorname{diag}(e^{i\psi_z})$ is the output phase layer.
The code eliminates adjacent rows of $(V^\dagger)^\dagger$ one
column at a time to determine the block settings. This specifies
both the ordering and the phase convention. The maximum Frobenius
reconstruction error is $\MeshMaxReconstruction$.

The $d^2$ real parameters are
$q=(t_1,\phi_1,\ldots,t_m,\phi_m,\psi_0,\ldots,\psi_{d-1})$.
We perturb them by independent normal variables of standard deviation
$\sigma$, in radians. Thus the noise model is defined in these
effective coordinates; converting to particular internal MZI arm
phases requires the corresponding hardware calibration. The target
Fourier phases are drawn uniformly on $[-\pi,\pi)$ with seed
$20261102+d$, separately for $d=8,16$, and then held fixed for all
states and scans at that dimension.

At the calibrated settings $q_0$, define the tangent generators
\begin{equation}
G_j=-i\left.\frac{\partial U}{\partial q_j}\right|_{q_0}U(q_0)^\dagger,
\qquad G=\sum_j z_jG_j,\qquad z_j\sim\mathcal N(0,1).
\label{eq:numerical-tangent-generators}
\end{equation}
We evaluate the derivatives analytically and check them by central
differences. The moments $Q,R,V$ are exact finite sums over the
$G_j$, since $\mathbb E[z_jz_k]=\delta_{jk}$; they do not require
sampling. Stationarity of the readout objective ensures that the
second derivative of the parameter-to-unitary map contributes no
additional quadratic term. Consequently, the generators above give
the full quadratic coefficient of the physical parameter model,
even though finite errors are applied to the product of blocks.

In Table~\ref{tab:profiles}, ``phase settings'' includes all $\phi_j$
and $\psi_z$, and ``mixing angles'' includes the $t_j$. For the
identity configuration all settings are zero. The nonzero shifts
are then $0,\pm1$; the identity profile also follows directly from
the local block derivatives. Purely diagonal output generators have
no readout effect at any order. Their power must therefore be
excluded when calibrating the off-diagonal scalar model.

\begin{table}[!htbp]
\centering
\begin{ruledtabular}
\begin{tabular}{ccccc}
$w$ & identity & all parameters & phase settings only & mixing angles only\\
\colrule
$0$ & $3.130$ & $1.528$ & $3.373$ & $0.342$\\
$1$ & $2.435$ & $1.084$ & $0.939$ & $1.178$\\
$2$ & $0.000$ & $0.866$ & $0.571$ & $1.056$\\
$3$ & $0.000$ & $0.851$ & $0.520$ & $1.063$\\
$4$ & $0.000$ & $0.871$ & $0.568$ & $1.065$\\
$5$ & $0.000$ & $0.851$ & $0.520$ & $1.063$\\
$6$ & $0.000$ & $0.866$ & $0.571$ & $1.056$\\
$7$ & $2.435$ & $1.084$ & $0.939$ & $1.178$\\
\end{tabular}
\end{ruledtabular}
\caption{Normalized profiles $V(w)/\overline V$ for the $d=8$ triangular mesh. The identity column uses zero gate and output phases; all other columns use the same analyzer as Tables~\ref{tab:reck} and \ref{tab:d8scan} and Fig.~\ref{fig:quadraticscaling}. Profiles are evaluated exactly from tangent generators, without Monte Carlo estimation.}
\label{tab:profiles}
\end{table}

Table~\ref{tab:reck} uses uniform weights, a point mass, the equal-tail
states at $\Phi=0.95$ and $0.50$, and the explicit skewed distribution
$p_k=2(k+1)/[d(d+1)]$, $k=0,\ldots,d-1$. For direct simulation we draw
$3000$ independent standard Gaussian vectors with seed $20261202+d$
and evaluate both $q_0+\sigma z$ and $q_0-\sigma z$ for each vector.
Their average is one independent antithetic sample, giving $6000$
unitary evaluations per configuration. Error bars are standard errors
over the $3000$ pair averages. This preserves the Gaussian expectation
while cancelling odd terms in each sample pair.

\begin{table}[!htbp]
\centering
\begin{ruledtabular}
\begin{tabular}{llccccc}
$d$ & state & direct mean $\pm$ SE & full & profile & off-diagonal scalar & all-shift iso.\\
\colrule
$8$ & $p$ uniform & $9.381\pm0.046$ & $9.304$ & $9.304$ & $9.304$ & $10.063$\\
$8$ & point mass & $0$ & $0.000$ & $0.000$ & $0.000$ & $0.000$\\
$8$ & $\Phi=0.95$ & $1.922\pm0.011$ & $1.907$ & $1.932$ & $1.932$ & $2.089$\\
$8$ & $\Phi=0.50$ & $7.526\pm0.038$ & $7.465$ & $7.504$ & $7.504$ & $8.116$\\
$8$ & skewed & $8.545\pm0.042$ & $8.477$ & $8.479$ & $8.487$ & $9.179$\\
$16$ & $p$ uniform & $20.513\pm0.055$ & $20.421$ & $20.421$ & $20.421$ & $22.031$\\
$16$ & point mass & $0$ & $0.000$ & $0.000$ & $0.000$ & $0.000$\\
$16$ & $\Phi=0.95$ & $3.300\pm0.010$ & $3.284$ & $3.251$ & $3.251$ & $3.508$\\
$16$ & $\Phi=0.50$ & $14.931\pm0.041$ & $14.862$ & $14.803$ & $14.803$ & $15.970$\\
$16$ & skewed & $18.390\pm0.050$ & $18.319$ & $18.402$ & $18.434$ & $19.888$\\
\end{tabular}
\end{ruledtabular}
\caption{Coherent-error losses in units of $10^{-6}$ at $\sigma=10^{-3}$. Direct values use $3000$ independent antithetic Gaussian pairs; SE is the standard error over pair averages. All state rows at a given dimension share the same Gaussian vectors and their sign reversals. Full and reduced columns are ensemble quadratic predictions, not estimates on the same finite sample. The profile and off-diagonal scalar agree on upper-boundary states; both skewed rows are evaluated using the explicit weights in the text. The point-mass row is zero analytically, with direct numerical residual below $10^{-15}$.}
\label{tab:reck}
\end{table}

Within each dimension, all state rows in Table~\ref{tab:reck} use
\emph{the same} $3000$ independent standard Gaussian vectors and
their sign reversals. Direct and finite-sample quadratic estimates
are evaluated on these same perturbations. The displayed full and
reduced columns are ensemble coefficients evaluated from tangent
moments, not estimates from those draws. The scale scan uses a
separate set of $25000$ independent vectors, shared across its error
scales. These shared draws correlate the corresponding Monte Carlo
estimates; they do not increase the number of independent samples.

The direct means lie within $\MeshMaxSE$ standard errors of the
ensemble quadratic predictions. To isolate truncation error from
Monte Carlo fluctuation, we also evaluate the full quadratic form on
the same Gaussian vectors. The largest relative difference between
these paired quadratic means and the direct means is
$\MeshMaxPairedRelative$ at $\sigma=10^{-3}$. The associated
standard errors and unrounded outputs are stored in
\texttt{numerics/results.json}.

For these examples, the reduced profile differs from the exact
quadratic coefficient by at most $\MeshMaxProfileDeviation\%$.
On the upper boundary the off-diagonal scalar and profile agree
algebraically, irrespective of profile flatness. On the skewed
states both are evaluated separately. The all-shift isotropic
normalization overestimates the full coefficients by
$\MeshIsoMinBias\%$--$\MeshIsoMaxBias\%$ in Table~\ref{tab:reck}.
These are properties of the specified configurations and parameter
model, not bounds for arbitrary interferometers.

\begin{table}[!htbp]
\centering
\begin{ruledtabular}
\begin{tabular}{ccccc}
$\Phi$ & $C$ & full coefficient & profile/scalar coefficient & relative deviation\\
\colrule
$1.000$ & $0.000000$ & $0.000000$ & $0.000000$ & ---\\
$0.950$ & $1.453256$ & $1.906955$ & $1.931512$ & $1.288\%$\\
$0.800$ & $3.316601$ & $4.368046$ & $4.408069$ & $0.916\%$\\
$0.600$ & $4.992296$ & $6.593434$ & $6.635222$ & $0.634\%$\\
$0.500$ & $5.645751$ & $7.465353$ & $7.503724$ & $0.514\%$\\
$0.400$ & $6.192296$ & $8.197671$ & $8.230133$ & $0.396\%$\\
$0.200$ & $6.916601$ & $9.180756$ & $9.192801$ & $0.131\%$\\
$0.125$ & $7.000000$ & $9.303646$ & $9.303646$ & $0.000\%$\\
\end{tabular}
\end{ruledtabular}
\caption{Boundary scan for the same $d=8$ mesh as Table~\ref{tab:reck}. Entries are coefficients of $\sigma^2$. The $\Phi=0.50$ row is also the quadratic prediction in Fig.~\ref{fig:quadraticscaling}. The undefined relative deviation at zero loss is left blank; no numerical quantity is missing.}
\label{tab:d8scan}
\end{table}

The scan in Fig.~\ref{fig:quadraticscaling} uses the same $d=8$
mesh and the $\Phi=0.50$ state. We draw $25000$ independent Gaussian
vectors with seed $20261402$, use antithetic pairs, and rescale the
same vectors at every error strength. There are therefore $50000$
unitary evaluations per point. The ensemble quadratic coefficient is
$c=\MeshQuadraticC$, evaluated independently from the tangent
generators. The measured loss follows $c\sigma^2$ at small $\sigma$,
with deviations at larger errors. Sharing the draws across scales
correlates the estimates; it does not create additional independent
samples.

\subsection{Strictness of the mode-dependent-loss bound}
\label{app:nummdl}
Using seed $20261602$, we generate $180$ loss instances in dimensions
$d=2,\ldots,6$. At each dimension, the weights are the equal-tail
families at $\Phi=0.9,0.5$ and one symmetric Dirichlet draw. The three
transmission profiles are uniform $0.7$, linear from $0.02$ to $1$,
and one independent uniform draw in $[0.02,1]^d$. For each combination
we use constant response phases and three independent uniform phase
assignments in $[0,2\pi)^d$. The code stores every state and profile,
so the instances are specified beyond the seed alone.

The universal certificate is compared with independently solved
photonic primal and dual problems. For the $90$ random-phase,
nonuniform-profile cases, define
\begin{equation}
\delta_{\mathrm{mdl}}=
\frac{\bar\tau(1+C(\eta))-P_{\mathrm{ph}}}
{\bar\tau(1+C(\eta))}.
\label{eq:mdlrelativegap}
\end{equation}
Its numerical values range from $\LossGapMinPercent\%$ to
$\LossGapMaxPercent\%$. Replacing $P_{\mathrm{ph}}$ by the repaired
dual objective gives a lower estimate of the certificate gap; this
remains at least $\LossCertifiedMinPercent\%$ in every one of these
$90$ cases. The strictness therefore exceeds the measured numerical
uncertainty. The endpoints describe this sample and are not universal
limits. Constant-phase and uniform-loss instances agree with the
analytic saturation conditions within their numerical enclosures.

\begin{table}[!htbp]
\centering
\begin{ruledtabular}
\begin{tabular}{lcccc}
instances & $\max r_{\rm comp}$ & $\max e_{\rm pr}$ & $\max e_{\rm du}$ & $\max(P_U-P_L)$\\
\colrule
Table I & $1.54\times10^{-13}$ & $1.14\times10^{-9}$ & $6.52\times10^{-10}$ & $9.92\times10^{-9}$\\
Loss study (180 instances) & $1.86\times10^{-8}$ & $2.95\times10^{-7}$ & $2.03\times10^{-8}$ & $2.07\times10^{-7}$\\
\end{tabular}
\end{ruledtabular}
\caption{Numerical discrimination diagnostics in the original complex matrices. The first three columns precede feasibility repair; the last column uses repaired primal and dual matrices. All values are computed from stored arrays. These are floating-point certificates, not interval-arithmetic proofs.}
\label{tab:diagnostics}
\end{table}

\subsection{Data and code availability}
\label{app:repro}
\HOneCodeLinkText
The code and data supporting the numerical results of this article are
openly available~\cite{code}. The repository contains the simulation and
SDP code, pinned package versions, seeds, complete mesh settings, state
matrices, primal and dual solutions, feasibility diagnostics, and
unrounded numerical outputs. Running \texttt{python reproduce\_all.py}
reruns the discrimination and robustness SDPs, the mesh simulations, and
the mode-dependent-loss study, validates the new outputs, and regenerates
Tables~I--V and Figs.~2, 4, and~5. The option \texttt{-{}-from-saved}
validates the stored certificates and regenerates the tables and figures
without rerunning the solvers or the Monte Carlo simulations. The README
lists the individual commands and identifies the files used by the
manuscript.

% \FloatBarrier
% \clearpage

\end{document}